\documentclass[preprint,12pt,authoryear]{elsarticle}

  \usepackage[margin=1in]{geometry}
\usepackage{amssymb,amsmath}
\usepackage{txfonts}
 \usepackage{url}
\usepackage[table]{xcolor}
 \usepackage{setspace} 

 \usepackage{float}
\usepackage[hidelinks]{hyperref}

 \usepackage{multicol}
 \usepackage{multirow}
\usepackage{graphicx}
\usepackage{subcaption}
\usepackage{booktabs,longtable,pdflscape,array,tabularx}

\usepackage[plain,noend]{algorithm2e}
\usepackage{algpseudocode} 
\usepackage{bm}

\graphicspath{{./art/}{./figures/}}
\def\l{\langle}

\journal{Environmetrics}
\newtheorem{prop}{Proposition}

\newtheorem{cor}{Corollary}
 \newtheorem{thm}{Theorem}
 \newtheorem{definition}{Definition}
 \newtheorem{lem}[thm]{Lemma}
 \newdefinition{rem}{Remark}
 \newproof{proof}{Proof}

\begin{document}

\begin{frontmatter}





\title{Multidimensional Integral Fractional Ornstein--Uhlenbeck Process with an Application to Animal Movement}
\author[label1]{J.H. Ram\'{\i}rez-Gonz\'{a}lez}
\ead{josehermenegildo.ramirezgonzalez@kaust.edu.sa}
\author[label1]{Erick A. Chac\'{o}n-Montalv\'{a}n}
\ead{erick.chaconmontalvan@kaust.edu.sa}
\author[label1]{Paula Moraga}
\ead{paula.moraga@kaust.edu.sa}
\author[label1]{Ying Sun}
\ead{ying.sun@kaust.edu.sa}

\affiliation[label1]{organization={CEMSE Division, Statistics Program, King Abdullah University of Science and Technology},
            city={Thuwal},
            postcode={23955-6900},
            state={Makkah},
            country={Saudi Arabia}}





\begin{abstract}
\setlength{\emergencystretch}{1em}
Fractional Ornstein--Uhlenbeck (fOU) processes model temporal dependence and memory, including long-range dependence, while retaining the classical Ornstein--Uhlenbeck process as a special case. We extend the integral fractional Ornstein--Uhlenbeck (ifOU) process to a multidimensional setting for animal telemetry. Longitude, Latitude, and Altitude velocities are represented by coordinate-specific fOU processes driven by a multivariate fractional Brownian motion, allowing each coordinate to retain its own damping, scale, and Hurst parameters. We establish covariance validity, characterize the admissible cross-correlation region, and derive the asymptotic behavior of cross-covariances and separated increments. We develop procedures for finite-dimensional simulation, Gaussian likelihood inference, and conditional velocity reconstruction. Replicated simulations examine estimation of cross-coordinate correlations, while joint estimation of the complete parameter vector is illustrated for one three-dimensional trajectory. The proposed model is applied to telemetry records from five common noctule bats migrating in Germany, including three trajectories with Altitude measurements.
\end{abstract}


\end{frontmatter}
\enlargethispage{-0.73in}

\doublespacing

\section{Introduction}

Animal telemetry provides information on how individuals use space and how movement changes over time. Its statistical analysis is complicated by the fact that locations are recorded at discrete and often irregular times, whereas movement evolves continuously. Continuous-time velocity models address this mismatch by representing position as the temporal integral of velocity. This formulation separates movement dynamics from the observation schedule and permits likelihood-based inference directly from recorded locations \citep{John,gurarie,Hooten,Sarkka}.

A foundational model in this class is the correlated velocity model of \citet{John}, in which velocity follows an Ornstein--Uhlenbeck (OU) process and position is obtained by integration. The model is Gaussian and computationally tractable, but the covariance of the OU velocity decays exponentially, resulting in short-range temporal dependence. This restriction may be limiting when movement decisions depend on previously visited locations, as represented explicitly in memory-based movement models \citep{Smouse2010}.

Fractional Ornstein--Uhlenbeck processes provide a natural extension because the fractional Brownian driving noise introduces a Hurst parameter that allows a broader range of temporal dependence \citep{fBM1}. Building on the integrated OU formulation, \citet{JY} introduced the integral fractional Ornstein--Uhlenbeck (ifOU) process for animal telemetry and developed inference from the finite-dimensional distribution of the position process. Their application to fin whale trajectories in the Gulf of California selected an ifOU Latitude model with $H>1/2$, corresponding to the long-memory regime under the covariance-summability criterion, rather than the classical case $H=1/2$. This result indicated that the additional temporal flexibility of the ifOU model can be supported by telemetry data. The analysis, however, treated Longitude and Latitude independently, leaving dependence between spatial coordinates unmodeled.

Continuing the study of models with long-range dependence for animal movement, \citet{RamirezEtAl2026} introduced a family of Gaussian processes that are neither stationary nor intrinsically stationary and whose covariance with remote future positions grows logarithmically. The models were applied separately to the Longitude and Latitude records of five common noctule bats migrating in Germany. Across the ten coordinate series, the empirical diagnostics did not uniformly support either weak stationarity of the recorded positions or stationarity of their first differences, while detrended fluctuation analysis indicated dependence over large temporal scales. Scaled fractional Brownian motion attained the smallest AIC in five of the ten coordinatewise comparisons, the nonstationary logarithmic-covariance model in three, and the ifOU and stationary Confluent models in one comparison each. Although Altitude was available for three individuals, the analysis was restricted to Longitude and Latitude and assumed independence between coordinates.

These results motivate a joint model that represents temporal dependence and dependence between coordinates simultaneously. A multidimensional extension cannot, in general, be obtained by assigning arbitrary pairwise correlations to separate ifOU processes. When the coordinates have different Hurst parameters, the correlations must satisfy joint restrictions to ensure that the covariance remains valid for every collection of observation times. A statistically usable extension therefore requires both a valid joint covariance construction and a parametrization that enforces these restrictions during likelihood optimization.

In this article, we introduce a multidimensional integral fractional Ornstein--Uhlenbeck (mifOU) process driven by a well-balanced multivariate fractional Brownian motion. Here, well-balanced means that each cross-covariance is unchanged when its two time arguments are interchanged. Each coordinate has its own damping, scale, and Hurst parameters, while dependence between coordinates is introduced through the multivariate driving process. We prove that the resulting joint covariance is valid, characterize the admissible correlation region, and show that the mifOU process is neither stationary nor intrinsically stationary. For correlated coordinate pairs, the process exhibits long-range dependence whenever the sum of the corresponding Hurst parameters differs from one. We use an LKJ-type partial-correlation parametrization that maps unconstrained optimization variables directly into the admissible correlation region, so every likelihood evaluation uses a valid covariance without a separate positive-definiteness check \citep{LJK}. The finite-dimensional Gaussian distribution then provides the basis for maximum likelihood inference, finite-dimensional simulation, and conditional velocity reconstruction from position observations.

The methodology is applied to the five bat trajectories studied by \citet{RamirezEtAl2026}. We retain the coordinatewise analyses of Longitude and Latitude, incorporate Altitude for the three individuals for which it is observed, and include empirical summaries of movement direction and association between coordinates. These summaries are not interpreted as estimates of the model correlation parameters, but provide a separate description of the observed movement. The application combines complete-model, horizontal-only, and matched-family comparisons so that differences in marginal covariance fit can be distinguished from evidence for cross-coordinate dependence.

The remainder of the article is organized as follows. Section~\ref{SSec2} defines the mifOU process, establishes its covariance and dependence properties, develops the admissible-correlation parametrization, and presents the procedures for likelihood inference, finite-dimensional simulation, and conditional velocity reconstruction. Section~\ref{SSec3} describes the bat data, the empirical diagnostics, the marginal and complete-trajectory model comparisons, and the principal findings. Section~\ref{SSec4} concludes the article. Appendix~\ref{ApenA} contains the proofs of Propositions~1--4, Appendix~\ref{ApenB} proves Theorem~2, Sections~\ref{sub:joint-estimation-supp}--\ref{sub:computational-scaling-supp} of Appendix~\ref{sec2.4} report the numerical experiments and computational considerations, and Sections~\ref{sub:altitude-diagnostics}--\ref{sub:complete-trajectory-supp} of Appendix~\ref{ApenD} provide the supporting empirical analyses.

\section{The multidimensional integral fractional Ornstein--Uhlenbeck process}\label{SSec2}


\subsection{Model definition and covariance construction}\label{sub:mifou_definition}

Let
\[
\mu_H(t)=\big[\mu_{H_1}(t),\ldots,\mu_{H_p}(t)\big]^\intercal
\quad\text{and}\quad
v_H(t)=\big[v_{H_1}(t),\ldots,v_{H_p}(t)\big]^\intercal
\]
denote the position and velocity of an animal at time \(t\), respectively, where \(p\) is the number of observed coordinates. For each coordinate, position is obtained by integrating velocity:
\begin{equation}\label{eq:position}
\mu_{H_i}(t)
=
\mu_{H_i}(0)+\int_0^t v_{H_i}(s)\,ds,
\qquad i=1,\ldots,p.
\end{equation}

Following \citet{JY}, we model the velocity in coordinate \(i\) by a fractional Ornstein--Uhlenbeck (fOU) process:
\begin{equation}\label{eq:velocity}
dv_{H_i}(t)
=
-\beta_i v_{H_i}(t)\,dt+\sigma_i\,dW_{H_i}(t),
\end{equation}
where \(\beta_i>0\), \(\sigma_i>0\), and \(W_{H_i}\) is a fractional Brownian motion with Hurst parameter \(H_i\). Previous coordinatewise ifOU analyses treated the driving processes independently. Here, the components of
\[
W_H=\big(W_{H_1},\ldots,W_{H_p}\big)^\intercal
\]
may be correlated, subject to the validity of their joint covariance.

By \citet[Proposition~2.1]{JY},
\begin{equation}\label{muH2}
\mu_{H_i}(t)
=
\mu_{H_i}(0)
+
v_{H_i}(0)\frac{1-e^{-\beta_i t}}{\beta_i}
+
\sigma_i\int_0^t e^{\beta_i(u-t)}W_{H_i}(u)\,du,
\qquad i=1,\ldots,p.
\end{equation}
Under deterministic initial conditions, \eqref{muH2} gives
\[
\operatorname{cov}\big\{\mu_{H_i}(s),\mu_{H_j}(t)\big\}
=
\sigma_i\sigma_j
\int_0^s\!\int_0^t
e^{\beta_i(u-s)}
\mathbb{E}\big[W_{H_i}(u)W_{H_j}(v)\big]
e^{\beta_j(v-t)}
\,dv\,du,
\qquad i,j=1,\ldots,p.
\]
The following proposition establishes that this construction yields a valid matrix-valued covariance kernel. Its proof is given under ``Proof of Proposition~1'' in Appendix~\ref{ApenA}.

\begin{prop}\label{Propos_mucov}
Under the assumptions of \citet[Theorem~2.1]{Lavancier}, let
\begin{equation}\label{K_kernel}
K_{i,j}(u,v)
:=
\mathbb{E}\big[W_{H_i}(u)W_{H_j}(v)\big],
\qquad
u,v\geq0,\quad i,j=1,\ldots,p,
\end{equation}
so that \(K=(K_{i,j})_{i,j=1}^p\) is a matrix-valued covariance kernel on
\(\mathbb{R}_+\times\mathbb{R}_+\). For \(s,t\geq0\) and \(i,j=1,\ldots,p\), define
\begin{equation}\label{cfR}
R_{i,j}(s,t)
:=
\sigma_i\sigma_j
\int_0^s\!\int_0^t
e^{\beta_i(u-s)}
K_{i,j}(u,v)
e^{\beta_j(v-t)}
\,dv\,du,
\end{equation}
where \(\sigma_i,\beta_i>0\) for \(i=1,\ldots,p\). Then
\(R=(R_{i,j})_{i,j=1}^p\) is a matrix-valued covariance kernel on
\(\mathbb{R}_+\times\mathbb{R}_+\).
\end{prop}

\begin{definition}\label{defmu3D}
For deterministic initial positions and velocities, the multidimensional integral fractional Ornstein--Uhlenbeck (mifOU) position process is the \(p\)-variate Gaussian process
\[
\mu_H=(\mu_{H_1},\ldots,\mu_{H_p})
\]
with mean functions
\[
m_i(t)
=
\mu_{H_i}(0)
+
v_{H_i}(0)\frac{1-e^{-\beta_i t}}{\beta_i},
\qquad i=1,\ldots,p,
\]
and cross-covariances
\[
\operatorname{cov}\big\{\mu_{H_i}(s),\mu_{H_j}(t)\big\}
=
R_{i,j}(s,t),
\qquad
s,t\geq0,\quad i,j=1,\ldots,p,
\]
where \(R_{i,j}\) is defined in \eqref{cfR}.
\end{definition}

The dependence between the position coordinates is determined by the cross-covariances of the driving multivariate fractional Brownian motion through the integral transformation in \eqref{cfR}. In particular, when the components \((W_{H_i})_{i=1}^p\) are independent, the position coordinates \((\mu_{H_i})_{i=1}^p\) are independent.

Although Proposition~\ref{Propos_mucov} applies to a general multivariate fractional Brownian motion, verifying the conditions of \citet[Theorem~2.1]{Lavancier} can be difficult during likelihood evaluation. We therefore adopt the \emph{well-balanced} specification
\[
\mathbb{E}\big[W_{H_i}(s)W_{H_j}(t)\big]
=
\mathbb{E}\big[W_{H_i}(t)W_{H_j}(s)\big],
\qquad s,t\in\mathbb{R},
\]
following \citet{Lavancier,Pierre}. Thus, ``well-balanced'' refers here to symmetry of each cross-covariance in its two time arguments. Under this specification, the cross-covariance has the explicit form
\begin{equation}\label{k}
K_{i,j}(s,t)
=
\frac{\rho_{i,j}}{2}
\left(
|s|^{H_i+H_j}
+
|t|^{H_i+H_j}
-
|t-s|^{H_i+H_j}
\right),
\end{equation}
where \(\rho_{i,i}=1\) and \(\rho_{i,j}=\rho_{j,i}\). In particular, \(K_{i,i}\) is the usual fractional Brownian motion covariance and
\(K_{i,j}(1,1)=\rho_{i,j}\). The validity of the joint covariance is then characterized by a finite-dimensional condition.

\begin{lem}\label{lem2.2}
Let
\[
\Gamma_{i,j}
:=
\Gamma(H_i+H_j+1)
\sin\!\left(\frac{\pi}{2}(H_i+H_j)\right).
\]
Then the matrix-valued kernel \(K=(K_{i,j})_{i,j=1}^p\) defined in \eqref{k} is a possibly degenerate covariance kernel if and only if the \(p\times p\) matrix
\[
Q_{H_1,\ldots,H_p}(\rho)_{ij}
=
\begin{cases}
\Gamma_{i,i}, & i=j,\\[2pt]
\rho_{i,j}\Gamma_{i,j}, & i\neq j,
\end{cases}
\]
is positive semidefinite. In the likelihood analysis, we restrict attention to the nondegenerate interior, where
\[
Q_{H_1,\ldots,H_p}(\rho)\succ0.
\]
\end{lem}

Lemma~\ref{lem2.2} provides the finite-dimensional restriction required for a nondegenerate well-balanced multivariate fractional Brownian motion covariance \citep{Lavancier,Pierre}. In the three-dimensional case, \(p=3\), the condition
\(Q_{H_1,H_2,H_3}(\rho)\succ0\) is equivalent to the pairwise inequalities
\begin{equation}\label{assumption}
\rho_{i,j}^2
<
\frac{\Gamma_{i,i}\Gamma_{j,j}}{\Gamma_{i,j}^2},
\qquad
1\leq i<j\leq3,
\end{equation}
together with
\[
\det Q_{H_1,H_2,H_3}
(\rho_{1,2},\rho_{1,3},\rho_{2,3})>0.
\]
The determinant condition is equivalent to
\begin{equation}\label{ellipse}
\begin{aligned}
&
\Gamma_{1,1}\Gamma_{2,3}^2\rho_{2,3}^2
+
\Gamma_{2,2}\Gamma_{1,3}^2\rho_{1,3}^2
+
\Gamma_{3,3}\Gamma_{1,2}^2\rho_{1,2}^2
\\
&\qquad
-
2\Gamma_{1,2}\Gamma_{1,3}\Gamma_{2,3}
\rho_{1,2}\rho_{1,3}\rho_{2,3}
<
\Gamma_{1,1}\Gamma_{2,2}\Gamma_{3,3}.
\end{aligned}
\end{equation}

\begin{rem}\label{re:conic}
Fix \(H_1,H_2,H_3\) and one of the three correlation parameters. Under \eqref{assumption}, the boundary defined by equality in \eqref{ellipse} is an ellipse in the two remaining correlation parameters. The admissible set is
\[
\mathcal{R}_{H_1,H_2,H_3}
=
\left\{
\boldsymbol{\rho}\in(-1,1)^3:
Q_{H_1,H_2,H_3}(\boldsymbol{\rho})\succ0
\right\}.
\]
\end{rem}

Consequently, likelihood calculations do not require direct verification of the operator-valued conditions in \citet{Lavancier}. It is sufficient to impose
\[
Q_{H_1,\ldots,H_p}(\rho)\succ0.
\]
All covariance calculations, likelihood fits, simulations, and empirical analyses below use the well-balanced kernel in equation~\eqref{k}.

\subsection{Covariance and dependence properties}\label{sub:mifou_properties}

This subsection studies long-range dependence and short- and long-memory behavior for separated increments. The proofs are given in Appendix~\ref{ApenA}.

We use the dependence terminology adopted in \citet{RamirezEtAl2026}. For fixed \(0\leq r<\nu\) and \(0\leq s<t\), the coordinate pair \((i,j)\) is said to exhibit \emph{long-range dependence} (LRD) with index \(\kappa\in\mathbb{R}\) if
\begin{equation}\label{eq:LRD_definition}
\lim_{T\to\infty}T^\kappa\operatorname{cov}\!\left(\mu_{H_i}(\nu)-\mu_{H_i}(r),\mu_{H_j}(T+t)-\mu_{H_j}(T+s)\right)=C_{i,j}^{\,r,\nu;s,t},
\qquad 0<\left|C_{i,j}^{\,r,\nu;s,t}\right|.
\end{equation}
For the integer-lag covariance sequence
\begin{equation}\label{eq:gamma_definition}
\gamma_{i,j}^{\,r,\nu;s,t}(n):=\operatorname{cov}\!\left(\mu_{H_i}(\nu)-\mu_{H_i}(r),\mu_{H_j}(n+t)-\mu_{H_j}(n+s)\right),\qquad n\in\mathbb{N},
\end{equation}
the pair is said to have \emph{short memory} if \(\sum_{n\geq1}|\gamma_{i,j}^{\,r,\nu;s,t}(n)|<\infty\), and \emph{long memory} if \(\sum_{n\geq1}|\gamma_{i,j}^{\,r,\nu;s,t}(n)|=\infty\). These conventions are used throughout the article. Thus, LRD and long memory refer to distinct criteria here: the former is defined through a nonzero polynomial scaling limit, whereas the latter is defined through covariance summability.

\begin{prop}\label{Propos1}
Let \(i,j\in\{1,\ldots,p\}\) and \(s,t\geq0\). If the cross-covariance of the associated mfBm is well-balanced, that is,
\[
K_{i,j}(s,t)=K_{i,j}(t,s),\qquad s,t\geq0,
\]
and \(\beta_i=\beta_j\), then
\[
R_{i,j}(s,t)=R_{i,j}(t,s),\qquad s,t\geq0.
\]
\end{prop}

The application does not impose \(\beta_i=\beta_j\). Consequently, \(R_{i,j}(s,t)\) need not be symmetric in its two time arguments, although the matrix-valued covariance satisfies \(R_{i,j}(s,t)=R_{j,i}(t,s)\).

\begin{prop}\label{Propos2}
Let \(s,t\geq0\) and write \(h_{ij}=H_i+H_j\). If \(h_{ij}\leq1\), then
\begin{equation}
R_{i,j}(s,t+T)\underset{T\to\infty}{\longrightarrow}
\begin{cases}
\displaystyle \frac{\sigma_i\sigma_j\rho_{i,j}}{2\beta_j}\int_0^s e^{-\beta_i u}(s-u)^{h_{ij}}\,du, & h_{ij}<1,\\[10pt]
\displaystyle \frac{\sigma_i\sigma_j\rho_{i,j}}{\beta_j}\int_0^s e^{-\beta_i u}(s-u)\,du, & h_{ij}=1.
\end{cases}
\end{equation}
If \(h_{ij}>1\), then
\begin{equation}
\frac{R_{i,j}(s,t+T)}{T^{h_{ij}-1}}\underset{T\to\infty}{\longrightarrow}
\frac{\sigma_i\sigma_j\rho_{i,j}}{2}\frac{h_{ij}}{\beta_j}\int_0^s e^{-\beta_i u}(s-u)\,du.
\end{equation}
\end{prop}

Proposition~\ref{Propos2} describes the large-lag behavior of the cross-covariance. When \(h_{ij}\leq1\), \(R_{i,j}(s,t+T)\) converges to a finite limit. When \(h_{ij}>1\) and \(\rho_{i,j}\neq0\),
\[
R_{i,j}(s,t+T)\sim\left(\frac{\sigma_i\sigma_j\rho_{i,j}}{2}\frac{h_{ij}}{\beta_j}\int_0^s e^{-\beta_i u}(s-u)\,du\right)T^{h_{ij}-1},\qquad T\to\infty,
\]
so its magnitude has exact order \(T^{h_{ij}-1}\), with sign determined by \(\rho_{i,j}\). If \(i\neq j\) and \(\rho_{i,j}=0\), then \(K_{i,j}\) and \(R_{i,j}\) are identically zero.

Figure~\ref{mld} illustrates Proposition~\ref{Propos2}. The top row shows the three marginal covariance sections. In the bottom row, \(H_1+H_2=1.25\) produces polynomial growth of \(R_{1,2}\), \(H_1+H_3=1\) produces a finite negative limit for \(R_{1,3}\), and \(H_2+H_3=0.75\) produces a finite positive limit for \(R_{2,3}\). The signs are determined by the corresponding values of \(\rho_{i,j}\).

\begin{figure}[H]
\centering
\includegraphics[width=0.6\textwidth]{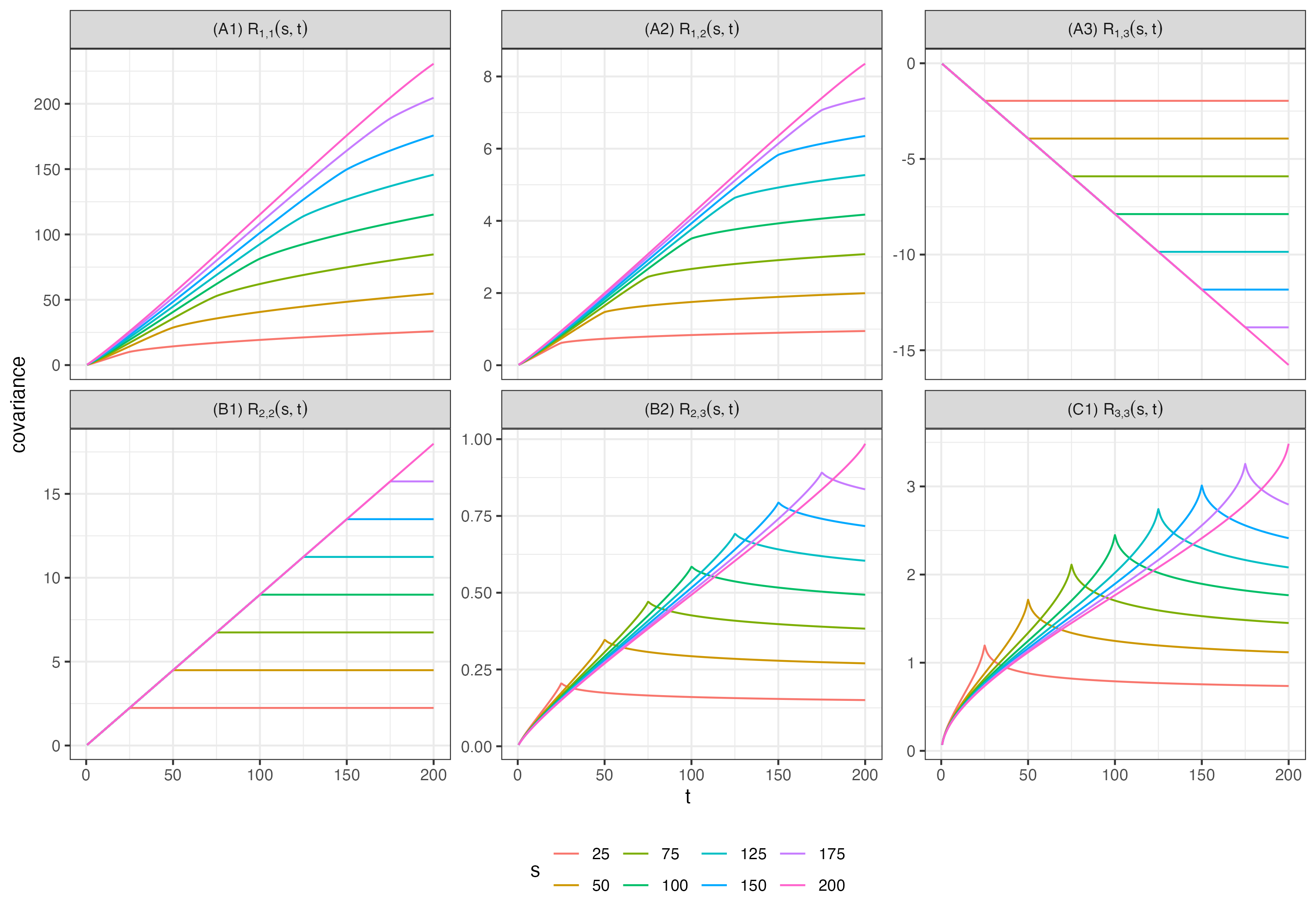}
\caption{Covariance sections for \(\boldsymbol{\sigma}=(2,3,2.5)\), \(\boldsymbol{\beta}=(7,10,5)\), \(\boldsymbol{H}=(0.75,0.50,0.25)\), and \((\rho_{12},\rho_{13},\rho_{23})=(0.185,-0.384,0.461)\). The top row shows \(R_{11}\), \(R_{22}\), and \(R_{33}\); the bottom row shows \(R_{12}\), \(R_{13}\), and \(R_{23}\) for the displayed fixed values of \(s\).}
\label{mld}
\end{figure}

\begin{prop}\label{Prop2.9}
Let \(0\leq r<\nu\), \(0\leq s<t\), and write \(h_{ij}=H_i+H_j\). Then
\begin{equation}
\begin{split}
&T^{2-h_{ij}}\operatorname{cov}\!\left(\mu_{H_i}(\nu)-\mu_{H_i}(r),\mu_{H_j}(T+t)-\mu_{H_j}(T+s)\right)\\
&\qquad\underset{T\to\infty}{\longrightarrow}
\frac{\sigma_i\sigma_j\rho_{i,j}h_{ij}(h_{ij}-1)}{2}\frac{t-s}{\beta_j}
\left[\frac{\nu-r}{\beta_i}+\frac{e^{-\beta_i\nu}-e^{-\beta_i r}}{\beta_i^2}\right].
\end{split}
\end{equation}
If \(h_{ij}=1\), the displayed limit is zero and
\[
\operatorname{cov}\!\left(\mu_{H_i}(\nu)-\mu_{H_i}(r),\mu_{H_j}(T+t)-\mu_{H_j}(T+s)\right)
=O\!\left((1+T)e^{-\beta_jT/2}\right).
\]
\end{prop}

For \(h_{ij}\neq1\), define
\[
C_{i,j}^{\,r,\nu;s,t}:=
\frac{\sigma_i\sigma_j\rho_{i,j}h_{ij}(h_{ij}-1)}{2}\frac{t-s}{\beta_j}
\left[\frac{\nu-r}{\beta_i}+\frac{e^{-\beta_i\nu}-e^{-\beta_i r}}{\beta_i^2}\right].
\]
Since
\[
\frac{\nu-r}{\beta_i}+\frac{e^{-\beta_i\nu}-e^{-\beta_i r}}{\beta_i^2}
=\frac{1}{\beta_i}\int_r^\nu\left(1-e^{-\beta_i u}\right)\,du>0
\]
for \(0\leq r<\nu\), one has \(C_{i,j}^{\,r,\nu;s,t}\neq0\) if and only if \(\rho_{i,j}\neq0\). Therefore, whenever \(\rho_{i,j}\neq0\) and \(h_{ij}\neq1\),
\[
\gamma_{i,j}^{\,r,\nu;s,t}(n)\sim C_{i,j}^{\,r,\nu;s,t}n^{h_{ij}-2},\qquad n\to\infty,
\]
and the pair exhibits LRD with index \(\kappa_{i,j}=2-h_{ij}=2-(H_i+H_j)\).

If \(h_{ij}<1\), then \(h_{ij}-2<-1\), so \(\sum_{n\geq1}|\gamma_{i,j}^{\,r,\nu;s,t}(n)|<\infty\), and the pair has short memory. If \(h_{ij}>1\), then \(-1<h_{ij}-2<0\), so \(\sum_{n\geq1}|\gamma_{i,j}^{\,r,\nu;s,t}(n)|=\infty\), and the pair has long memory. If \(h_{ij}=1\), the covariance of separated increments decays exponentially; hence the pair has short memory and does not satisfy the LRD criterion adopted here. Thus, for \(\rho_{i,j}\neq0\), the pair exhibits LRD whenever \(H_i+H_j\neq1\), has short memory when \(H_i+H_j\leq1\), and has long memory when \(H_i+H_j>1\).

For an individual coordinate, \(\rho_{i,i}=1\) and \(h_{ii}=2H_i\). Hence the marginal process exhibits LRD when \(H_i\neq1/2\), has short memory when \(H_i\leq1/2\), and has long memory when \(H_i>1/2\). In the Brownian case \(H_i=1/2\), the covariance of separated increments decays exponentially and the marginal process does not satisfy the adopted LRD criterion.

\subsection{Likelihood inference and admissible correlation parametrization}\label{sub:likelihood}

Several asymptotic results are available for estimators of fractional Ornstein--Uhlenbeck parameters; see, for example, \cite{Alexandre,Nualart,Tanaka,Xiao}. However, as noted in \citet{JY}, these results require knowledge of the velocity trajectory or of a specific sample from it. In practice, only sparsely sampled position trajectories are commonly available. We therefore estimate the model parameters from the finite-dimensional Gaussian likelihood of the position process.

Bound-constrained optimization methods, such as L-BFGS-B and
truncated Newton methods, impose separate lower and upper bounds on
each parameter and therefore search over a hyperrectangle. For
\(p\geq2\), however, the admissible region of the correlation
parameters depends on \(H_1,\ldots,H_p\) and cannot in general be
specified through independent bounds on the individual correlations.
Such methods may therefore evaluate parameter values for which
\(Q_{H_1,\ldots,H_p}(\rho)\) is not positive definite. The
reparametrization introduced below allows optimization over a fixed
hypercube while ensuring that every transformed value defines a valid
covariance, without additional admissibility checks.

\begin{thm}\label{TeoT}
Let \(\mathcal{R}_{H_1,H_2,H_3}\) denote the positive-definite correlation parameter space for fixed Hurst parameters. Then the mapping \(T_{H_1,H_2,H_3}:(-1,1)^3\rightarrow\mathcal{R}_{H_1,H_2,H_3}\), defined by
\begin{equation}\label{T}
\begin{aligned}
&T_{H_1,H_2,H_3}(z_{1,2},z_{1,3},z_{2,3})\\
&\quad=\Bigg(
\frac{z_{1,2}\sqrt{\Gamma_{1,1}\Gamma_{2,2}}}{\Gamma_{1,2}},
\frac{z_{1,3}\sqrt{\Gamma_{1,1}\Gamma_{3,3}}}{\Gamma_{1,3}},
\frac{\left[z_{2,3}\sqrt{1-z_{1,2}^{2}}\sqrt{1-z_{1,3}^{2}}+z_{1,2}z_{1,3}\right]\sqrt{\Gamma_{2,2}\Gamma_{3,3}}}{\Gamma_{2,3}}
\Bigg),
\end{aligned}
\end{equation}
is bijective, with inverse
\begin{equation}\label{T2}
\begin{aligned}
&T^{-1}_{H_1,H_2,H_3}(\rho_{1,2},\rho_{1,3},\rho_{2,3})\\
&\quad=\Bigg(
\frac{\Gamma_{1,2}\rho_{1,2}}{\sqrt{\Gamma_{1,1}\Gamma_{2,2}}},
\frac{\Gamma_{1,3}\rho_{1,3}}{\sqrt{\Gamma_{1,1}\Gamma_{3,3}}},
\frac{\displaystyle\frac{\Gamma_{2,3}\rho_{2,3}}{\sqrt{\Gamma_{2,2}\Gamma_{3,3}}}-\frac{\Gamma_{1,2}\Gamma_{1,3}\rho_{1,2}\rho_{1,3}}{\Gamma_{1,1}\sqrt{\Gamma_{2,2}\Gamma_{3,3}}}}{\displaystyle\sqrt{1-\frac{\Gamma_{1,2}^{2}\rho_{1,2}^{2}}{\Gamma_{1,1}\Gamma_{2,2}}}\sqrt{1-\frac{\Gamma_{1,3}^{2}\rho_{1,3}^{2}}{\Gamma_{1,1}\Gamma_{3,3}}}}
\Bigg).
\end{aligned}
\end{equation}
\end{thm}

The proof is given in Appendix~\ref{ApenB}. In the two-dimensional case,
\[
T_{H_1,H_2}(z_{1,2})=\frac{z_{1,2}\sqrt{\Gamma_{1,1}\Gamma_{2,2}}}{\Gamma_{1,2}}.
\]
Although the application considers two- and three-dimensional trajectories, the same construction extends to general \(p\).

\subsection{Finite-dimensional simulation and velocity reconstruction}\label{sub:simulation_velocity}

For observation times \(\boldsymbol t=(t_0,\ldots,t_n)\), finite-dimensional simulation reduces to assembling the Gaussian covariance matrix in Definition~\ref{defmu3D}. The following recursion computes its entries from covariances of auxiliary variables.

\begin{cor}\label{cor:covariance_mu}
Let \(0=t_0<t_1<\cdots<t_n\), with \(n\in\mathbb{N}\), and define
\[
\mu_{i,j}:=\mu_{H_i}(t_j),\qquad i=1,\ldots,p,\quad j=0,\ldots,n.
\]
Then
\begin{equation}\label{covf}
\operatorname{cov}(\mu_{i_1,i},\mu_{i_2,j})
=\sigma_{i_1}\sigma_{i_2}\sum_{\ell=0}^{i-1}\sum_{r=0}^{j-1}
e^{-\beta_{i_1}(t_i-t_{\ell+1})}
\operatorname{cov}(y_{i_1,\ell},y_{i_2,r})
e^{-\beta_{i_2}(t_j-t_{r+1})},
\end{equation}
where
\[
y_{i,j}:=\int_{t_j}^{t_j+\Delta_j}e^{\beta_i(s-(t_j+\Delta_j))}W_{H_i}(s)\,ds,
\qquad \Delta_j:=t_{j+1}-t_j.
\]
The cross-covariance of the auxiliary variables is
\begin{equation}\label{eq:cov_y_irregular}
\begin{aligned}
\operatorname{cov}(y_{i_1,i},y_{i_2,j})
&=\int_{t_i}^{t_i+\Delta_i}\int_{t_j}^{t_j+\Delta_j}
e^{\beta_{i_1}(u-(t_i+\Delta_i))}
K_{i_1,i_2}(u,v)
e^{\beta_{i_2}(v-(t_j+\Delta_j))}
\,dv\,du\\
&=\frac{\rho_{i_1,i_2}}{2}\Bigg[
\frac{1-e^{-\beta_{i_2}\Delta_j}}{\beta_{i_2}}
\int_0^{\Delta_i}e^{-\beta_{i_1}u}(t_i+\Delta_i-u)^{H_{i_1}+H_{i_2}}\,du\\
&\qquad+
\frac{1-e^{-\beta_{i_1}\Delta_i}}{\beta_{i_1}}
\int_0^{\Delta_j}e^{-\beta_{i_2}v}(t_j+\Delta_j-v)^{H_{i_1}+H_{i_2}}\,dv\\
&\qquad-
\int_0^{\Delta_i}\int_0^{\Delta_j}
e^{-\beta_{i_1}u}
\left|t_j+\Delta_j-v+u-t_i-\Delta_i\right|^{H_{i_1}+H_{i_2}}
e^{-\beta_{i_2}v}\,dv\,du
\Bigg].
\end{aligned}
\end{equation}
\end{cor}

For the simulation recursion, we use \(v_{H_i}(0)=0\). Under this convention, Corollary~\ref{cor:covariance_mu} extends the univariate recursion of \citet{JY} through
\begin{equation}\label{MYJ}
\mu_{i,j}=\mu_{i,0}+\sigma_i\sum_{k=0}^{j-1}y_{i,k}e^{-\beta_i(t_j-t_{k+1})}.
\end{equation}
For a nonzero deterministic initial velocity, the corresponding deterministic mean term in Definition~\ref{defmu3D} must be added; the covariance formula is unchanged. For a regular grid \(t_j=j\Delta\), the cross-covariance of the auxiliary variables becomes
\begin{equation}\label{covy}
\begin{aligned}
\operatorname{cov}(y_{i_1,i},y_{i_2,j})
&=\frac{\rho_{i_1,i_2}}{2}\Bigg[
\frac{1-e^{-\beta_{i_2}\Delta}}{\beta_{i_2}}
\int_0^\Delta e^{-\beta_{i_1}u}\big((i+1)\Delta-u\big)^{H_{i_1}+H_{i_2}}\,du\\
&\qquad+
\frac{1-e^{-\beta_{i_1}\Delta}}{\beta_{i_1}}
\int_0^\Delta e^{-\beta_{i_2}v}\big((j+1)\Delta-v\big)^{H_{i_1}+H_{i_2}}\,dv\\
&\qquad-
\int_0^\Delta\int_0^\Delta
e^{-\beta_{i_1}u}
\left|(j-i)\Delta+u-v\right|^{H_{i_1}+H_{i_2}}
e^{-\beta_{i_2}v}\,dv\,du
\Bigg].
\end{aligned}
\end{equation}
For a regular grid, equation~\eqref{covy} also permits the numerical integrals to be reused across covariance entries. For a fixed coordinate pair, the two one-dimensional terms require at most two sequences of $n$ integrals, while the double-integral term depends only on the signed lag $j-i$. Consequently, a cross-coordinate block requires at most $2n-1$ distinct double integrals, and a marginal block requires only the $n$ nonnegative lags, rather than a separate double integral for each of its $n^2$ entries. For fixed $p$, the number of numerical integrations needed to construct all marginal and cross-coordinate blocks is therefore $O(p^2n)$. This extends the univariate regular-grid reduction of \citet[Remark~3.2]{JY} to the two- and three-dimensional settings considered here. Filling and storing the dense covariance matrix and factoring it for Gaussian simulation or likelihood evaluation retain their usual matrix costs.

Equations~\eqref{covf}--\eqref{covy} determine the finite-dimensional Gaussian distribution of the position process and provide a direct simulation procedure.

Velocity is not observed directly. Following \citet{JY}, we generate velocity trajectories from their conditional Gaussian distribution given the observed positions and plug-in maximum likelihood estimates. Definition~\ref{defmu3D} allows deterministic initial velocities, but the following formulas and implementation use \(v_{H_i}(0)=0\) for every coordinate. If another initial velocity is used, its contribution must be incorporated into the conditioning equations.

\begin{rem}\label{simvel}
For \(i=1,\ldots,p\), define
\[
\zeta_{t,\Delta}^{i,1}:=\sigma_i\left(W_{H_i}(t+\Delta)-e^{-\beta_i\Delta}W_{H_i}(t)\right)
\]
and
\begin{equation}\label{Z2}
\zeta_{t,\Delta}^{i,2}:=\beta_i(1-e^{-\beta_i\Delta})\mu_{H_i}(0)
-\beta_i\left(\mu_{H_i}(t+\Delta)-e^{-\beta_i\Delta}\mu_{H_i}(t)\right).
\end{equation}
Their covariances are
\begin{equation}\label{CR1}
\begin{aligned}
\operatorname{cov}\left(\zeta_{t,\Delta}^{i,1},\zeta_{s,\Delta}^{j,1}\right)
=\sigma_i\sigma_j\Big[&
K_{i,j}(t+\Delta,s+\Delta)
-e^{-\beta_j\Delta}K_{i,j}(t+\Delta,s)\\
&-e^{-\beta_i\Delta}K_{i,j}(t,s+\Delta)
+e^{-(\beta_i+\beta_j)\Delta}K_{i,j}(t,s)\Big],
\end{aligned}
\end{equation}
\begin{equation}\label{CR4}
\operatorname{cov}\left(\zeta_{t,\Delta}^{i,2},\zeta_{s,\Delta}^{j,2}\right)
=\sigma_i\sigma_j\beta_i\beta_j e^{-(\beta_i+\beta_j)\Delta}
\int_t^{t+\Delta}\int_s^{s+\Delta}
e^{\beta_i(u-t)}K_{i,j}(u,v)e^{\beta_j(v-s)}\,dv\,du,
\end{equation}
and
\begin{equation}\label{CR2}
\begin{aligned}
\operatorname{cov}\left(\zeta_{t,\Delta}^{i,2},\zeta_{s,\Delta}^{j,1}\right)
=-\beta_i\sigma_i\sigma_j\Bigg[
&e^{-\beta_i\Delta}\int_t^{t+\Delta}e^{\beta_i(u-t)}K_{i,j}(u,s+\Delta)\,du\\
&-e^{-(\beta_i+\beta_j)\Delta}\int_t^{t+\Delta}e^{\beta_i(u-t)}K_{i,j}(u,s)\,du
\Bigg].
\end{aligned}
\end{equation}
\end{rem}

By \citet[Proposition~3.5]{JY},
\begin{equation}\label{eq:velocity_recursive}
v_{H_i}(t+\Delta)-e^{-\beta_i\Delta}v_{H_i}(t)
=\zeta_{t,\Delta}^{i,1}+\zeta_{t,\Delta}^{i,2}.
\end{equation}
For a regular grid \(\boldsymbol t=(t_0,\ldots,t_n)\), write
\[
\zeta_{\boldsymbol t,\Delta}^{i,r}
:=\left(\zeta_{t_0,\Delta}^{i,r},\ldots,
\zeta_{t_{n-1},\Delta}^{i,r}\right)^\intercal,
\qquad i=1,\ldots,p,\quad r=1,2.
\]
Algorithm~\ref{alg:velocity} generates conditional velocity trajectories from the observed positions and the fitted parameters.

\begin{algorithm}[H]
\baselineskip=14pt
\caption{Generate conditional velocity trajectories from observed positions and plug-in parameter estimates}
\label{alg:velocity}
\centering\begin{minipage}{0.90\linewidth}
\begin{algorithmic}[1]
\small
\State Compute \(\left(\zeta_{\boldsymbol t,\Delta}^{1,2},\ldots,\zeta_{\boldsymbol t,\Delta}^{p,2}\right)\) from the observed positions using equation~\eqref{Z2} and the estimates \(\{\widehat{\beta}_i\}\).
\State Generate \(\left(\zeta_{\boldsymbol t,\Delta}^{1,1},\ldots,\zeta_{\boldsymbol t,\Delta}^{p,1}\right)\) conditionally on
\Statex \hspace{\algorithmicindent}\(\left(\zeta_{\boldsymbol t,\Delta}^{1,2},\ldots,\zeta_{\boldsymbol t,\Delta}^{p,2}\right)\), using equations~\eqref{CR1}, \eqref{CR2}, and \eqref{CR4},
\Statex \hspace{\algorithmicindent}and the estimates \(\{\widehat{\sigma}_i\}\), \(\{\widehat{\beta}_i\}\), \(\{\widehat{H}_i\}\), and \(\{\widehat{\rho}_{i,j}\}\).
\State Set \(v_{H_1}(0)=\cdots=v_{H_p}(0)=0\).
\State Use Equation~\eqref{eq:velocity_recursive} recursively to obtain \(\left(v_{H_1}(\boldsymbol t),\ldots,v_{H_p}(\boldsymbol t)\right)\).
\State Return \(\left(v_{H_1}(\boldsymbol t),\ldots,v_{H_p}(\boldsymbol t)\right)\).
\end{algorithmic}
\end{minipage}
\end{algorithm}

\subsection{Finite-sample numerical experiments}\label{sub:numerical_assessment}

The simulation study of \citet[Appendix~B.2]{JY} examined maximum likelihood estimation of the univariate ifOU parameters \((\sigma,\beta,H)\) from discretely observed position trajectories. They used \(\mu_H(0)=10\), \(\sigma=2\), \(\beta=3\), \(H\in\{0.1,0.2,\ldots,0.9\}\), a fixed observation horizon \(T=10\), and \(\Delta\in\{1/10,1/20,\ldots,1/50\}\). For each combination of \(H\) and \(\Delta\), \(50\) trajectories were simulated and the three parameters were estimated jointly. The reported mean squared errors generally decreased as \(\Delta\) decreased.

Building on this univariate benchmark, we examine joint estimation
and latent-velocity reconstruction in the multidimensional model and
study cross-coordinate correlation estimation under increasing- and
fixed-horizon designs. Sections~\ref{sub:joint-estimation-supp}--\ref{sub:infill-correlation-supp} of Appendix~\ref{sec2.4} report
the three numerical experiments, and Section~\ref{sub:computational-scaling-supp} of Appendix~\ref{sec2.4} discusses computational
considerations.

The first experiment jointly estimates all twelve parameters from one three-dimensional trajectory simulated on \([0,10]\) with \(\Delta=1/30\), \(\boldsymbol{\sigma}=(2,3,2.5)\), \(\boldsymbol{\beta}=(7,5,4.2)\), \(\boldsymbol{H}=(0.8,0.6,0.2)\), and \((\rho_{1,2},\rho_{1,3},\rho_{2,3})=(0.2801,-0.2647,0.3376)\). The maximum likelihood estimates are \(\widehat{\boldsymbol{\sigma}}=(1.7927,2.8061,2.4141)\), \(\widehat{\boldsymbol{\beta}}=(8.5931,5.1493,4.6100)\), \(\widehat{\boldsymbol{H}}=(0.7814,0.5939,0.1697)\), and \((\widehat{\rho}_{1,2},\widehat{\rho}_{1,3},\widehat{\rho}_{2,3})=(0.2679,-0.2448,0.3610)\). For each of the twelve parameters, Table~\ref{tab:original-sim3d} of Appendix~\ref{sec2.4} reports \(99\%\) Gaussian-approximation intervals and one-parameter likelihood-deviance intervals obtained with the remaining parameters fixed at the joint MLE\@. For \(\rho_{1,2}\), the two intervals are \((0.1954,0.3405)\) and \((0.1862,0.3324)\); for \(\rho_{1,3}\), they are \((-0.2925,-0.1971)\) and \((-0.2870,-0.1907)\); and for \(\rho_{2,3}\), they are \((0.3052,0.4169)\) and \((0.2974,0.4103)\), respectively. The true values of the three correlations are contained in both corresponding intervals. The simulated trajectory and the one-parameter likelihood sections are shown in Figures~\ref{fig:sim3d-supp}--\ref{PLC2} of Appendix~\ref{sec2.4}.

Because all twelve quantities enter the same likelihood optimization, this experiment directly examines simultaneous finite-sample estimation of the three scale, three damping, three Hurst, and three cross-correlation parameters for that realization.

The simulated trajectory used for the joint fit, together with its twelve MLEs, is then supplied to Algorithm~\ref{alg:velocity}. With all three position coordinates treated as observed, the MLEs fixed, and \(v_{H_i}(0)=0\), we generated 1,000 independent trajectories from the fitted conditional Gaussian law of the latent velocity. At each grid time, their Monte Carlo average estimates the conditional mean, and their empirical 0.025 and 0.975 quantiles define a pointwise \(95\%\) conditional band. The conditioning positions, conditional means, and pointwise bands for Longitude, Latitude, and Altitude are shown in Figure~\ref{fig:sim-velocity-reconstruction} of Appendix~\ref{sec2.4}. 

The second experiment examines cross-coordinate correlation
estimation with the marginal parameters fixed at their generating
values. Weak and moderate two- and three-dimensional designs are
considered, with 50 independent trajectories generated for each
design and each $n\in\{50,100\}$ on a grid with spacing
$\Delta=0.5$. Figure~\ref{fig:simrho} and Table~\ref{tab:simrho} of Appendix~\ref{sec2.4} report the
empirical distributions, biases, and RMSE values of the correlation
MLEs. The RMSE is lower at $n=100$ for every coordinate pair. The empirical biases are small, but for the weak nonzero correlations the RMSE remains comparable with $|\rho|$. Because $\Delta$ is fixed, increasing $n$
also increases the observation horizon; the comparison therefore
does not isolate the effect of sample size at a fixed horizon.

The third experiment examines estimation of a single
cross-coordinate correlation under fixed-horizon grid refinement.
The observation horizon is fixed at $T=10$, all marginal parameters
are held at their generating values, and
$\rho_{1,2}\in\{0.1,0.3\}$ is the sole unknown parameter. For each
combination of $\rho_{1,2}$ and
$\Delta\in\{1/10,1/20,1/30,1/40,1/50\}$, 100 independent Monte
Carlo datasets are generated. Section~\ref{sub:infill-correlation-supp} of Appendix~\ref{sec2.4} gives
the complete design; Figure~\ref{fig:infillrho-boxplots} and Table~\ref{tab:infillrho} report the empirical
distributions and numerical summaries. As $\Delta$ decreases from
$1/10$ to $1/50$, the RMSE decreases from $0.0934$ to $0.0414$ for
$\rho_{1,2}=0.1$ and from $0.0807$ to $0.0347$ for
$\rho_{1,2}=0.3$. The mean likelihood-ratio interval width also
decreases under grid refinement, whereas the empirical bias remains
small.

\subsection{Interactive simulation application}

An interactive Shiny application implements finite-dimensional simulation from user-specified marginal and canonical partial-correlation parameters. It displays the simulated Longitude, Latitude, and Altitude coordinates on a map and returns circular summaries of successive movement directions in the three coordinate planes; Figure~\ref{fig:shiny-application} shows its interface. The application is available at \url{https://jhramgon.shinyapps.io/aplication_r_imou/}.

\begin{figure}[H]
\centering
\includegraphics[width=0.85\textwidth]{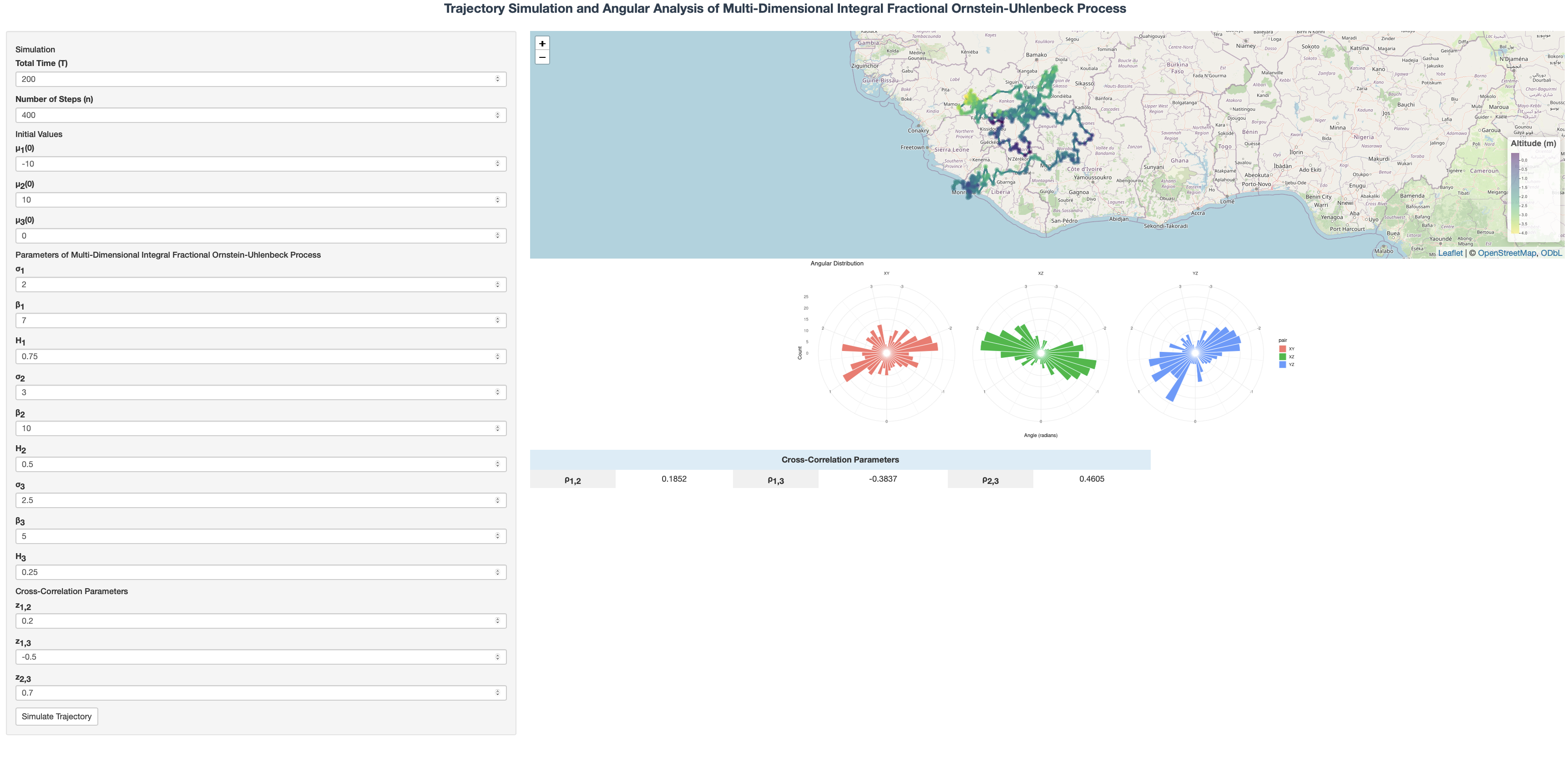}
\caption{Interactive implementation of the mifOU simulator. The input panel specifies the marginal and canonical partial-correlation parameters; the output displays the simulated trajectory, the three pairwise movement-angle distributions, and the resulting cross-coordinate correlations.}
\label{fig:shiny-application}
\end{figure}

\section{Application to bat migration trajectories}\label{SSec3}

We revisit the five common noctule migration trajectories analyzed by \citet{RamirezEtAl2026} from a multivariate perspective. The previous analysis treated Longitude and Latitude separately. For all five bats, these coordinates are taken from the two-dimensional dataset; the three-dimensional dataset supplies Altitude for Bats~3--5. The observations are grouped into consecutive half-minute intervals, and the available coordinate values within each occupied interval are averaged without interpolating empty intervals. This produces 47, 45, 62, 241, and 165 retained locations for Bats~1--5, respectively. The records span 249.5, 133.0, 101.5, 167.0, and 119.0 minutes.

The application is organized in two stages. First, empirical
diagnostics describe temporal variation, dependence across scales,
movement direction, and cross-coordinate association. Second,
Gaussian likelihoods are used to compare marginal covariance
families and complete multivariate trajectory models. The two stages
keep the descriptive analysis separate from model fitting.

\subsection{Data preprocessing and empirical diagnostics}
\label{sub:empirical-stage}

The migration trajectories were recorded by \citet{mara} in Germany
during 2016--2017 and are publicly available at
\url{https://doi.org/10.5441/001/1.5d736bf0}. Longitude and Latitude
are retained in degrees for likelihood evaluation. Altitude,
available for Bats~3--5 and recorded in meters, is converted to
kilometers for likelihood analysis. Its observed ranges are
499.0--588.4, 461.4--1289.4, and 390.1--643.7 meters for
Bats~3--5, respectively.

The empirical diagnostics examine local level and variability,
serial dependence before and after differencing, scale-dependent
fluctuations, movement direction, and cross-coordinate association.
For Altitude, centered rolling means and sample variances are
calculated in complete windows of 10 successive retained
observations. Sample autocorrelation functions are computed for the
observed series and their first and second differences through lag
20. For each plotted series \(Z\), the displayed dashed lines are
the approximate pointwise white-noise bounds
\[
\pm\frac{1.96}{\sqrt{n_Z}},
\]
where \(n_Z\) is the length of that series. Thus, the bounds use
\(N\), \(N-1\), and \(N-2\) observations for the levels, first
differences, and second differences, respectively.

The augmented Dickey--Fuller (ADF) and KPSS tests are interpreted
jointly because they use different null hypotheses. The ADF null is
the presence of a unit root, whereas the KPSS null is level
stationarity. Rejection of the ADF null together with failure to
reject the KPSS null is compatible with stationarity. Failure to
reject both null hypotheses is inconclusive.

Detrended fluctuation analysis (DFA) summarizes how locally
detrended fluctuations change with the observation scale. For a
series \(X_1,\ldots,X_N\), define the centered cumulative profile
\[
Y(j)
=
\sum_{k=1}^{j}\bigl(X_k-\overline X\bigr),
\qquad
j=1,\ldots,N.
\]
For each scale \(s\), let
\[
q_s
=
\left\lfloor\frac{N}{s}\right\rfloor.
\]
The profile is partitioned into \(q_s\) complete nonoverlapping
windows of length \(s\), starting from the beginning of the series.
The partition is then repeated from the end, yielding \(2q_s\)
windows and retaining the observations otherwise excluded when
\(N\) is not divisible by \(s\). Let
\(\mathcal I_{\nu,s}\) denote window \(\nu\), and let
\(\widehat Y_{\nu,s}(j)\) be the linear trend fitted within that
window. Define
\[
F_{\nu}^{2}(s)
=
\frac{1}{s}
\sum_{j\in\mathcal I_{\nu,s}}
\left\{
Y(j)-\widehat Y_{\nu,s}(j)
\right\}^{2},
\qquad
\nu=1,\ldots,2q_s,
\]
and
\[
F(s)
=
\left\{
\frac{1}{2q_s}
\sum_{\nu=1}^{2q_s}
F_{\nu}^{2}(s)
\right\}^{1/2}.
\]
Up to 12 distinct integer scales are obtained by flooring a
logarithmically spaced grid from \(4\) to \(\lfloor N/4\rfloor\).
If \(J\) distinct scales are available, the reported DFA slope
\(\widehat{\alpha}\) is obtained by fitting
\[
\log F(s)=a+\alpha\log s
\]
over the largest \(\max\{4,\lceil J/2\rceil\}\) scales. It is treated
as a descriptive scaling coefficient, not as an estimate of a mifOU
Hurst parameter or as a formal test of long-range dependence.

Empty half-minute intervals are not interpolated. All
sequence-based diagnostics are therefore applied in retained
observation order: rolling windows, autocorrelation lags, and DFA
scales do not represent fixed increments of elapsed time. The
resulting diagnostics are interpreted descriptively for these
irregularly spaced records. Complete Altitude diagnostics are
reported in Figures~\ref{fig:altstat}--\ref{fig:altdfa} and Tables~\ref{tab:altstat}--\ref{tab:altdfa} of Appendix~\ref{ApenD}.

The Longitude and Latitude diagnostics follow
\citet{RamirezEtAl2026}. Across the ten horizontal coordinate
series, KPSS rejected level stationarity in every observed series,
whereas ADF rejected the unit-root null in only two cases. After
first differencing, the two tests were compatible with stationarity
for four series; the remaining six did not satisfy both criteria.
Several first-difference autocorrelation functions also contained
values outside the approximate pointwise bounds. The fitted DFA
slopes exceeded one in all ten series, but they are interpreted only
as finite-range scaling summaries. Thus, first differencing does not
provide a uniform stationary description of the horizontal
coordinates.

The corresponding Altitude results are summarized in
Table~\ref{tab:altitude-diagnostics-main}. For all three Altitude
series, ADF fails to reject the unit-root null. For Bats~3--4, KPSS
also fails to reject level stationarity, so the paired results are
inconclusive. For Bat~5, KPSS rejects level stationarity. After
first differencing, the ADF and KPSS results are compatible with
stationarity for Bats~4--5. The result for Bat~3 remains
inconclusive because ADF gives \(p=0.072\). The fitted DFA slopes
exceed one for all three trajectories, although the regression
interval for Bat~5 includes one. Taken together, these diagnostics
do not justify restricting the Altitude comparison to stationary
covariance models.

\begin{table}[H]
\centering
\small
\setlength{\tabcolsep}{4pt}
\renewcommand{\arraystretch}{1.12}
\begin{tabular}{crrrrc}
\toprule
& \multicolumn{2}{c}{Altitude level}
& \multicolumn{2}{c}{First difference}
& \shortstack{DFA slope\\(95\% regression interval)}\\
\cmidrule(lr){2-3}\cmidrule(lr){4-5}
Bat
& ADF \(p\)
& KPSS \(p\)
& ADF \(p\)
& KPSS \(p\)
& \(\widehat{\alpha}\)\\
\midrule
3
& 0.383
& \(\geq 0.100\)
& 0.072
& \(\geq 0.100\)
& \(1.578\ [1.241,1.915]\)\\
4
& 0.664
& 0.094
& 0.017
& \(\geq 0.100\)
& \(1.742\ [1.337,2.147]\)\\
5
& 0.525
& 0.024
& \(\leq 0.010\)
& \(\geq 0.100\)
& \(1.229\ [0.811,1.648]\)\\
\bottomrule
\end{tabular}
\caption{Selected Altitude diagnostics before model fitting. ADF
and KPSS \(p\)-values are reported for the observed levels and first
differences. The final column gives the DFA slope fitted over the
selected large-scale range and its regression interval. Values shown
as \(\leq 0.010\) or \(\geq 0.100\) are bounds of the tabulated
\(p\)-value approximations. Complete ADF and KPSS results, including
second differences, and the full DFA summaries are reported in
Tables~\ref{tab:altstat}--\ref{tab:altdfa} of Appendix~\ref{ApenD}.}
\label{tab:altitude-diagnostics-main}
\end{table}

Movement direction and cross-coordinate association are summarized
from successive displacements expressed in kilometers. Let
\(\lambda_k\) and \(\phi_k\) denote Longitude and Latitude in
radians, let \(A_k\) denote Altitude in meters, and let
\(\overline\phi\) be the trajectory-specific mean latitude. Define
\[
x_{1,k}
=
R_{\mathrm E}\cos(\overline\phi)(\lambda_k-\lambda_0),
\qquad
x_{2,k}
=
R_{\mathrm E}(\phi_k-\phi_0),
\qquad
x_{3,k}
=
\frac{A_k}{1000},
\]
where \(R_{\mathrm E}=6371.0088\ \mathrm{km}\) and
\((\lambda_0,\phi_0)\) is the first retained horizontal location.
For Bats~1--2, only the two horizontal coordinates are available.
For successive retained observations, write
\[
\Delta x_{a,k}
=
x_{a,k+1}-x_{a,k}.
\]
For a coordinate pair \((a,b)\), let
\[
\mathcal K_{a,b}
=
\left\{
k:
(\Delta x_{a,k},\Delta x_{b,k})\neq(0,0)
\right\}.
\]
For \(k\in\mathcal K_{a,b}\), define the projected displacement
direction by
\[
\theta_{a,b,k}
=
\operatorname{atan2}
\left(
\Delta x_{b,k},
\Delta x_{a,k}
\right)
\in(-\pi,\pi].
\]
If \(m_{a,b}=|\mathcal K_{a,b}|\), directional concentration is
summarized by the mean resultant length
\[
\overline R_{a,b}
=
\left|
\frac{1}{m_{a,b}}
\sum_{k\in\mathcal K_{a,b}}
e^{\,\mathrm{i}\theta_{a,b,k}}
\right|.
\]
Values near one indicate concentration around a common direction,
whereas values near zero indicate little directional concentration.

Because observation gaps are unequal, coordinate-specific
displacement rates are defined by
\[
u_{a,k}
=
\frac{\Delta x_{a,k}}{t_{k+1}-t_k},
\]
with time measured in minutes. Pearson correlations summarize
linear association between pairs of rate components, whereas
Spearman correlations summarize monotone rank association. The
reported two-sided \(p\)-values are conventional and unadjusted:
they account for neither serial dependence between adjacent rates
nor the multiple coordinate pairs examined. They are therefore
interpreted descriptively. These empirical correlations are not
estimates of the model parameters \(\rho_{i,j}\).

For Bat~1, the Longitude--Latitude association is weak and negative:
the Pearson and Spearman coefficients are \(-0.284\)
(\(p=0.0557\)) and \(-0.201\) (\(p=0.180\)), respectively. For
Bat~2, both associations are more clearly negative, with
coefficients \(-0.381\) (\(p=0.0107\)) and \(-0.420\)
(\(p=0.0045\)), respectively. For Bats~3--5, both measures identify
Longitude--Latitude as the principal associated coordinate pair.
The Pearson coefficients range from \(0.317\) to \(0.580\), and the
Spearman coefficients range from \(0.444\) to \(0.747\).
Associations involving Altitude are smaller: their Pearson
correlations have absolute value at most \(0.146\), and none of the
corresponding Pearson tests has an unadjusted \(p\)-value below
\(0.05\). Among the Spearman correlations involving Altitude, only
the Latitude--Altitude pair for Bat~5 has \(p<0.05\), with
coefficient \(0.186\) and \(p=0.0172\).

In the Longitude--Latitude plane, the mean resultant lengths for
Bats~3--5 are \(0.875\), \(0.765\), and \(0.079\), respectively.
Horizontal displacements are therefore strongly concentrated around
a preferred direction for Bats~3--4 but not for Bat~5. Directional
concentration and correlation between coordinate rates describe
different features of the trajectories and need not lead to the
same ordering across individuals. Complete directional summaries
and cross-coordinate correlations are reported in
Figures~\ref{fig:angles12}--\ref{fig:angles5} and Table~\ref{tab:desccorr} of Appendix~\ref{ApenD}.

For Bats~3--4, the positive Longitude--Latitude rate associations and
high mean resultant lengths describe coordinated horizontal movement
along a preferred direction. Bat~5 shows a positive horizontal rate
association but little directional concentration. Associations involving
Altitude are smaller. These summaries are descriptive and are not
assigned to a specific behavioral mechanism.

\subsection{Likelihood-based model comparison at half-minute resolution}
\label{sub:model-fitting-stage}

Complete-trajectory models are constructed from the seven univariate
families considered by \citet{RamirezEtAl2026}, together with
admissible correlated iOU and ifOU specifications. All independent
products are retained rather than only the product of the
coordinatewise winners. For Bats~1--2, the \(7^2=49\) independent
products are compared with correlated iOU and ifOU
Longitude--Latitude models, giving 51 candidates. For Bats~3--5,
the \(7^3=343\) independent products are augmented by
\(3\times2\times7=42\) models formed by one correlated iOU or ifOU
coordinate pair and one independent singleton family, together with
full three-dimensional iOU and ifOU models. Thus, each
three-dimensional comparison contains 387 candidates. The other five
univariate families enter a correlated candidate only as an
independent singleton component; no cross-covariance involving those
families is imposed.

Every correlated candidate is fitted by a separate maximization of
its Gaussian likelihood over all free marginal and dependence
parameters. In particular, a full three-dimensional model has its
own maximum-likelihood estimates and is not evaluated at estimates
inherited from an independent or reduced model.

At half-minute resolution, the common multivariate implementation is
used within each fixed combination of marginal families to quantify
the AIC change produced by introducing cross-coordinate dependence.
The resulting change is anchored to the coordinatewise reference AIC
values. Let \(C\) denote a correlated candidate and let \(I(C)\)
denote the independently fitted product model with the same
coordinate-specific marginal families. Write
\(\mathrm{AIC}_{\mathrm g}\) for the AIC obtained with the common
multivariate implementation. For coordinate \(j\), let
\(\mathrm{AIC}^{\mathrm{ref}}_j(M_j)\) denote the reference
univariate AIC of family \(M_j\). For Longitude and Latitude, these
are exactly the values reported in Table~1 of
\citet{RamirezEtAl2026}; for Altitude, they are those in
Table~\ref{tab:altitude-all-main}. The calibrated score is
\[
\mathrm{AIC}_{\mathrm{cal}}(C)
=
\sum_j\mathrm{AIC}^{\mathrm{ref}}_j(M_j)
+
\left\{
\mathrm{AIC}_{\mathrm g}(C)
-
\mathrm{AIC}_{\mathrm g}\!\left(I(C)\right)
\right\},
\]
where the sum is over the coordinates observed for that individual.
The term in braces compares two separate maximum-likelihood fits
under the same marginal-family combination and the same likelihood
implementation; it therefore does not hold the marginal parameters
at the estimates from \(I(C)\). Absolute values of
\(\mathrm{AIC}_{\mathrm g}\) are not compared across different
marginal-family combinations. For each independent product,
\(\mathrm{AIC}_{\mathrm{cal}}\) is the sum of its coordinatewise
reference AIC values.

The reference AIC values retain the parameter counts used in their
respective univariate comparisons. In the common multivariate fits
entering the term in braces, \(C\) and \(I(C)\) use the same nominal
marginal parametrizations, so all shared marginal counts cancel from
their AIC difference. Complete-model parameter counts include the
nominal marginal dimensions and the active correlations; within
these multivariate fits, \(\zeta\) is represented by
\((\beta,\sigma^2)\), including when its optimum occurs at
\(\widehat\beta=0\). The ifOU, iOU, and multivariate OU-position
fits use \(0.001<\beta\leq400\), and each two- or
three-dimensional block is maximized separately over its complete
parameter vector. The newly fitted Confluent Altitude components and
the Confluent fits at one- and two-minute resolution use
\(0.001<\alpha\leq100\) and
\(0.001<\beta\leq2000\).  Weights derived from
\(\mathrm{AIC}_{\mathrm{cal}}\) are normalized separately within
each bat over the 51 or 387 candidates. Tables~9--10 of Web
Appendix~\ref{ApenD} separate the evidence for estimated Hurst parameters from
the evidence for cross-coordinate correlation, Table~\ref{tab:aic-components} reports the
two common-fit AIC values and their difference, and Tables~12--13
report the one- and two-minute analyses.

Information criteria are compared only among candidates fitted to
the same response vector for the same individual. In particular,
bivariate Longitude--Latitude likelihoods are not compared with
three-dimensional likelihoods. For Bats~3--5, every complete
candidate uses the same Longitude, Latitude, and Altitude
observations, so criterion differences reflect covariance family,
dependence structure, and parameter count rather than the inclusion
of an additional coordinate.

\subsubsection{Marginal covariance models for Altitude}
\label{sub:altitude-models}

Each observed Altitude series is fitted with the same seven families
used in the Longitude and Latitude comparison: the exponential-weight
\(\zeta\) process; ifOU \(\mu_H\);
iOU \(\mu_{1/2}\); scaled fractional Brownian motion \(\sigma W_H\);
the stationary confluent hypergeometric covariance; and the
stationary Stein models \(S^{(2)}\) and \(S^{(3)}\).

Suppose Bat \(b\) has \(N_b\) retained Altitude observations, and let
\(A_b(t_i)\) denote the recorded Altitude in meters, with the first
retained observation indexed by \(t_0=0\). The analyzed response is
\[
Y_{b,i}
=
\frac{A_b(t_i)-A_b(0)}{1000},
\qquad i=1,\ldots,N_b-1,
\]
and is therefore expressed in kilometers relative to the first
retained observation. Within each bat, all seven models are fitted
to the same response vector and observation times. For a model with
parameter vector \(\theta\) and \(k\) estimated parameters, the
reported criterion is
\[
\mathrm{AIC}
=
-2\ell(\widehat{\theta})+2k.
\]
This comparison is strictly univariate and estimates no correlation
between Altitude and the horizontal coordinates. For a stationary
candidate, subtraction of the initial observed Altitude is
deterministic centering; its covariance is therefore evaluated as
\(R(t_i-t_j)\), rather than as the covariance of a stochastic
difference from \(A_b(0)\).

For \(\zeta\), the interior specification \(\beta>0\) and the
restriction \(\beta=0\) are fitted separately. The interior model
uses \(k=2\), whereas the boundary model fixes \(\beta=0\) and uses
\(k=1\); the lower-AIC specification represents the family. 

Table~\ref{tab:altitude-all-main} reports the seven Altitude fits for
each individual. For Bat~3, fBM has the smallest AIC, while iOU and
ifOU are \(0.605\) and \(2.020\) units higher, respectively; both
OU-position fits attain \(\beta_{\max}=400\). For Bat~4,
\(\zeta\) has Akaike weight \(0.994\) and an AIC \(10.695\) units
lower than that of ifOU. For Bat~5, iOU has the smallest AIC, while
ifOU and \(\zeta\) are \(0.819\) and \(2.577\) units higher,
respectively. Thus, the minimum-AIC Altitude family differs among
the three individuals.

\begin{table}[H]
\centering
\scriptsize
\setlength{\tabcolsep}{3pt}
\renewcommand{\arraystretch}{1.02}
\setlength{\abovecaptionskip}{3pt}
\setlength{\belowcaptionskip}{0pt}
\begin{tabular}{@{}clccr@{}}
\toprule
Bat ID
& Model
& \shortstack{Estimated\\parameters}
& \shortstack{Maximum-likelihood\\estimates}
& AIC\\
\midrule
3 & $\zeta$
  & $(\beta,\sigma^2)$
  & $(0.026807,2.8523\times10^{-4})$
  & -376.783\\
  & ifOU $\mu_H$
  & $(\beta,\sigma,H)$
  & $(400^*,3.3276,0.44914)$
  & -409.818\\
  & iOU $\mu_{1/2}$
  & $(\beta,\sigma)$
  & $(400^*,3.3428)$
  & -411.233\\
  & \textbf{fBM $\sigma W_H$}
  & $(H,\sigma)$
  & $(0.45014,0.008295)$
  & \textbf{-411.838}\\
  & Confluent
  & $(\sigma^2,\eta,\alpha,\beta)$
  & $(0.001906,0.4585,100.000,1027.162)$
  & -403.608\\
  & Stein $S^{(2)}$
  & $(\sigma^2,\lambda)$
  & $(0.002089,0.01733)$
  & -407.247\\
  & Stein $S^{(3)}$
  & $(\sigma^2,\lambda)$
  & $(0.001997,0.01205)$
  & -407.284\\
\midrule
4 & \textbf{$\zeta$}
  & $(\beta,\sigma^2)$
  & $(0.005518,1.7032\times10^{-4})$
  & \textbf{-1532.139}\\
  & ifOU $\mu_H$
  & $(\beta,\sigma,H)$
  & $(0.06778,0.04709,0.06148)$
  & -1521.443\\
  & iOU $\mu_{1/2}$
  & $(\beta,\sigma)$
  & $(1.0041,0.03699)$
  & -1509.059\\
  & fBM $\sigma W_H$
  & $(H,\sigma)$
  & $(0.96548,0.04102)$
  & -1518.242\\
  & Confluent
  & $(\sigma^2,\eta,\alpha,\beta)$
  & $(0.2187,1.0334,0.4358,51.418)$
  & -1511.882\\
  & Stein $S^{(2)}$
  & $(\sigma^2,\lambda)$
  & $(0.2373,7.738\times10^{-4})$
  & -1324.613\\
  & Stein $S^{(3)}$
  & $(\sigma^2,\lambda)$
  & $(0.2085,5.872\times10^{-4})$
  & -1324.741\\
\midrule
5 & $\zeta$
  & $(\beta,\sigma^2)$
  & $(0.012414,2.0344\times10^{-4})$
  & -1059.092\\
  & ifOU $\mu_H$
  & $(\beta,\sigma,H)$
  & $(0.87739,0.03180,0.30353)$
  & -1060.851\\
  & \textbf{iOU $\mu_{1/2}$}
  & $(\beta,\sigma)$
  & $(1.6401,0.03799)$
  & \textbf{-1061.670}\\
  & fBM $\sigma W_H$
  & $(H,\sigma)$
  & $(0.88172,0.02126)$
  & -1045.823\\
  & Confluent
  & $(\sigma^2,\eta,\alpha,\beta)$
  & $(0.004594,1.2053,0.3788,5.963)$
  & -1054.243\\
  & Stein $S^{(2)}$
  & $(\sigma^2,\lambda)$
  & $(0.008746,0.01164)$
  & -995.987\\
  & Stein $S^{(3)}$
  & $(\sigma^2,\lambda)$
  & $(0.007222,0.009401)$
  & -996.255\\
\bottomrule
\end{tabular}
\caption{Univariate Altitude fits at half-minute resolution.
MLEs follow the parameter order shown; bold marks the minimum AIC
within each bat and a star marks \(\beta_{\max}=400\). The Confluent
fits use \(\alpha\in(0.001,100]\) and
\(\beta\in(0.001,2000]\), with Bat~3 attaining
\(\alpha_{\max}=100\). Table~\ref{tab:altitude-top3-supp} of Appendix~\ref{ApenD} reports the three
leading models, their \(\Delta\)AIC values, and Akaike weights.}
\label{tab:altitude-all-main}
\end{table}

The selected Altitude families imply different temporal structures.
Bat~3 selects fBM with \(\widehat H=0.4501\): its level is
nonstationary, while its increments are stationary and have short
memory; iOU is only \(0.605\) AIC units higher. Bat~4 selects
\(\zeta\), \(10.695\) units ahead of ifOU; this process is neither
stationary nor intrinsically stationary and has logarithmic
covariance growth. Bat~5 selects iOU, with ifOU only \(0.819\) units
higher; at \(H=1/2\), its separated-increment covariance decays
exponentially. The selected Altitude models therefore do not share a
stationarity or memory~regime~.

\begin{landscape}
\vspace*{0.5em}
\subsubsection{Complete-trajectory model comparison}
\label{sub:complete-models}

\begin{table}[H]
\centering
\tiny
\setlength{\tabcolsep}{2.8pt}
\renewcommand{\arraystretch}{1.00}
\caption{Three lowest calibrated-AIC complete-trajectory models at
half-minute resolution. Each correlated candidate was fitted by a
separate joint maximum-likelihood optimization, and its reported
estimates are specific to that model. The horizontal independent fits
and AIC values are those of Table~1 in \citet{RamirezEtAl2026}.}
\label{tab:topthree-main}
\begin{tabularx}{\linewidth}{@{}cc>{\raggedright\arraybackslash}p{.18\linewidth}>{\raggedright\arraybackslash}X>{\raggedright\arraybackslash}p{.165\linewidth}>{\raggedleft\arraybackslash}p{.14\linewidth}@{}}
\toprule
Bat & Rank & Model & Model-specific estimates & Active correlations & Model support\\
\midrule
1 & \textbf{1} & \textbf{fBM(Lon)\(\oplus\)fBM(Lat)} & \textbf{Lon fBM \((H,\sigma)=(0.70300,0.006)\);\newline
Lat fBM \((H,\sigma)=(0.77900,0.004)\)} & \textbf{zero (independent)} & \textbf{\shortstack[r]{$k=4$, $\mathrm{AIC}_{\rm cal}=-532.067$\\$\Delta=0.000$, $w=0.500$}}\\
 & 2 & ifOU(Lon)\(\oplus\)fBM(Lat) & Lon ifOU \((\beta^*,\sigma,H)=(118.714^{*},0.72200,0.70300)\);\newline
Lat fBM \((H,\sigma)=(0.77900,0.004)\) & zero (independent) & \shortstack[r]{$k=5$, $\mathrm{AIC}_{\rm cal}=-530.042$\\$\Delta=2.025$, $w=0.182$}\\
 & 3 & fBM(Lon)\(\oplus\)ifOU(Lat) & Lon fBM \((H,\sigma)=(0.70300,0.006)\);\newline
Lat ifOU \((\beta^*,\sigma,H)=(219.167^{*},0.85900,0.77900)\) & zero (independent) & \shortstack[r]{$k=5$, $\mathrm{AIC}_{\rm cal}=-530.024$\\$\Delta=2.043$, $w=0.180$}\\
\addlinespace[2pt]
2 & \textbf{1} & \textbf{fBM(Lon)\(\oplus\)fBM(Lat)} & \textbf{Lon fBM \((H,\sigma)=(0.68200,0.006)\);\newline
Lat fBM \((H,\sigma)=(0.69400,0.005)\)} & \textbf{zero (independent)} & \textbf{\shortstack[r]{$k=4$, $\mathrm{AIC}_{\rm cal}=-555.663$\\$\Delta=0.000$, $w=0.465$}}\\
 & 2 & fBM(Lon)\(\oplus\)ifOU(Lat) & Lon fBM \((H,\sigma)=(0.68200,0.006)\);\newline
Lat ifOU \((\beta^*,\sigma,H)=(99.1470^{*},0.49400,0.69400)\) & zero (independent) & \shortstack[r]{$k=5$, $\mathrm{AIC}_{\rm cal}=-553.641$\\$\Delta=2.022$, $w=0.169$}\\
 & 3 & ifOU(Lon)\(\oplus\)fBM(Lat) & Lon ifOU \((\beta^*,\sigma,H)=(101.543^{*},0.62500,0.68200)\);\newline
Lat fBM \((H,\sigma)=(0.69400,0.005)\) & zero (independent) & \shortstack[r]{$k=5$, $\mathrm{AIC}_{\rm cal}=-553.640$\\$\Delta=2.023$, $w=0.169$}\\
\addlinespace[2pt]
3 & \textbf{1} & \textbf{ifOU(Lon,Lat)\(\oplus\)fBM(Alt)} & \textbf{Lon ifOU \((\beta,\sigma,H)=(8.3005,0.02857,0.84145)\);\newline
Lat ifOU \((\beta,\sigma,H)=(2.7698,0.008055,0.74152)\);\newline
Alt fBM \((H,\sigma)=(0.45014,0.008295)\)} & \textbf{\(\widehat\rho_{\rm Lon,Lat}=0.639\)} & \textbf{\shortstack[r]{$k=9$, $\mathrm{AIC}_{\rm cal}=-1585.780$\\$\Delta=0.000$, $w=0.358$}}\\
 & 2 & ifOU(Lon,Lat)\(\oplus\)iOU(Alt) & Lon ifOU \((\beta,\sigma,H)=(8.3005,0.02857,0.84145)\);\newline
Lat ifOU \((\beta,\sigma,H)=(2.7698,0.008055,0.74152)\);\newline
Alt iOU \((\beta,\sigma)=(400^{*},3.3428)\) & \(\widehat\rho_{\rm Lon,Lat}=0.639\) & \shortstack[r]{$k=9$, $\mathrm{AIC}_{\rm cal}=-1585.175$\\$\Delta=0.605$, $w=0.264$}\\
 & 3 & ifOU(Lon,Lat,Alt) & Lon ifOU \((\beta,\sigma,H)=(6.7097,0.02323,0.83045)\);\newline
Lat ifOU \((\beta,\sigma,H)=(2.3093,0.006953,0.71698)\);\newline
Alt ifOU \((\beta,\sigma,H)=(394.234,3.3018,0.45186)\) & \((0.631,0.183,0.338)\) in Lon--Lat, Lon--Alt, Lat--Alt order & \shortstack[r]{$k=12$, $\mathrm{AIC}_{\rm cal}=-1584.285$\\$\Delta=1.495$, $w=0.169$}\\
\bottomrule
\end{tabularx}
\end{table}

\newpage
\begin{table}[H]
\ContinuedFloat
\centering
\tiny
\setlength{\tabcolsep}{2.8pt}
\renewcommand{\arraystretch}{1.00}
\caption[]{Three lowest calibrated-AIC complete-trajectory models at half-minute resolution (continued).}
\begin{tabularx}{\linewidth}{@{}cc>{\raggedright\arraybackslash}p{.18\linewidth}>{\raggedright\arraybackslash}X>{\raggedright\arraybackslash}p{.165\linewidth}>{\raggedleft\arraybackslash}p{.14\linewidth}@{}}
\toprule
Bat & Rank & Model & Model-specific estimates & Active correlations & Model support\\
\midrule
4 & \textbf{1} & \textbf{ifOU(Lon,Lat)\(\oplus\)\(\zeta\)(Alt)} & \textbf{Lon ifOU \((\beta,\sigma,H)=(141.330,0.66776,0.87475)\);\newline
Lat ifOU \((\beta,\sigma,H)=(22.1510,0.06221,0.88398)\);\newline
Alt \(\zeta\) \((\beta,\sigma^2)=(0.005518,1.7032\times10^{-4})\)} & \textbf{\(\widehat\rho_{\rm Lon,Lat}=0.579\)} & \textbf{\shortstack[r]{$k=9$, $\mathrm{AIC}_{\rm cal}=-6417.739$\\$\Delta=0.000$, $w=0.993$}}\\
 & 2 & ifOU(Lon,Lat)\(\oplus\)ifOU(Alt) & Lon ifOU \((\beta,\sigma,H)=(141.330,0.66776,0.87475)\);\newline
Lat ifOU \((\beta,\sigma,H)=(22.1510,0.06221,0.88398)\);\newline
Alt ifOU \((\beta,\sigma,H)=(0.06778,0.04709,0.06148)\) & \(\widehat\rho_{\rm Lon,Lat}=0.579\) & \shortstack[r]{$k=10$, $\mathrm{AIC}_{\rm cal}=-6407.044$\\$\Delta=10.695$, $w=0.005$}\\
 & 3 & ifOU(Lon,Lat,Alt) & Lon ifOU \((\beta,\sigma,H)=(202.086,0.95648,0.87514)\);\newline
Lat ifOU \((\beta,\sigma,H)=(22.7547,0.06333,0.88185)\);\newline
Alt ifOU \((\beta,\sigma,H)=(0.07884,0.03915,0.09202)\) & \((0.577,0.026,0.042)\) in Lon--Lat, Lon--Alt, Lat--Alt order & \shortstack[r]{$k=12$, $\mathrm{AIC}_{\rm cal}=-6404.460$\\$\Delta=13.278$, $w=0.001$}\\
\addlinespace[2pt]
5 & \textbf{1} & \textbf{\(\zeta\)(Lon)\(\oplus\)Confluent(Lat)\(\oplus\)iOU(Alt)} & \textbf{Lon \(\zeta\) \((\beta,\sigma^2)=(0.001,4.1660\times10^{-6})\);\newline
Lat Confluent \((\sigma^2,\eta,\alpha,\beta)=(0.005,0.90100,9.0990,406.667)\);\newline
Alt iOU \((\beta,\sigma)=(1.6401,0.03799)\)} & \textbf{zero (independent)} & \textbf{\shortstack[r]{$k=8$, $\mathrm{AIC}_{\rm cal}=-4299.775$\\$\Delta=0.000$, $w=0.105$}}\\
 & 2 & \(\zeta\)(Lon)\(\oplus\)fBM(Lat)\(\oplus\)iOU(Alt) & Lon \(\zeta\) \((\beta,\sigma^2)=(0.001,4.1660\times10^{-6})\);\newline
Lat fBM \((H,\sigma)=(0.88800,0.003)\);\newline
Alt iOU \((\beta,\sigma)=(1.6401,0.03799)\) & zero (independent) & \shortstack[r]{$k=6$, $\mathrm{AIC}_{\rm cal}=-4299.147$\\$\Delta=0.629$, $w=0.077$}\\
 & 3 & \(\zeta\)(Lon)\(\oplus\)Confluent(Lat)\(\oplus\)ifOU(Alt) & Lon \(\zeta\) \((\beta,\sigma^2)=(0.001,4.1660\times10^{-6})\);\newline
Lat Confluent \((\sigma^2,\eta,\alpha,\beta)=(0.005,0.90100,9.0990,406.667)\);\newline
Alt ifOU \((\beta,\sigma,H)=(0.87739,0.03180,0.30353)\) & zero (independent) & \shortstack[r]{$k=9$, $\mathrm{AIC}_{\rm cal}=-4298.956$\\$\Delta=0.819$, $w=0.070$}\\
\bottomrule
\end{tabularx}
\vspace{2pt}
\begin{minipage}{0.985\linewidth}
\tiny
\textit{Notes.} The quantities \(\Delta\) and \(w\) are based on
\(\mathrm{AIC}_{\mathrm{cal}}\), with weights normalized separately
within each bat over the 51 or 387 candidates. Parameter order is
\((\beta,\sigma,H)\) for ifOU and \((\beta,\sigma)\) for iOU, with
\(H=1/2\) fixed. A star marks the numerical bound
\(\beta_{\max}=400\) in a multivariate fit. In an independent row,
\(\beta^*\) retains the profile-plateau meaning specified in Table~1 of
\citet{RamirezEtAl2026}. Correlation triples are ordered as
Longitude--Latitude, Longitude--Altitude, and Latitude--Altitude. Each
full three-dimensional model is fitted jointly and has its own
coordinate-specific estimates. Independent components shown here use
the univariate reference estimates whose AIC values anchor
\(\mathrm{AIC}_{\mathrm{cal}}\).
\end{minipage}
\end{table}
\end{landscape}

Table~\ref{tab:topthree-main} reports the three lowest calibrated
AIC values for each individual. Every correlated row results from a
separate maximization over all free marginal and dependence
parameters. Tables~\ref{tab:familyaic}--\ref{tab:3daic} of Appendix~\ref{ApenD} report the two restricted
OU-position comparisons. Table~\ref{tab:aic-components} gives the global-refit increment and
the reference marginal AIC sum entering the calibrated score, and
Tables~12--13 report the one- and two-minute rankings and
penalty-sensitivity results.

The restricted comparisons separate evidence for fractional temporal
structure from evidence for cross-coordinate dependence. Ordinary AIC and AICc
favor mifOU over the \(H_i=1/2\) miOU restriction for all five
trajectories after optimizing the dependence structure within each
temporal family. After allowing iOU as the \(H=1/2\) special case within
each dependence class, the restricted OU-position comparison selects
independent coordinates for Bats~1--2 and a correlated
Longitude--Latitude block for Bats~3--5. The complete seven-family
ranking, however, retains a more flexible independent marginal product
for Bat~5.

\paragraph{Empirical associations and fitted dependence}

The Pearson and Spearman coefficients summarize association between
adjacent coordinate-specific displacement rates, whereas
\(\widehat\rho\) parametrizes dependence in the latent mfBm driver.
These quantities refer to different stochastic objects and need not
agree in magnitude or sign. The empirical coefficients are used as
descriptive context; evidence for dependence in the fitted model is
assessed through likelihood comparison.

For Bat~1, the complete ranking selects the independent
fBM(Longitude)\(\oplus\)fBM(Latitude) product, with calibrated
Akaike weight \(0.500\). Its three leading complete models are
independent products, consistent with the weak descriptive
Longitude--Latitude association. 

For Bat~2, the same independent fBM product has the smallest
calibrated AIC, with weight \(0.465\). Although the observed
coordinate-specific rates show a negative association, the three
leading complete models are again independent products. 

For Bat~3, all three leading models contain a correlated ifOU
Longitude--Latitude block. The first two estimate
\(\widehat\rho_{\rm Lon,Lat}=0.639\) and combine the
block with independent fBM or iOU Altitude components. The full
three-dimensional ifOU model remains competitive in the complete
calibrated ranking and estimates
\((\widehat\rho_{\rm Lon,Lat},\widehat\rho_{\rm Lon,Alt},
\widehat\rho_{\rm Lat,Alt})=(0.631,0.183,0.338)\).
Longitude--Latitude remains the principal fitted pair, although the
small separation of the full model does not clearly exclude additional
dependence involving Altitude.

For Bat~4, all three leading models contain a correlated ifOU
Longitude--Latitude block. The minimum combines this block, with
\(\widehat\rho_{\rm Lon,Lat}=0.579\), and an
independent \(\zeta\) Altitude component; its Akaike weight is
\(0.993\). The full three-dimensional ifOU model
has \(\Delta\mathrm{AIC}_{\rm cal}=13.278\) relative to the
overall complete-model minimum and estimates the two correlations
involving Altitude as
\(0.026\) and \(0.042\). Within the restricted OU-position comparison, the full model is
\(2.583\) ordinary-AIC units above the reduced
Longitude--Latitude model. The \(10.695\)-unit
separation between the first two complete models concerns the
Altitude marginal family and favors \(\zeta\) over ifOU.

For Bat~5, the three leading models are independent products. The
minimum combines \(\zeta\) Longitude, Confluent Latitude, and iOU
Altitude and has Akaike weight \(0.105\).
Replacing the Confluent Latitude component by fBM raises the calibrated
AIC by \(0.629\), whereas replacing iOU Altitude by ifOU raises it by
\(0.819\). The best correlated model is
ifOU(Lon,Lat)\(\oplus\)iOU(Alt), with
\(\widehat\rho_{\rm Lon,Lat}=-0.194\). It ranks eighth and lies
\(\Delta\mathrm{AIC}_{\rm cal}=1.986\) above the complete-model
minimum. The complete ranking therefore selects the independent
marginal product, although the correlated alternative remains
competitive.

The selected complete-trajectory models contain nonstationary
position components. For Bats~1--2, both selected components are
scaled fBM processes and all four fitted Hurst parameters exceed
\(1/2\). For Bats~3--4, the selected Longitude--Latitude blocks are
ifOU models. Their fitted Hurst sums are
\(0.84145+0.74152=1.58297\) and
\(0.87475+0.88398=1.75873\),
respectively; together with the nonzero fitted correlations, both
horizontal pairs have long memory for separated cross-coordinate
increments. Bat~3 has an independent fBM Altitude component with
\(\widehat H=0.4501\), Bat~4 an independent \(\zeta\) Altitude
component, and Bat~5 independent \(\zeta\), Confluent, and iOU
components. Thus, the selected models do not imply a common
stationarity or memory regime across coordinates.
\subsubsection{Sensitivity to the complexity penalty and temporal resolution}

To examine the effect of changing only the complexity penalty in the
half-minute ranking, define an implied fit term
\(\ell_{\mathrm{cal}}\) by
\[
-2\ell_{\mathrm{cal}}
:=
\mathrm{AIC}_{\mathrm{cal}}-2k
\]
and compute
\[
\mathrm{AICc}_{\mathrm{cal}}(m)
=
\mathrm{AIC}_{\mathrm{cal}}
+
\frac{2k(k+1)}{m-k-1},
\qquad
\mathrm{BIC}_{\mathrm{cal}}(m)
=
-2\ell_{\mathrm{cal}}+k\log m,
\]
where \(m=n-1\) is the number of retained transitions and \(k\) is
the complete-model parameter count. These quantities are
penalty-sensitivity scores derived from
\(\mathrm{AIC}_{\mathrm{cal}}\); they are not presented as
small-sample criteria newly derived for the irregular multivariate
Gaussian model. The AICc correction was developed for regression
and autoregressive settings \citep{HurvichTsai1989}. Here each
observation time is counted once, rather than treating simultaneous
coordinates as independent replicates. The BIC-type score replaces
the AIC penalty by \(k\log m\) \citep{Schwarz1978}.

\begin{table}[H]
\centering
\scriptsize
\setlength{\tabcolsep}{3.5pt}
\begin{tabular}{crlrrrr}
\toprule
Bat & \(m\) & \(\mathrm{AIC}_{\rm cal}\)-selected model & \(k\) & \(\mathrm{AIC}_{\rm cal}\) & \(\mathrm{AICc}_{\rm cal}(m)\) & \(\mathrm{BIC}_{\rm cal}(m)\)\\
\midrule
1 & 46 & fBM(Lon)\(\oplus\)fBM(Lat) & 4 & -532.067 & -531.091 & -524.752\\
2 & 44 & fBM(Lon)\(\oplus\)fBM(Lat) & 4 & -555.663 & -554.637 & -548.526\\
3 & 61 & ifOU(Lon,Lat)\(\oplus\)fBM(Alt) & 9 & -1585.780 & -1582.251 & -1566.782\\
4 & 240 & ifOU(Lon,Lat)\(\oplus\)\(\zeta\)(Alt) & 9 & -6417.739 & -6416.956 & -6386.413\\
5 & 164 & \(\zeta\)(Lon)\(\oplus\)Confluent(Lat)\(\oplus\)iOU(Alt) & 8 & -4299.775 & -4298.846 & -4274.976\\
\bottomrule
\end{tabular}
\captionsetup{justification=raggedright,singlelinecheck=false}
\caption{Calibrated-AIC winners and their penalty-sensitivity scores at
half-minute resolution, with \(m=n-1\) retained transitions.
\(\mathrm{AICc}_{\rm cal}\) selects the same rows.
\(\mathrm{BIC}_{\rm cal}\) selects the same rows for Bats~1--4 and
selects \(\zeta\)(Lon)\(\oplus\)fBM(Lat)\(\oplus\)iOU(Alt) for Bat~5, with
\(\mathrm{BIC}_{\rm cal}=-4280.548\).}
\label{tab:criterion-sensitivity-main}
\end{table}

At half-minute resolution, \(\mathrm{AIC}_{\rm cal}\) and
\(\mathrm{AICc}_{\rm cal}\) select the same model for every
individual. \(\mathrm{BIC}_{\rm cal}\) differs only for the Bat~5
Latitude marginal, where it favors fBM over Confluent. The data were also aggregated into nonoverlapping
one- and two-minute intervals, and the complete candidate set was
refitted at each resolution. Because all coarser-resolution candidates
use the same likelihood implementation, their ordinary AIC, AICc, and
BIC values are compared directly. Complete rankings are reported in
Tables~\ref{tab:resolution-top3}--\ref{tab:criterion-sensitivity-supp} of Appendix~\ref{ApenD}.

Under the global complete-model criteria, Bats~1--2 retain
independent horizontal products at both coarser resolutions. Bat~1
selects fBM for both coordinates at one minute; at two minutes AIC and
AICc select fBM Longitude with Confluent Latitude, whereas BIC retains
fBM for both coordinates. Bat~2 retains fBM for both coordinates under
every resolution and penalty. Bat~3 retains the
correlated Longitude--Latitude block with independent fBM Altitude;
its iOU Altitude counterpart is only \(0.130\) and
\(0.045\) AIC units higher at one and two minutes. Bat~4
retains the horizontal block, while its selected Altitude family
changes from \(\zeta\) at half a minute to ifOU under one-minute AIC
and to iOU under one-minute AICc and BIC and under all three two-minute
criteria. The one-minute ifOU and iOU models differ by only
\(0.053\) AIC units. For Bat~5, the global complete-model minimum
remains an independent product at all three resolutions; its Latitude
family changes from Confluent at half a minute to \(\zeta\) at one and
two minutes. The global complete-model pattern is stable under
aggregation, whereas several marginal-family choices remain sensitive
to resolution or penalty. This statement concerns joint marginal-and-dependence selection
and does not replace a direct assessment of the correlation term. The
following subsection therefore returns to the half-minute records and
separates horizontal marginal fit from cross-coordinate dependence.

\subsubsection{Separating horizontal marginal fit from cross-coordinate dependence}
\label{subsub:horizontal-dependence-isolation}

The complete ranking in Table~\ref{tab:topthree-main} compares marginal
covariance families and dependence structures simultaneously. This is
the relevant comparison for selecting a complete covariance model, but
it does not isolate the contribution of a cross-coordinate correlation.
The independent class may combine any of the seven marginal families,
whereas correlated blocks are available only within the iOU and ifOU
families. Consequently, an independent product can attain a smaller
global criterion score because it describes one or both marginal
coordinates better, even when introducing correlation improves fit
within a matched temporal family. A global independent minimum should therefore not be
interpreted, by itself, as evidence that
\(\rho_{\rm Lon,Lat}=0\).

We next compare the bivariate Longitude--Latitude records for all five
bats. For Bats~3--5, Altitude is removed from both the response and
the score. For each individual, the candidate set contains
the \(7^2=49\) independent products of the marginal families and the
separately maximized correlated iOU and ifOU blocks. All 51 candidates
use the same two coordinates and observation times. At half-minute
resolution, the calibrated score is defined as in Section~\ref{sub:model-fitting-stage}, but the
reference sum contains only Longitude and Latitude:
\[
 \operatorname{AIC}^{(2D)}_{\rm cal}(C)
 =
 \sum_{j\in\{\mathrm{Lon},\mathrm{Lat}\}}
 \operatorname{AIC}^{\rm ref}_{j}(M_j)
 +
 \left\{
 \operatorname{AIC}^{(2D)}_{g}(C)
 -
 \operatorname{AIC}^{(2D)}_{g}\!\left(I(C)\right)
 \right\}.
\]
Here \(I(C)\) is the independently fitted product with the same
coordinate-specific families as \(C\). For an independent candidate,
\(C=I(C)\) and the term in braces is zero. Each correlated block is
maximized over its own marginal and dependence parameters; no estimates
are inherited from an independent fit.

{\scriptsize
\setlength{\tabcolsep}{3.0pt}
\renewcommand{\arraystretch}{1.04}
\begin{longtable}{@{}cc>{\raggedright\arraybackslash}p{.285\textwidth}crrr>{\raggedright\arraybackslash}p{.185\textwidth}@{}}
\caption{Five lowest calibrated-AIC models for the bivariate
Longitude--Latitude records of Bats~1--5 at half-minute resolution.
For Bats~3--5, Altitude is excluded from both the response and the
score. Each individual is compared over 49 independent marginal
products and two separately maximized correlated blocks. Akaike weights
are normalized over the 51 candidates fitted to that individual.}
\label{tab:horizontal-top5-main}\\
\toprule
Bat & Rank & Horizontal model & \(k\) & \(\mathrm{AIC}^{(2D)}_{\rm cal}\) & \(\Delta\mathrm{AIC}^{(2D)}_{\rm cal}\) & \(w\) & Active correlation\\
\midrule
\endfirsthead
\multicolumn{8}{c}{\tablename\ \thetable\ (continued)}\\
\toprule
Bat & Rank & Horizontal model & \(k\) & \(\mathrm{AIC}^{(2D)}_{\rm cal}\) & \(\Delta\mathrm{AIC}^{(2D)}_{\rm cal}\) & \(w\) & Active correlation\\
\midrule
\endhead
\midrule
\multicolumn{8}{r}{Continued on next page}\\
\endfoot
\bottomrule
\endlastfoot
1 & 1 & fBM(Lon)\(\oplus\)fBM(Lat) & 4 & -532.067 & 0.000 & 0.500 & --\\*
  & 2 & ifOU(Lon)\(\oplus\)fBM(Lat) & 5 & -530.042 & 2.025 & 0.182 & --\\*
  & 3 & fBM(Lon)\(\oplus\)ifOU(Lat) & 5 & -530.024 & 2.043 & 0.180 & --\\*
  & 4 & ifOU(Lon)\(\oplus\)ifOU(Lat) & 6 & -527.999 & 4.068 & 0.065 & --\\*
  & 5 & ifOU(Lon,Lat) & 7 & -526.848 & 5.219 & 0.037 & \(\widehat\rho_{\rm Lon,Lat}=-0.147\)\\
\midrule
2 & 1 & fBM(Lon)\(\oplus\)fBM(Lat) & 4 & -555.663 & 0.000 & 0.465 & --\\*
  & 2 & fBM(Lon)\(\oplus\)ifOU(Lat) & 5 & -553.641 & 2.022 & 0.169 & --\\*
  & 3 & ifOU(Lon)\(\oplus\)fBM(Lat) & 5 & -553.640 & 2.023 & 0.169 & --\\*
  & 4 & ifOU(Lon)\(\oplus\)ifOU(Lat) & 6 & -551.618 & 4.045 & 0.062 & --\\*
  & 5 & ifOU(Lon,Lat) & 7 & -551.182 & 4.481 & 0.049 & \(\widehat\rho_{\rm Lon,Lat}=-0.211\)\\
\midrule
3 & 1 & ifOU(Lon,Lat) & 7 & -1173.942 & 0.000 & 0.99996 & \(\widehat\rho_{\rm Lon,Lat}=0.639\)\\*
  & 2 & iOU(Lon,Lat) & 5 & -1151.519 & 22.423 & \(1.35\times10^{-5}\) & \(\widehat\rho_{\rm Lon,Lat}=0.659\)\\*
  & 3 & \(\zeta\)(Lon)\(\oplus\)ifOU(Lat) & 5 & -1151.147 & 22.795 & \(1.12\times10^{-5}\) & --\\*
  & 4 & \(\zeta\)(Lon)\(\oplus\)iOU(Lat) & 4 & -1150.335 & 23.608 & \(7.48\times10^{-6}\) & --\\*
  & 5 & fBM(Lon)\(\oplus\)ifOU(Lat) & 5 & -1148.490 & 25.452 & \(2.97\times10^{-6}\) & --\\
\midrule
4 & 1 & ifOU(Lon,Lat) & 7 & -4885.600 & 0.000 & \(>0.999999999\) & \(\widehat\rho_{\rm Lon,Lat}=0.579\)\\*
  & 2 & fBM(Lon)\(\oplus\)\(\zeta\)(Lat) & 4 & -4841.955 & 43.646 & \(3.33\times10^{-10}\) & --\\*
  & 3 & ifOU(Lon)\(\oplus\)\(\zeta\)(Lat) & 5 & -4839.926 & 45.675 & \(1.21\times10^{-10}\) & --\\*
  & 4 & Confluent(Lon)\(\oplus\)\(\zeta\)(Lat) & 6 & -4831.673 & 53.927 & \(1.95\times10^{-12}\) & --\\*
  & 5 & \(\zeta\)(Lon)\(\oplus\)\(\zeta\)(Lat) & 3 & -4825.237 & 60.363 & \(7.80\times10^{-14}\) & --\\
\midrule
5 & 1 & \(\zeta\)(Lon)\(\oplus\)Confluent(Lat) & 6 & -3238.106 & 0.000 & 0.212 & --\\*
  & 2 & \(\zeta\)(Lon)\(\oplus\)fBM(Lat) & 4 & -3237.477 & 0.629 & 0.155 & --\\*
  & 3 & Confluent(Lon)\(\oplus\)Confluent(Lat) & 8 & -3237.116 & 0.990 & 0.129 & --\\*
  & 4 & Confluent(Lon)\(\oplus\)fBM(Lat) & 6 & -3236.488 & 1.618 & 0.094 & --\\*
  & 5 & ifOU(Lon,Lat) & 7 & -3236.119 & 1.986 & 0.078 & \(\widehat\rho_{\rm Lon,Lat}=-0.194\)\\\end{longtable}
}

For Bats~1--2, the independent fBM product has the smallest
horizontal calibrated score. The correlated ifOU block ranks fifth,
with \(\Delta\mathrm{AIC}^{(2D)}_{\rm cal}\) equal to \(5.219\) and
\(4.481\), respectively. For Bats~3--4, by contrast, the correlated
ifOU block is the clear horizontal minimum: its
\(\Delta\mathrm{AIC}^{(2D)}_{\rm cal}\) values relative to the
second-ranked model are \(22.423\) and \(43.646\), respectively. For
Bat~5, the independent
\(\zeta\)(Longitude)\(\oplus\)Confluent(Latitude) product has the
smallest horizontal calibrated score, but the correlated ifOU block
remains competitive: it ranks fifth with
\(\Delta\mathrm{AIC}^{(2D)}_{\rm cal}=1.986\).

The horizontal results are consistent with the descriptive evidence
in the sense that Longitude--Latitude is the principal empirical
association for Bats~3--5. Agreement is strongest for Bats~3--4,
where the correlated block is preferred by a wide
\(\mathrm{AIC}^{(2D)}_{\rm cal}\) margin. For Bat~5, the empirical
rate correlations indicate horizontal association and the correlated
model remains competitive, with
\(\Delta\mathrm{AIC}^{(2D)}_{\rm cal}=1.986\).

A more direct assessment of the dependence term is obtained by holding
the temporal family comparable. Table~\ref{tab:3daic} of Appendix~\ref{ApenD} restricts
the analysis to the OU-position family, profiles the Hurst parameters,
and maximizes each dependence class separately. Within that comparison,
the Longitude--Latitude class has the smallest ordinary AIC and AICc
for Bats~3--5. Relative to independent coordinates, its ordinary-AIC
improvements are \(27.445\), \(91.699\), and \(5.226\) units, and its
AICc improvements are \(26.684\), \(91.521\), and \(4.986\) units,
respectively. The evidence is therefore strong for Bats~3--4 and more
limited, but still material within the matched family, for Bat~5.

For Bat~3, the full three-dimensional ifOU model remains competitive
under both comparisons: it has
\(\Delta\mathrm{AIC}_{\rm cal}=1.495\) relative to the complete
seven-family minimum in Table~\ref{tab:topthree-main}, and
\(\Delta\mathrm{AIC}=0.890\) and
\(\Delta\mathrm{AICc}=3.860\) relative to the reduced OU-position
minimum in Table~\ref{tab:3daic} of Appendix~\ref{ApenD}.

Table~\ref{tab:3daic} also shows that the dependence is concentrated in the
horizontal pair. The separately fitted Longitude--Altitude and
Latitude--Altitude blocks are well above the Longitude--Latitude class
for Bats~3--5. For Bats~4--5, the fitted correlations involving
Altitude are near zero. This pattern
agrees with the descriptive displacement-rate summaries, which show
substantially smaller associations involving Altitude.

The complete, horizontal-only, and matched-family comparisons thus
answer different questions. At half-minute resolution, the complete
and horizontal rankings use calibrated scores, whereas Table~\ref{tab:3daic} uses ordinary AIC and AICc from the common multivariate
likelihood. A global criterion minimum selects a joint
marginal-and-dependence specification; it does not, by itself, measure
the contribution of a correlation parameter when competing models
also differ in marginal covariance family. The horizontal table and
the matched-family comparison separate these two sources of fit.

\section{Conclusions}\label{SSec4}

We introduced the multidimensional integral fractional
Ornstein--Uhlenbeck process as an extension of the coordinatewise
ifOU model to dependent multivariate trajectories. Coordinate-specific
damping, scale, and Hurst parameters are retained, while dependence
is introduced through a well-balanced multivariate fractional Brownian
motion. We established covariance validity, characterized the
admissible correlation region, and derived the large-lag behavior of
the position cross-covariances and separated increments. The resulting
finite-dimensional Gaussian distribution provides a basis for
likelihood inference, simulation, and conditional velocity
reconstruction from observed positions.

The numerical experiments illustrate joint estimation of the complete
marginal and dependence parameter vector and examine correlation
estimation under repeated sampling. Moderate correlations are estimated
with small empirical bias in the considered designs, whereas weak
correlations remain less precisely estimated. Under fixed-horizon grid
refinement, both the RMSE and the likelihood-ratio interval width
decrease as the observation grid becomes finer. These results describe
finite-sample performance for the stated simulation settings and do not
establish asymptotic properties of the joint estimator.

In the bat application, the horizontal comparison favors independent
products for Bats~1--2, strongly favors a correlated ifOU
Longitude--Latitude model for Bats~3--4, and retains that correlated
model as a competitive alternative for Bat~5. Within the matched
OU-position comparison, the Longitude--Latitude dependence class also
minimizes AIC and AICc for Bats~3--5. Evidence for dependence involving
Altitude is weaker, although the full three-dimensional ifOU model
remains competitive for Bat~3. The selected marginal covariance
families vary across coordinates and individuals, indicating that
marginal temporal structure and cross-coordinate dependence should be
assessed separately rather than imposed uniformly.

\section*{Data Availability Statement}
The data supporting the findings of this study are publicly available as open data at \url{https://doi.org/10.5441/001/1.5d736bf0}.

\section*{Conflict of Interest Statement}
The authors declare no conflicts of interest.

\section*{Funding Statement}
This research received no specific grant from any funding agency in the public, commercial, or not-for-profit sectors.

\clearpage
\appendix
\makeatletter
\@addtoreset{figure}{section}
\@addtoreset{table}{section}
\@addtoreset{equation}{section}
\makeatother
\renewcommand{\thefigure}{\thesection.\arabic{figure}}
\renewcommand{\thetable}{\thesection.\arabic{table}}
\renewcommand{\theequation}{\thesection.\arabic{equation}}
\renewcommand{\theHfigure}{\thesection.\arabic{figure}}
\renewcommand{\theHtable}{\thesection.\arabic{table}}
\renewcommand{\theHequation}{\thesection.\arabic{equation}}
\section{Proofs of the covariance and dependence results}\label{ApenA}

\textbf{Proof of Proposition 1.} Fix \(s,t\geq0\) and define
\[
g_i(s,u):=\sigma_i e^{\beta_i(u-s)}\mathbf{1}_{[0,s]}(u),
\qquad i=1,\ldots,p.
\]
Because \(K=(K_{i,j})_{i,j=1}^p\) is a covariance kernel, there exists a centered second-order process \(W=(W_1,\ldots,W_p)^\intercal\) with
\(\operatorname{cov}\{W_i(u),W_j(v)\}=K_{i,j}(u,v)\). By Theorem~2.1 of \citet{Lavancier}, \(K\) is bounded on compact rectangles. Hence, for each fixed \(s\),
\[
\int_0^s \big\|g_i(s,u)W_i(u)\big\|_{L^2}\,du
\leq |\sigma_i|s\sup_{0\leq u\leq s}K_{i,i}(u,u)^{1/2}<\infty.
\]
The mean-square integral
\[
Y_i(s):=\int_0^s g_i(s,u)W_i(u)\,du
\]
is therefore well defined. Moreover,
\[
\int_0^s\!\!\int_0^t
|g_i(s,u)|\,|K_{i,j}(u,v)|\,|g_j(t,v)|\,dv\,du<\infty,
\]
so Fubini's theorem gives
\[
\operatorname{cov}\{Y_i(s),Y_j(t)\}
=
\int_0^s\!\!\int_0^t
g_i(s,u)K_{i,j}(u,v)g_j(t,v)\,dv\,du
=
R_{i,j}(s,t).
\]
Thus \(R\) is the covariance kernel of the centered process \(Y=(Y_1,\ldots,Y_p)^\intercal\). In particular,
\(R_{i,j}(s,t)=R_{j,i}(t,s)\), and for arbitrary times \(s_1,\ldots,s_n\) and vectors \(a^{(1)},\ldots,a^{(n)}\in\mathbb{R}^p\),
\[
\sum_{k,\ell=1}^n
\big\langle a^{(k)},R(s_k,s_\ell)a^{(\ell)}\big\rangle
=
\operatorname{var}\!\left(
\sum_{k=1}^n\sum_{i=1}^p a_i^{(k)}Y_i(s_k)
\right)\geq0.
\]
Therefore \(R\) is a matrix-valued covariance kernel on \(\mathbb{R}_+\times\mathbb{R}_+\).
\(\hfill\square\)

\medskip
\textbf{Proof of Proposition 2.}
Assume $K_{i,j}(u,v)=K_{i,j}(v,u)$ and
$\beta_i=\beta_j=\beta$. Interchanging $u$ and $v$ in the defining integral yields
\[
\begin{split}
R_{i,j}(s,t)
&=\sigma_i\sigma_j\int_0^s\!\int_0^t
 e^{-\beta(s-u)}K_{i,j}(u,v)e^{-\beta(t-v)}\,dv\,du\\
&=\sigma_i\sigma_j\int_0^t\!\int_0^s
 e^{-\beta(t-v)}K_{i,j}(v,u)e^{-\beta(s-u)}\,du\,dv
=R_{i,j}(t,s).
\end{split}
\]
\hfill$\square$

\medskip
\textbf{Proof of Proposition 3.}
Write
\[
h=H_i+H_j\in(0,2),
\qquad
c_{i,j}=\frac{\sigma_i\sigma_j\rho_{i,j}}{2}.
\]
For fixed $a,b\ge0$, the covariance formula in the main text and the change of variables
$w=T+b-v$ give
\begin{equation}\label{eq:large-lag-start}
R_{i,j}(a,T+b)
=c_{i,j}\int_0^a e^{-\beta_i(a-u)}
\int_0^{T+b}e^{-\beta_jw}G_{T,b}(u,w)\,dw\,du,
\end{equation}
where
\[
G_{T,b}(u,w)
=u^h+(T+b-w)^h-\lvert T+b-w-u\rvert^h.
\]
Fix $a$ and $b$, and choose $T_0$ such that $T/2+b-a\ge T/4>0$ for all
$T\ge T_0$. Decompose
\[
R_{i,j}(a,T+b)=M_{T,b}(a)+E_{T,b}(a),
\]
where $M_{T,b}(a)$ is the contribution of $0\le w\le T/2$ in
\eqref{eq:large-lag-start}, and $E_{T,b}(a)$ is the contribution of
$T/2<w\le T+b$. For $u\in[0,a]$ and $0\le w\le T+b$,
\[
\lvert G_{T,b}(u,w)\rvert
\le a^h+(T+b)^h+(T+b+a)^h
\le C_{a,b,h}(1+T^h).
\]
Consequently,
\begin{equation}\label{eq:prop3-tail}
\lvert E_{T,b}(a)\rvert
\le C(1+T^h)\int_{T/2}^{\infty}e^{-\beta_jw}\,dw
\le C'(1+T^h)e^{-\beta_jT/2}.
\end{equation}
In particular, $E_{T,b}(a)\to0$ when $h\le1$, and
$T^{1-h}E_{T,b}(a)\to0$ when $h>1$.

Suppose first that $0<h<1$. On $w\le T/2$ and $T\ge T_0$, put
$x=T+b-w\ge u$. Concavity of $x\mapsto x^h$ gives
\[
0\le x^h-(x-u)^h\le u^h.
\]
Thus
\[
\lvert G_{T,b}(u,w)\rvert\le2u^h,
\qquad 0\le w\le T/2,
\]
and, for every fixed $(u,w)$,
$G_{T,b}(u,w)\to u^h$. After extending the integrand by
$\mathbf1_{\{w\le T/2\}}$ to $w\in[0,\infty)$, it is dominated by
$2e^{-\beta_i(a-u)}e^{-\beta_jw}u^h$, which is integrable on
$[0,a]\times[0,\infty)$. The dominated convergence theorem and
\eqref{eq:prop3-tail} therefore yield
\[
R_{i,j}(a,T+b)\longrightarrow
\frac{c_{i,j}}{\beta_j}
\int_0^a e^{-\beta_i(a-u)}u^h\,du.
\]

If $h=1$, then on the common region $w\le T/2$ one has exactly
\[
G_{T,b}(u,w)=u+(T+b-w)-(T+b-w-u)=2u.
\]
Together with \eqref{eq:prop3-tail}, this gives
\[
R_{i,j}(a,T+b)\longrightarrow
\frac{2c_{i,j}}{\beta_j}
\int_0^a e^{-\beta_i(a-u)}u\,du.
\]

Finally, let $1<h<2$. On $w\le T/2$,
\[
(T+b-w)^h-(T+b-w-u)^h
=h\int_0^u(T+b-w-z)^{h-1}\,dz.
\]
Define
\[
A_T(u,w)
:=T^{1-h}u^h
+h\int_0^u\left(\frac{T+b-w-z}{T}\right)^{h-1}dz.
\]
Then
\begin{align*}
T^{1-h}M_{T,b}(a)
&=c_{i,j}\int_0^a\int_0^\infty
 e^{-\beta_i(a-u)}e^{-\beta_jw}
 \mathbf1_{\{w\le T/2\}}A_T(u,w)\,dw\,du.
\end{align*}
For $T\ge T_0$, $w\le T/2$, and $z\le u\le a$,
\[
\frac14\le\frac{T+b-w-z}{T}\le1+\frac{b}{T},
\]
so the factor raised to $h-1$ is uniformly bounded. The displayed integrand is therefore dominated by a constant multiple of
$e^{-\beta_i(a-u)}e^{-\beta_jw}(u^h+u)$, which is integrable. Pointwise,
$T^{1-h}u^h\to0$ and
$((T+b-w-z)/T)^{h-1}\to1$. Dominated convergence and
\eqref{eq:prop3-tail} give
\[
T^{1-h}R_{i,j}(a,T+b)\longrightarrow
\frac{c_{i,j}h}{\beta_j}
\int_0^a e^{-\beta_i(a-u)}u\,du.
\]
Taking $a=s$ and $b=t$, and then applying the substitution $x=s-u$, gives the three formulas stated in Proposition~\ref{Propos2} of the main text. \hfill$\square$

\medskip
\textbf{Proof of Proposition 4.}
Under the well-balanced covariance specification of the main text, fix
\(i,j\in\{1,\ldots,p\}\), and let
\[
h=H_i+H_j\in(0,2),
\qquad
c_{i,j}=\frac{\sigma_i\sigma_j\rho_{i,j}}{2}.
\]
Fix \(0\le s<t\) and \(a\ge0\). For \(b\in\{s,t\}\), define
\[
G_{T,b}(u,w)
=
u^h+(T+b-w)^h-\lvert T+b-w-u\rvert^h.
\]
As in \eqref{eq:large-lag-start}, the change of variables
\(w=T+b-v\) gives
\begin{equation}\label{eq:increment-asymptotic-start}
R_{i,j}(a,T+b)
=
c_{i,j}\int_0^a e^{-\beta_i(a-u)}
\int_0^{T+b}e^{-\beta_jw}G_{T,b}(u,w)\,dw\,du.
\end{equation}
Set
\[
D_T(a)
=
R_{i,j}(a,T+t)-R_{i,j}(a,T+s).
\]
Choose \(T_0\) so large that
\begin{equation}\label{eq:T0-condition}
T/2+s-a\ge T/4>0,
\qquad T\ge T_0.
\end{equation}
For \(b\in\{s,t\}\), write
\begin{align*}
M_{T,b}(a)
&=
c_{i,j}\int_0^a e^{-\beta_i(a-u)}
\int_0^{T/2}e^{-\beta_jw}G_{T,b}(u,w)\,dw\,du,\\
E_{T,b}(a)
&=
c_{i,j}\int_0^a e^{-\beta_i(a-u)}
\int_{T/2}^{T+b}e^{-\beta_jw}G_{T,b}(u,w)\,dw\,du.
\end{align*}
Then
\begin{equation}\label{eq:DT-decomposition}
D_T(a)
=
\{M_{T,t}(a)-M_{T,s}(a)\}
+
\{E_{T,t}(a)-E_{T,s}(a)\}.
\end{equation}
This decomposition is exact; in particular, the different upper
limits \(T+s\) and \(T+t\) are entirely contained in the two
\(E\)-terms.

For \(u\in[0,a]\), \(b\in\{s,t\}\), and \(0\le w\le T+b\),
\[
\lvert G_{T,b}(u,w)\rvert
\le
a^h+(T+t)^h+(T+t+a)^h
\le
C_{a,s,t,h}(1+T^h).
\]
Since \(\beta_i,\beta_j>0\),
\begin{align}
\lvert E_{T,b}(a)\rvert
&\le
C(1+T^h)
\int_0^a e^{-\beta_i(a-u)}\,du
\int_{T/2}^{\infty}e^{-\beta_jw}\,dw \notag\\
&\le
C'(1+T^h)e^{-\beta_jT/2}.
\label{eq:tail-negligible}
\end{align}
Consequently,
\begin{equation}\label{eq:scaled-tail-negligible}
T^{2-h}
\lvert E_{T,t}(a)-E_{T,s}(a)\rvert
\longrightarrow0.
\end{equation}

Suppose first that \(h\neq1\). By \eqref{eq:T0-condition}, whenever
\(T\ge T_0\), \(w\le T/2\), \(b\in[s,t]\), and
\(0\le z\le u\le a\), one has
\[
T+b-w-z
\ge
T/2+s-a
\ge
T/4>0.
\]
Thus, on this region, the absolute values in \(G_{T,s}\) and
\(G_{T,t}\) can be removed. For \(f(x)=x^h\), two applications of
the fundamental theorem of calculus give
\begin{align}
&G_{T,t}(u,w)-G_{T,s}(u,w) \notag\\
&=
(T+t-w)^h-(T+t-w-u)^h
-(T+s-w)^h+(T+s-w-u)^h \notag\\
&=
\int_s^t
\{f'(T+b-w)-f'(T+b-w-u)\}\,db \notag\\
&=
h(h-1)\int_s^t\int_0^u
(T+b-w-z)^{h-2}\,dz\,db.
\label{eq:second-difference-power}
\end{align}

For \(T\ge T_0\), define on
\([0,a]\times[0,\infty)\times[s,t]\times[0,a]\)
\[
Q_T(u,w,b,z)
:=
\begin{cases}
\displaystyle
\left(\frac{T}{T+b-w-z}\right)^{2-h},
& w\le T/2\ \text{and}\ z\le u,\\[8pt]
0,
& \text{otherwise}.
\end{cases}
\]
The first branch is well defined because
\(T+b-w-z\ge T/4>0\). Moreover,
\[
0\le
Q_T(u,w,b,z)
\le
4^{2-h}\mathbf 1_{\{z\le u\}},
\]
and, for every fixed \((u,w,b,z)\),
\[
Q_T(u,w,b,z)
\longrightarrow
\mathbf 1_{\{z\le u\}}.
\]

Substituting \eqref{eq:second-difference-power} into
\(M_{T,t}(a)-M_{T,s}(a)\) and using Fubini's theorem gives
\begin{align}
&T^{2-h}\{M_{T,t}(a)-M_{T,s}(a)\} \notag\\
&=
c_{i,j}h(h-1)
\int_0^a\int_0^\infty\int_s^t\int_0^a
e^{-\beta_i(a-u)}e^{-\beta_jw}
Q_T(u,w,b,z)
\,dz\,db\,dw\,du.
\label{eq:scaled-main-part}
\end{align}
The use of Fubini's theorem is legitimate because the absolute
value of the integrand is bounded by
\[
4^{2-h}
e^{-\beta_i(a-u)}
e^{-\beta_jw}
\mathbf 1_{\{z\le u\}},
\]
which is integrable. Indeed,
\begin{align*}
&\int_0^a\int_0^\infty\int_s^t\int_0^a
e^{-\beta_i(a-u)}e^{-\beta_jw}
\mathbf 1_{\{z\le u\}}
\,dz\,db\,dw\,du\\
&\qquad=
\frac{t-s}{\beta_j}
\int_0^a e^{-\beta_i(a-u)}u\,du
<\infty.
\end{align*}
The dominated convergence theorem therefore yields
\begin{equation}\label{eq:DT-limit}
T^{2-h}\{M_{T,t}(a)-M_{T,s}(a)\}
\longrightarrow
c_{i,j}h(h-1)\frac{t-s}{\beta_j}
\int_0^a e^{-\beta_i(a-u)}u\,du.
\end{equation}
Combining \eqref{eq:DT-decomposition},
\eqref{eq:scaled-tail-negligible}, and \eqref{eq:DT-limit}, we obtain
\[
T^{2-h}D_T(a)
\longrightarrow
c_{i,j}h(h-1)\frac{t-s}{\beta_j}
\int_0^a e^{-\beta_i(a-u)}u\,du.
\]

By covariance bilinearity,
\begin{align*}
&\operatorname{cov}\!\left(
\mu_{H_i}(\nu)-\mu_{H_i}(r),
\mu_{H_j}(T+t)-\mu_{H_j}(T+s)
\right)\\
&\qquad=
D_T(\nu)-D_T(r).
\end{align*}
Define
\[
J_i(a)
:=
\int_0^a e^{-\beta_i(a-u)}u\,du
=
\frac{a}{\beta_i}
-
\frac{1-e^{-\beta_i a}}{\beta_i^2}.
\]
Then
\[
J_i(\nu)-J_i(r)
=
\frac{\nu-r}{\beta_i}
+
\frac{e^{-\beta_i\nu}-e^{-\beta_i r}}{\beta_i^2}.
\]
It follows that, for \(h\neq1\),
\begin{align*}
&T^{2-h}\operatorname{cov}\!\left(
\mu_{H_i}(\nu)-\mu_{H_i}(r),
\mu_{H_j}(T+t)-\mu_{H_j}(T+s)
\right)\\
&\quad\longrightarrow
\frac{\sigma_i\sigma_j\rho_{i,j}h(h-1)}{2}
\frac{t-s}{\beta_j}
\left[
\frac{\nu-r}{\beta_i}
+
\frac{e^{-\beta_i\nu}-e^{-\beta_i r}}{\beta_i^2}
\right].
\end{align*}
Furthermore,
\[
J_i(\nu)-J_i(r)
=
\frac{1}{\beta_i}
\int_r^\nu
\left(1-e^{-\beta_i x}\right)\,dx
>0.
\]
Since \(h\in(0,2)\setminus\{1\}\), the limiting coefficient is
therefore nonzero exactly when \(\rho_{i,j}\neq0\).

It remains to consider \(h=1\). In this case, for \(u,v\ge0\),
\[
K_{i,j}(u,v)
=
\frac{\rho_{i,j}}{2}
\{u+v-\lvert v-u\rvert\}
=
\rho_{i,j}\min(u,v).
\]
Let \(x\ge a\). The covariance representation gives
\[
R_{i,j}(a,x)
=
\sigma_i\sigma_j\rho_{i,j}
\int_0^a e^{-\beta_i(a-u)}
\int_0^x e^{-\beta_j(x-v)}
\min(u,v)\,dv\,du.
\]
For \(0\le u\le a\le x\),
\begin{align*}
\int_0^x e^{-\beta_j(x-v)}\min(u,v)\,dv
&=
\int_0^u e^{-\beta_j(x-v)}v\,dv
+
u\int_u^x e^{-\beta_j(x-v)}\,dv\\
&=
\frac{u}{\beta_j}
-
\frac{e^{-\beta_jx}(e^{\beta_j u}-1)}
{\beta_j^2}.
\end{align*}
Define
\[
A_{i,j}(a)
:=
\int_0^a
e^{-\beta_i(a-u)}
\left(e^{\beta_j u}-1\right)\,du.
\]
For \(T+s\ge a\), the preceding identity gives
\[
D_T(a)
=
\frac{\sigma_i\sigma_j\rho_{i,j}}{\beta_j^2}
e^{-\beta_j(T+s)}
\left(1-e^{-\beta_j(t-s)}\right)
A_{i,j}(a).
\]
Consequently, for \(T+s\ge\nu\),
\begin{align*}
&\operatorname{cov}\!\left(
\mu_{H_i}(\nu)-\mu_{H_i}(r),
\mu_{H_j}(T+t)-\mu_{H_j}(T+s)
\right)\\
&\quad=
\frac{\sigma_i\sigma_j\rho_{i,j}}{\beta_j^2}
e^{-\beta_j(T+s)}
\left(1-e^{-\beta_j(t-s)}\right)
\{A_{i,j}(\nu)-A_{i,j}(r)\}.
\end{align*}
A direct calculation gives
\[
A_{i,j}'(a)
=
\frac{\beta_j}{\beta_i+\beta_j}
\left(e^{\beta_j a}-e^{-\beta_i a}\right),
\]
and hence
\[
A_{i,j}(\nu)-A_{i,j}(r)
=
\frac{\beta_j}{\beta_i+\beta_j}
\int_r^\nu
\left(e^{\beta_jx}-e^{-\beta_i x}\right)\,dx.
\]
Therefore, when \(h=1\) and \(T+s\ge\nu\),
\begin{align*}
&\operatorname{cov}\!\left(
\mu_{H_i}(\nu)-\mu_{H_i}(r),
\mu_{H_j}(T+t)-\mu_{H_j}(T+s)
\right)\\
&\quad=
\frac{\sigma_i\sigma_j\rho_{i,j}
e^{-\beta_j(T+s)}
\left(1-e^{-\beta_j(t-s)}\right)}
{\beta_j(\beta_i+\beta_j)}
\int_r^\nu
\left(e^{\beta_jx}-e^{-\beta_i x}\right)\,dx.
\end{align*}
In particular,
\[
\operatorname{cov}\!\left(
\mu_{H_i}(\nu)-\mu_{H_i}(r),
\mu_{H_j}(T+t)-\mu_{H_j}(T+s)
\right)
=
O\!\left(e^{-\beta_jT}\right).
\]
Thus the displayed scaled limit is zero when \(h=1\), and the
separated-increment covariance decays exponentially. Since
\[
O\!\left(e^{-\beta_jT}\right)
\subset
O\!\left((1+T)e^{-\beta_jT/2}\right),
\]
this also implies the bound stated in the proposition. This
completes the proof. \hfill\(\square\)

\section{Admissible-correlation parametrization}\label{ApenB}

\textbf{Proof of Theorem 2.}
Let
\[
D=\operatorname{diag}\!\left(\sqrt{\Gamma_{1,1}},
\sqrt{\Gamma_{2,2}},\sqrt{\Gamma_{3,3}}\right),
\qquad
r_{i,j}:=\frac{\Gamma_{i,j}\rho_{i,j}}
{\sqrt{\Gamma_{i,i}\Gamma_{j,j}}}.
\]
Then
\[
 C(\boldsymbol\rho):=D^{-1}Q_{H_1,H_2,H_3}(\boldsymbol\rho)D^{-1}
 =\begin{pmatrix}
 1&r_{1,2}&r_{1,3}\\
 r_{1,2}&1&r_{2,3}\\
 r_{1,3}&r_{2,3}&1
 \end{pmatrix}.
\]
Since \(D\) is nonsingular, \(Q(\boldsymbol\rho)\succ0\) if and only if
\(C(\boldsymbol\rho)\succ0\).

For $\boldsymbol z=(z_{1,2},z_{1,3},z_{2,3})\in(-1,1)^3$, define
\[
 r_{1,2}=z_{1,2},\qquad r_{1,3}=z_{1,3},\qquad
 r_{2,3}=z_{1,2}z_{1,3}
 +z_{2,3}\sqrt{1-z_{1,2}^{2}}\sqrt{1-z_{1,3}^{2}}.
\]
The leading principal minors are \(1\), \(1-z_{1,2}^{2}\), and
\[
 \det C=(1-z_{1,2}^{2})(1-z_{1,3}^{2})(1-z_{2,3}^{2})>0.
\]
Thus \(C\succ0\), and hence \(Q\succ0\), by Sylvester's criterion.
Lemma~\ref{lem2.2} of the main text then yields a covariance kernel with
\(K_{i,j}(1,1)=\rho_{i,j}\), so \(|\rho_{i,j}|\leq1\). Equality is excluded
because \(Q\succ0\) is an open condition: a small outward perturbation would
otherwise preserve positive definiteness while violating the covariance
inequality. Therefore \(T(\boldsymbol z)\in\mathcal R_{H_1,H_2,H_3}\).

Conversely, let $\boldsymbol\rho\in\mathcal R_{H_1,H_2,H_3}$. Then $C(\boldsymbol\rho)\succ0$, so $|r_{1,2}|<1$ and $|r_{1,3}|<1$. Set
\[
 z_{1,2}=r_{1,2},\qquad z_{1,3}=r_{1,3},\qquad
 z_{2,3}=\frac{r_{2,3}-r_{1,2}r_{1,3}}
 {\sqrt{1-r_{1,2}^{2}}\sqrt{1-r_{1,3}^{2}}}.
\]
The determinant identity
\[
 \det C=(1-r_{1,2}^{2})(1-r_{1,3}^{2})(1-z_{2,3}^{2})
\]
implies \(|z_{2,3}|<1\). Substitution gives the inverse displayed in
Theorem~2, so \(T\) and \(T^{-1}\) are mutual inverses. \hfill\(\square\)

\section{Numerical experiments}\label{sec2.4}

This appendix reports the numerical experiments summarized in
Section~\ref{sub:numerical_assessment} of the main text. Section~\ref{sub:joint-estimation-supp} considers the joint
twelve-parameter fit to one simulated three-dimensional trajectory
and the conditional reconstruction of its latent velocity.
Section~\ref{sub:replicated-correlation-supp} reports the replicated two- and three-dimensional
correlation-estimation experiment with the marginal parameters fixed
at their generating values. Section~\ref{sub:infill-correlation-supp} reports the fixed-horizon
bivariate correlation-estimation experiment.

\subsection{Joint-estimation illustration for one trajectory}\label{sub:joint-estimation-supp}

A three-dimensional trajectory is simulated on $[0,10]$ with $\Delta=1/30$ and the parameter values listed in Table~\ref{tab:original-sim3d}. All twelve marginal and cross-correlation parameters are estimated jointly from the exact Gaussian likelihood. Figure~\ref{fig:sim3d-supp} shows the simulated path, and Table~\ref{tab:original-sim3d} reports the joint estimates, Gaussian-approximation intervals, and one-parameter likelihood-deviance intervals for this realization. The replicated experiments in Sections~\ref{sub:replicated-correlation-supp}--\ref{sub:infill-correlation-supp} fix the marginal parameters at their true values and study correlation estimation separately.

\begin{figure}[H]
\centering
\includegraphics[width=.78\textwidth]{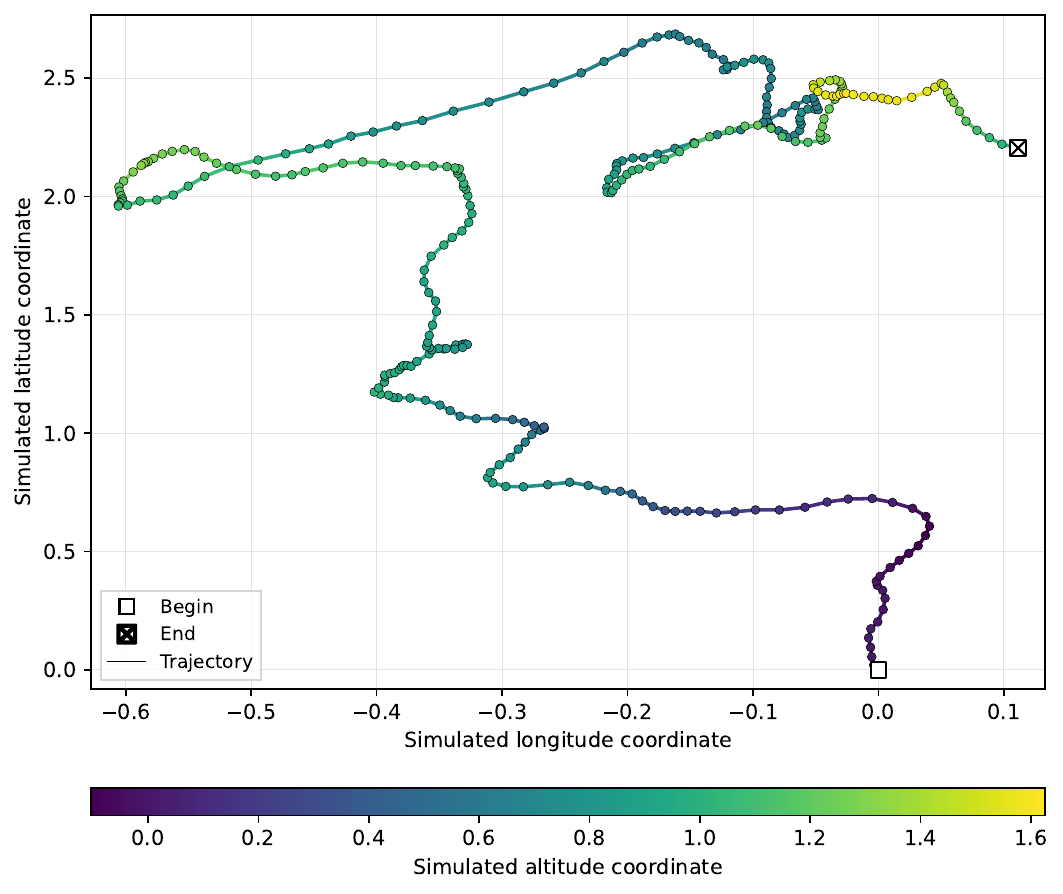}
\caption{Simulated mifOU trajectory in the Longitude--Latitude plane, with Altitude represented by color, for $T=10$ and $\Delta=1/30$. Parameter values and joint estimates are reported in Table~\ref{tab:original-sim3d}.}
\label{fig:sim3d-supp}
\end{figure}

\begin{table}[H]
\centering\scriptsize
\setlength{\tabcolsep}{4pt}
\renewcommand{\arraystretch}{1.08}
\begin{tabularx}{\textwidth}{@{}llrr>{\centering\arraybackslash}X>{\centering\arraybackslash}X@{}}
\toprule
Component & Parameter & True value & MLE & 99\% GA interval & 99\% DEV interval\\
\midrule
Longitude--Latitude & $\rho_{\mathrm{Longitude},\mathrm{Latitude}}$ & 0.2801 & 0.2679 & $[0.1954,0.3405]$ & $[0.1862,0.3324]$\\
Longitude--Altitude & $\rho_{\mathrm{Longitude},\mathrm{Altitude}}$ & -0.2647 & -0.2448 & $[-0.2925,-0.1971]$ & $[-0.2870,-0.1907]$\\
Latitude--Altitude & $\rho_{\mathrm{Latitude},\mathrm{Altitude}}$ & 0.3376 & 0.3610 & $[0.3052,0.4169]$ & $[0.2974,0.4103]$\\
\midrule
Longitude & $\sigma_{\mathrm{Longitude}}$ & 2.0000 & 1.7927 & $[1.6289,1.9565]$ & $[1.6412,1.9707]$\\
Latitude & $\sigma_{\mathrm{Latitude}}$ & 3.0000 & 2.8061 & $[2.5552,3.0569]$ & $[2.5738,3.0784]$\\
Altitude & $\sigma_{\mathrm{Altitude}}$ & 2.5000 & 2.4141 & $[2.2028,2.6255]$ & $[2.2179,2.6429]$\\
Longitude & $\beta_{\mathrm{Longitude}}$ & 7.0000 & 8.5931 & $[6.1380,11.0481]$ & $[6.1444,11.0545]$\\
Latitude & $\beta_{\mathrm{Latitude}}$ & 5.0000 & 5.1493 & $[2.9690,7.3296]$ & $[2.9797,7.3405]$\\
Altitude & $\beta_{\mathrm{Altitude}}$ & 4.2000 & 4.6100 & $[2.8004,6.4195]$ & $[2.8167,6.4360]$\\
Longitude & $H_{\mathrm{Longitude}}$ & 0.8000 & 0.7814 & $[0.7623,0.8004]$ & $[0.7612,0.7994]$\\
Latitude & $H_{\mathrm{Latitude}}$ & 0.6000 & 0.5939 & $[0.5695,0.6183]$ & $[0.5682,0.6172]$\\
Altitude & $H_{\mathrm{Altitude}}$ & 0.2000 & 0.1697 & $[0.1174,0.2220]$ & $[0.1040,0.2145]$\\
\bottomrule
\end{tabularx}
\caption{True values and joint MLEs, together with 99\% Gaussian-approximation (GA) and likelihood-deviance (DEV) intervals, for the simulated three-dimensional trajectory.}
\label{tab:original-sim3d}
\end{table}

The estimates and intervals in Table~\ref{tab:original-sim3d} describe one simulated realization. All twelve marginal and cross-correlation parameters enter the same likelihood optimization, so the experiment directly illustrates full-vector finite-sample estimation.

\begin{figure}[H]
\centering
\includegraphics[width=.82\textwidth]{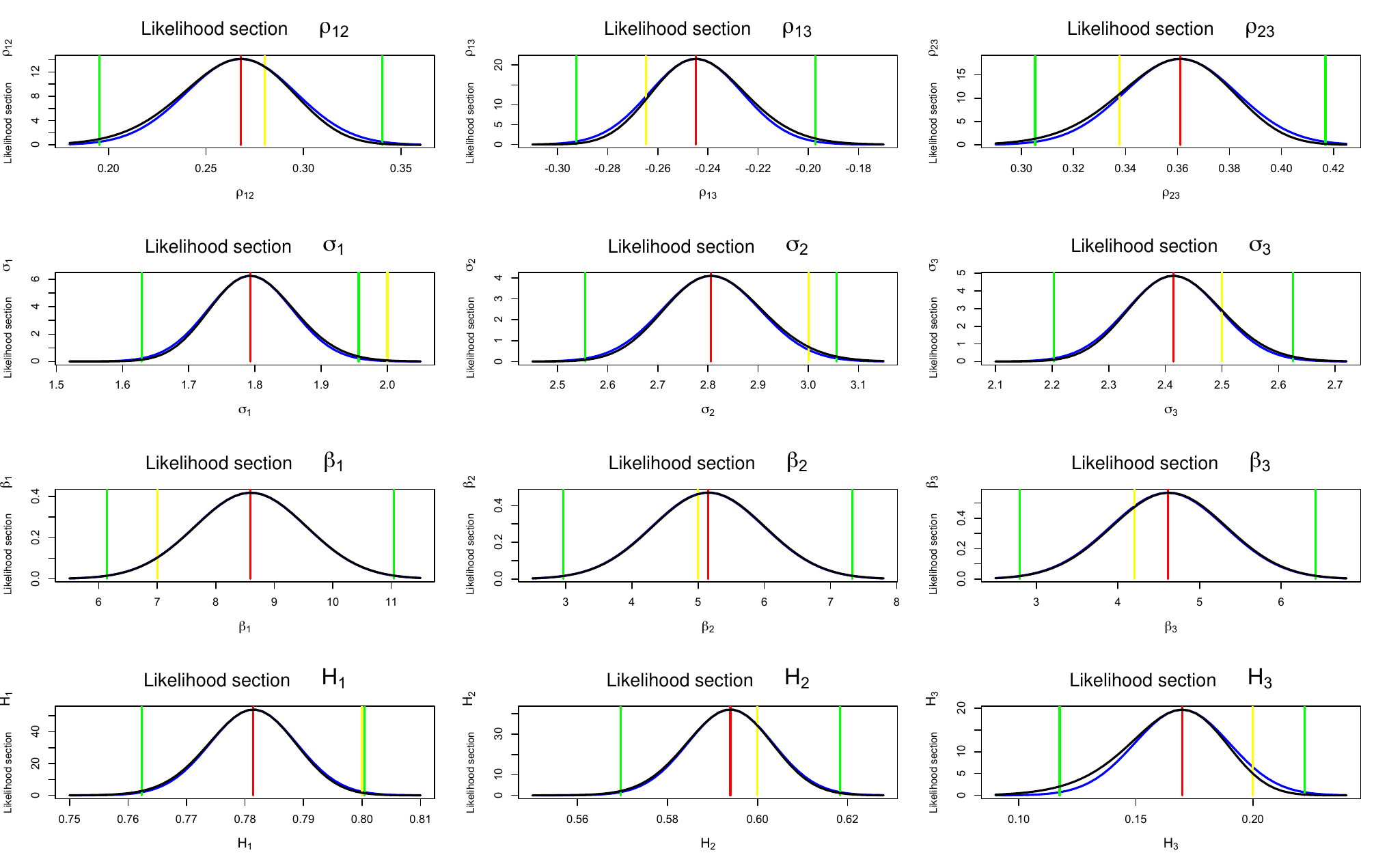}
\caption{One-parameter likelihood sections and local Gaussian approximations for the three-dimensional experiment. Each section varies one parameter while fixing the remaining eleven at the joint MLE.}
\label{PL}
\end{figure}

\begin{figure}[H]
\centering
\includegraphics[width=.82\textwidth]{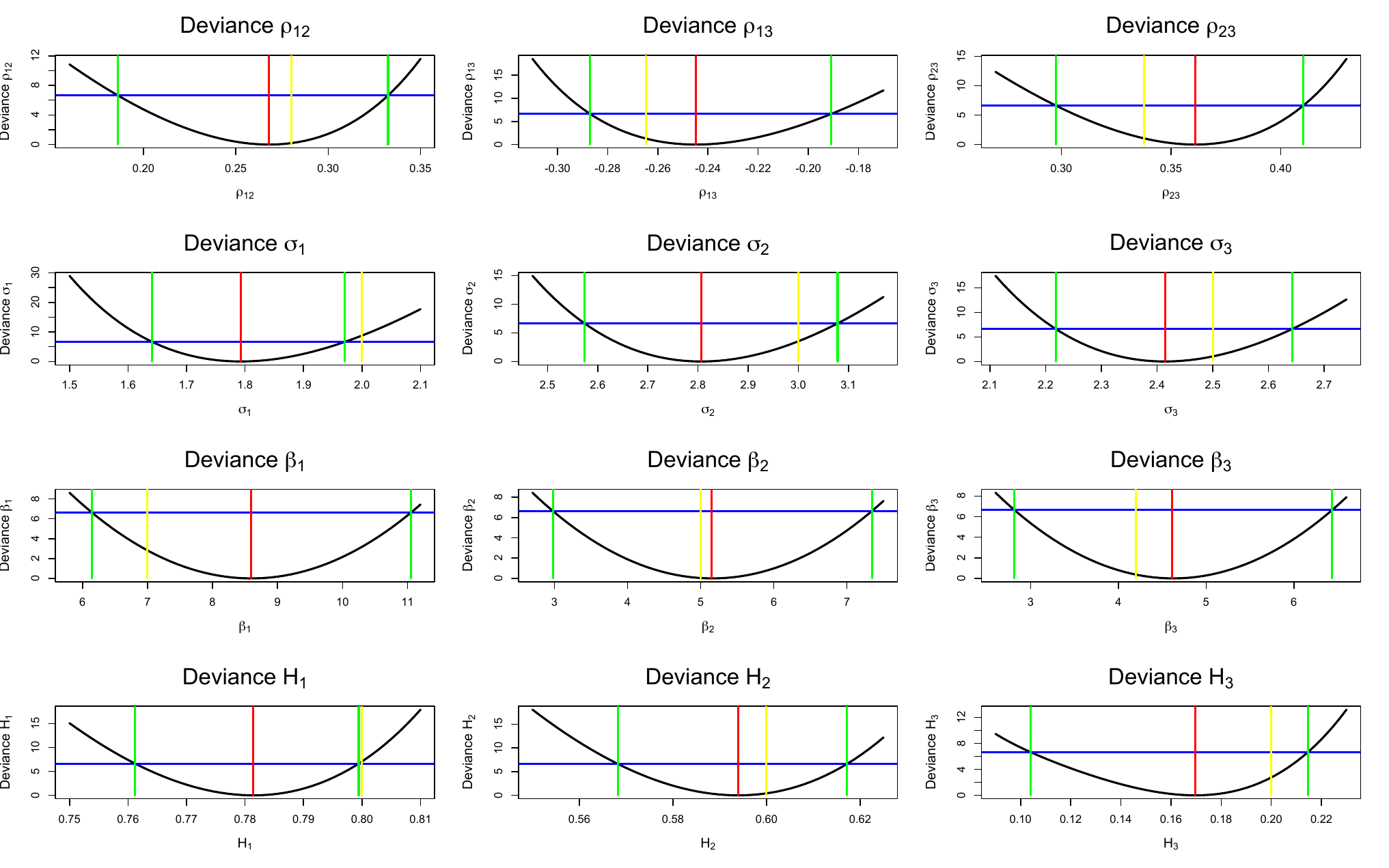}
\caption{One-parameter likelihood-deviance sections for the three-dimensional experiment. Each section varies one parameter while fixing the remaining eleven at the joint MLE.}
\label{PLC2}
\end{figure}

Figures~\ref{PL}--\ref{PLC2} show one-parameter sections rather than profile likelihoods because the remaining eleven parameters are not reoptimized along each curve. They describe the local likelihood shape for this realization.

\subsubsection{Conditional velocity reconstruction}

Algorithm~\ref{alg:velocity} of the main text is applied to the same simulated trajectory used for the joint fit in Table~\ref{tab:original-sim3d}, rather than to a separately generated path. The conditioning record contains 301 positions per coordinate on the grid \(t_j=j/30\), \(j=0,\ldots,300\). Step~1 computes the 900 observed auxiliary variables \(\boldsymbol{\zeta}^{2}\) from these positions. In Step~2, the conditional Gaussian law of the 900-dimensional vector \(\boldsymbol{\zeta}^{1}\) is evaluated with the twelve MLEs in Table~\ref{tab:original-sim3d} fixed. Using seed 20260730, \(B=1{,}000\) independent conditional draws are generated, and each draw is propagated through the recursion in Steps~3--4 with \(v_{H_i}(0)=0\).

For coordinate \(i\) and grid time \(t_j\), the displayed Monte Carlo conditional mean is
\[
\overline v_i(t_j)=\frac{1}{B}\sum_{b=1}^{B}v_i^{(b)}(t_j),
\]
and the pointwise \(95\%\) conditional band is given by the empirical 0.025 and 0.975 quantiles of \(\{v_i^{(b)}(t_j)\}_{b=1}^{B}\). Figure~\ref{fig:sim-velocity-reconstruction} shows the observed conditioning positions in the left column and these conditional velocity summaries in the right column.

\begin{figure}[H]
\centering
\includegraphics[width=\textwidth]{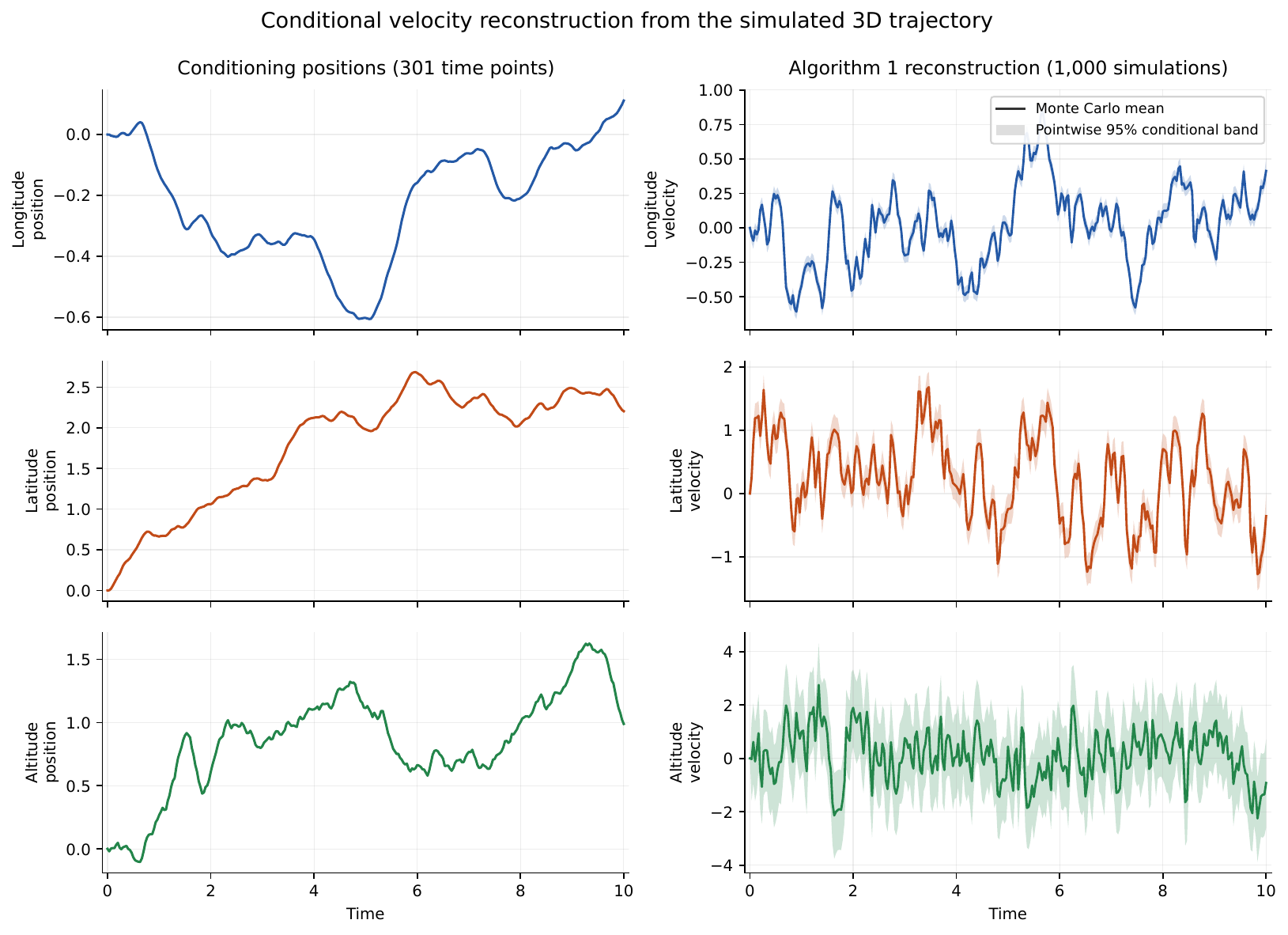}
\caption{Conditional velocity reconstruction from the simulated trajectory used in Table~\ref{tab:original-sim3d}. The left panels show the conditioning positions. The right panels show the Monte Carlo conditional means (solid lines) and pointwise 95\% conditional bands (shaded regions), computed from \(B=1{,}000\) draws under Algorithm~1 with \(v_{H_i}(0)=0\) and the twelve MLEs fixed. The bands are the empirical 0.025 and 0.975 quantiles and exclude parameter-estimation uncertainty.}
\label{fig:sim-velocity-reconstruction}
\end{figure}

The median pointwise band widths are 0.108, 0.475, and 3.050 for Longitude, Latitude, and Altitude, respectively; the corresponding maximum widths are 0.182, 0.675, and 3.318. Thus, the Longitude and Latitude bands are visibly narrower than the Altitude band on their coordinate-specific velocity scales.

\subsection{Replicated correlation-estimation experiment}\label{sub:replicated-correlation-supp}

We examine finite-sample estimation of the cross-coordinate
correlations with the marginal parameters fixed at their generating
values. The two-dimensional designs use
$(\sigma_1,\sigma_2)=(1,1)$ and $(\beta_1,\beta_2)=(3,8)$,
whereas the three-dimensional designs use
$\boldsymbol{\sigma}=(1,1,1)$ and
$\boldsymbol{\beta}=(7,5,4.2)$. The remaining generating parameters
for the four designs are summarized below:

\begin{center}
\small
\setlength{\tabcolsep}{5pt}
\begin{tabular}{lcccc}
\toprule
Design
& $\boldsymbol{H}$
& $\rho_{12}$
& $\rho_{13}$
& $\rho_{23}$\\
\midrule
2D weak
& $(0.55,0.90)$
& $0.0749$
& --
& --\\
2D moderate
& $(0.10,0.90)$
& $-0.2300$
& --
& --\\
3D weak
& $(0.90,0.60,0.20)$
& $0.0784$
& $0$
& $-0.0830$\\
3D moderate
& $(0.80,0.60,0.20)$
& $0.2801$
& $-0.2648$
& $0.3376$\\
\bottomrule
\end{tabular}
\end{center}

The two-dimensional moderate design lies on the boundary
$H_1+H_2=1$ of the separated-increment dependence result. The
three-dimensional moderate design uses the same Hurst configuration
as the three-dimensional experiment in Section~\ref{sub:joint-estimation-supp}.

For each design and each $n\in\{50,100\}$, 50 independent
trajectories are generated on a regular grid with spacing
$\Delta=0.5$. During estimation, all marginal parameters
$(\sigma_i,\beta_i,H_i)$ are held fixed at their generating values.
One canonical dependence parameter is estimated in each
two-dimensional design, whereas the three canonical parameters are
estimated jointly in the three-dimensional designs. Before computing
the Monte Carlo summaries, all estimates are transformed to the
$\rho$ scale.

Figure~\ref{fig:simrho} displays the resulting empirical
distributions, and Table~\ref{tab:simrho} reports the empirical mean,
bias, and RMSE for each coordinate pair. No fitted canonical
parameter exceeded $0.94$ in absolute value.

\begin{figure}[H]
\centering
\includegraphics[width=\textwidth]{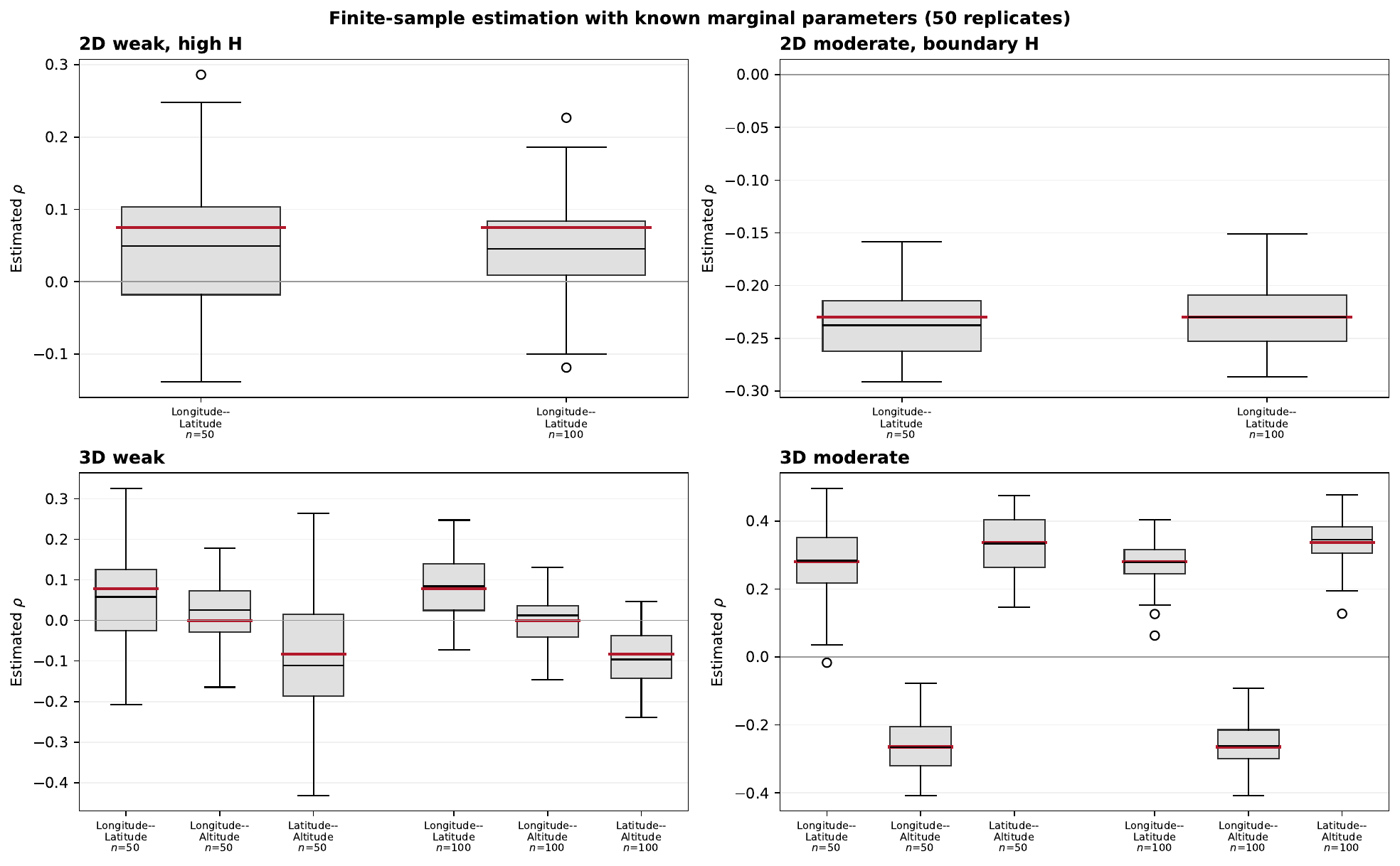}
\caption{Monte Carlo distributions of the cross-coordinate
correlation MLEs. Each boxplot summarizes 50 independent estimates
for the indicated design, coordinate pair, and sample size.
Horizontal segments denote the generating correlations. Marginal
parameters are fixed at their generating values.}
\label{fig:simrho}
\end{figure}

\begin{table}[H]
\centering
\small
\setlength{\tabcolsep}{5pt}
\renewcommand{\arraystretch}{1.06}
\begin{tabularx}{\textwidth}
{@{}>{\raggedright\arraybackslash}Xcrrrr@{}}
\toprule
Coordinate pair
& $n$
& \shortstack{Generating\\$\rho$}
& \shortstack{Mean\\$\widehat{\rho}$}
& Bias
& RMSE\\
\midrule

\multicolumn{6}{@{}l}{%
\textbf{2D weak}\quad
$\boldsymbol{H}=(0.55,0.90)$}\\
Longitude--Latitude & 50
 & 0.0749 & 0.0432 & -0.0318 & 0.1010\\
Longitude--Latitude & 100
 & 0.0749 & 0.0503 & -0.0247 & 0.0758\\

\addlinespace[4pt]
\multicolumn{6}{@{}l}{%
\textbf{2D moderate}\quad
$\boldsymbol{H}=(0.10,0.90)$}\\
Longitude--Latitude & 50
 & -0.2300 & -0.2357 & -0.0056 & 0.0341\\
Longitude--Latitude & 100
 & -0.2300 & -0.2294 & 0.0006 & 0.0276\\

\addlinespace[4pt]
\multicolumn{6}{@{}l}{%
\textbf{3D weak}\quad
$\boldsymbol{H}=(0.90,0.60,0.20)$}\\
Longitude--Latitude & 50
 & 0.0784 & 0.0541 & -0.0242 & 0.1109\\
Longitude--Altitude & 50
 & 0 & 0.0237 & 0.0237 & 0.0794\\
Latitude--Altitude & 50
 & -0.0830 & -0.0955 & -0.0124 & 0.1364\\
\addlinespace[1pt]
Longitude--Latitude & 100
 & 0.0784 & 0.0777 & -0.0006 & 0.0751\\
Longitude--Altitude & 100
 & 0 & 0.0011 & 0.0011 & 0.0554\\
Latitude--Altitude & 100
 & -0.0830 & -0.0942 & -0.0111 & 0.0783\\

\addlinespace[4pt]
\multicolumn{6}{@{}l}{%
\textbf{3D moderate}\quad
$\boldsymbol{H}=(0.80,0.60,0.20)$}\\
Longitude--Latitude & 50
 & 0.2801 & 0.2754 & -0.0048 & 0.1183\\
Longitude--Altitude & 50
 & -0.2648 & -0.2653 & -0.0005 & 0.0783\\
Latitude--Altitude & 50
 & 0.3376 & 0.3307 & -0.0069 & 0.0860\\
\addlinespace[1pt]
Longitude--Latitude & 100
 & 0.2801 & 0.2763 & -0.0038 & 0.0700\\
Longitude--Altitude & 100
 & -0.2648 & -0.2610 & 0.0038 & 0.0572\\
Latitude--Altitude & 100
 & 0.3376 & 0.3402 & 0.0025 & 0.0680\\

\bottomrule
\end{tabularx}
\caption{Empirical summaries of the cross-coordinate correlation
MLEs based on 50 independent trajectories per design and sample
size. Marginal parameters are fixed at their generating values, and
all summaries are reported on the $\rho$ scale.}
\label{tab:simrho}
\end{table}

For every coordinate pair, the RMSE is lower at $n=100$ than at
$n=50$. Because $\Delta$ is fixed, this comparison involves both
more observations and a longer observation horizon. 

These results describe correlation-only estimation conditional on
known marginal parameters. They do not include the additional
uncertainty arising from joint estimation of the marginal and
dependence parameters. The canonical parametrization is used solely
to restrict optimization to the admissible correlation region; all
estimates and Monte Carlo summaries are reported on the $\rho$ scale.

\subsection{Fixed-horizon correlation estimation under grid refinement}\label{sub:infill-correlation-supp}

We examine the estimation of $\rho_{12}$ as the observation grid is
refined while the horizon remains fixed at $T=10$. The marginal
parameters are fixed at their generating values, and only
$\rho_{12}$ is estimated. The grid sequence follows
\citet[Appendix~B.2]{JY}. The simulation design is summarized below.

\begin{center}
\begin{tabular}{ll}
\toprule
Quantity & Value\\
\midrule
Observation horizon
  & $T=10$\\
Marginal parameters
  & $(\sigma_1,\sigma_2)=(2,3)$\\
  & $(\beta_1,\beta_2)=(7,5)$\\
  & $(H_1,H_2)=(0.8,0.6)$\\
True correlations
  & $\rho_{12}\in\{0.1,0.3\}$\\
Sampling intervals
  & $\Delta\in\{1/10,1/20,1/30,1/40,1/50\}$\\
Sample sizes
  & $n\in\{100,200,300,400,500\}$\\
Replicates
  & $100$ for each combination of $\rho_{12}$ and $\Delta$\\
\bottomrule
\end{tabular}
\end{center}

Let ${\bf Y}_{i,n}$ denote coordinate $i$ observed at
$t_k=k\Delta$, $k=1,\ldots,n$, where $n\Delta=T$, and define
\[
{\bf Y}_n
=
\left(
{\bf Y}_{1,n}^{\mathsf T},
{\bf Y}_{2,n}^{\mathsf T}
\right)^{\mathsf T}.
\]
Conditional on the marginal parameters, its covariance matrix is
\[
\boldsymbol\Sigma_n(\rho_{12})
=
\begin{pmatrix}
K_{1,n} & \rho_{12}C_{12,n}\\
\rho_{12}C_{12,n}^{\mathsf T} & K_{2,n}
\end{pmatrix},
\]
where $K_{i,n}$ is the marginal ifOU covariance matrix of coordinate
$i$ and $C_{12,n}$ is the cross-covariance matrix evaluated at unit
correlation.

The admissible correlation map is
$\rho_{12}=c_Hz_{12}$. Setting $B_n=c_HC_{12,n}$ gives
\[
\boldsymbol\Sigma_n(z_{12})
=
\begin{pmatrix}
K_{1,n} & z_{12}B_n\\
z_{12}B_n^{\mathsf T} & K_{2,n}
\end{pmatrix}.
\]
Let $K_{i,n}=L_iL_i^{\mathsf T}$ and
\[
L_1^{-1}B_nL_2^{-\mathsf T}=UDV^{\mathsf T}.
\]
Whitening by $L_1$ and $L_2$ and rotating by $U$ and $V$ transform
the covariance matrix into
\[
\begin{pmatrix}
I_n & z_{12}D\\
z_{12}D & I_n
\end{pmatrix}.
\]
The determinant and quadratic form are therefore evaluated directly
from the singular values in $D$. This transformation is exact, and
the likelihood optimization is one-dimensional in $z_{12}$. The MLE
is subsequently transformed to the $\rho_{12}$ scale.

An approximate $95\%$ likelihood-ratio interval is obtained from the
values of $z_{12}$ satisfying
\[
2\left\{
\ell(\widehat z_{12})-\ell(z_{12})
\right\}
\leq
\chi^2_{1,0.95},
\]
and its endpoints are transformed to the $\rho_{12}$ scale.

Figure~\ref{fig:infillrho-boxplots} shows the Monte Carlo
distributions of the MLEs, and Table~\ref{tab:infillrho} reports the
corresponding estimation and interval summaries.

\begin{figure}[H]
\centering
\includegraphics[width=\textwidth]{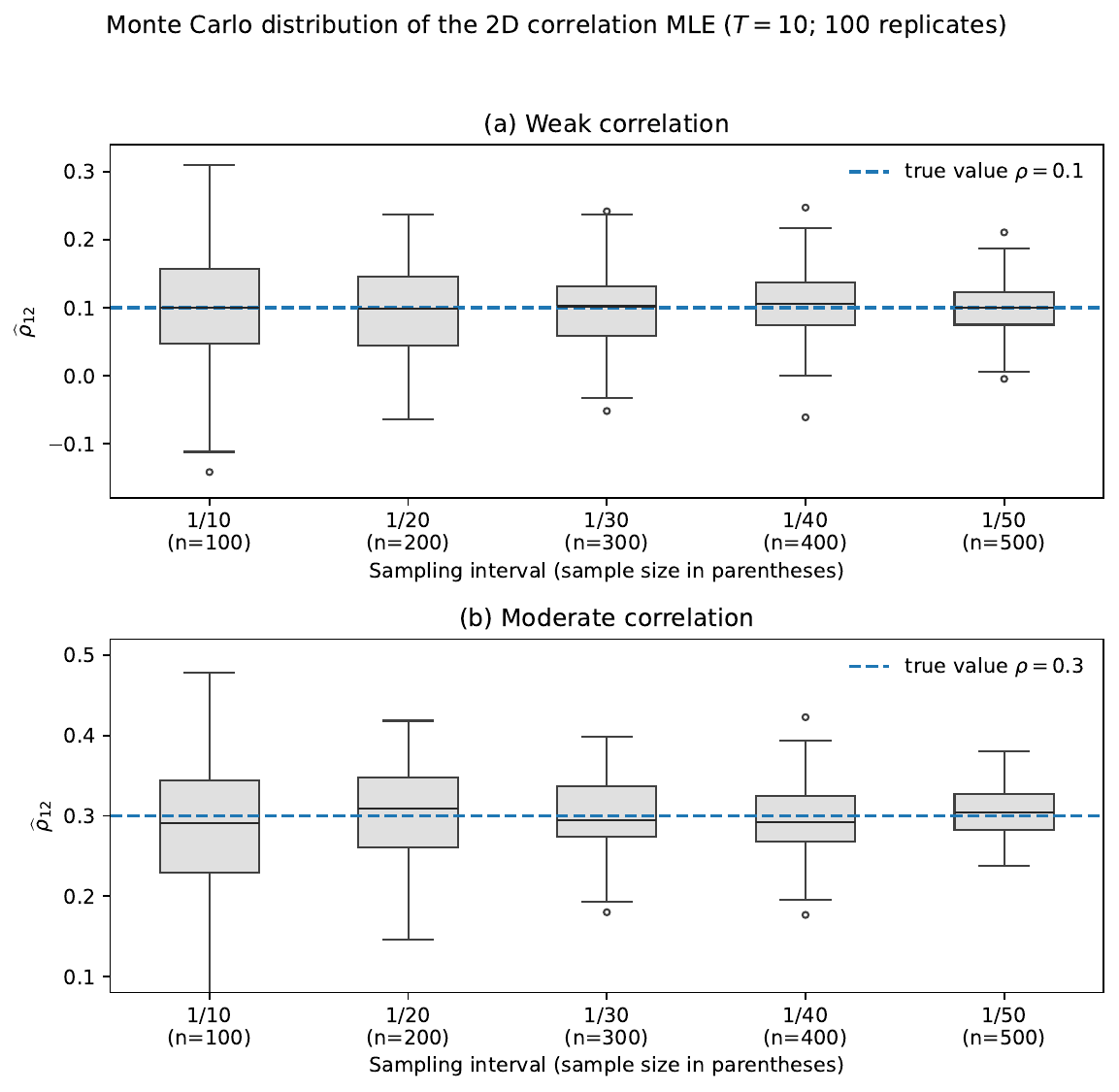}
\caption{Monte Carlo distributions of $\widehat{\rho}_{12}$ under
fixed-horizon grid refinement. Each boxplot summarizes 100
independent trajectories. From left to right, $\Delta$ decreases and
$n$ increases while $n\Delta=T=10$. Horizontal lines indicate the
true correlations.}
\label{fig:infillrho-boxplots}
\end{figure}

\begin{table}[H]
\centering
\scriptsize
\begin{tabular}{llrrrrrr}
\toprule
True $\rho_{12}$ & $\Delta$ & $n$
& Mean $\widehat{\rho}_{12}$
& Bias & RMSE & Coverage & Mean width\\
\midrule
$0.1$ & $1/10$ & 100
& 0.0993 & -0.0007 & 0.0934 & 0.95 & 0.3455\\
      & $1/20$ & 200
& 0.0981 & -0.0019 & 0.0671 & 0.91 & 0.2470\\
      & $1/30$ & 300
& 0.0972 & -0.0028 & 0.0567 & 0.92 & 0.2012\\
      & $1/40$ & 400
& 0.1044 &  0.0044 & 0.0525 & 0.90 & 0.1733\\
      & $1/50$ & 500
& 0.0995 & -0.0005 & 0.0414 & 0.91 & 0.1558\\[1mm]
$0.3$ & $1/10$ & 100
& 0.2911 & -0.0089 & 0.0807 & 0.94 & 0.3100\\
      & $1/20$ & 200
& 0.3036 &  0.0036 & 0.0605 & 0.95 & 0.2138\\
      & $1/30$ & 300
& 0.3011 &  0.0011 & 0.0457 & 0.95 & 0.1755\\
      & $1/40$ & 400
& 0.2955 & -0.0045 & 0.0434 & 0.94 & 0.1539\\
      & $1/50$ & 500
& 0.3067 &  0.0067 & 0.0347 & 0.96 & 0.1348\\
\bottomrule
\end{tabular}
\caption{Monte Carlo summaries for fixed-horizon estimation of
$\rho_{12}$. Each row is based on 100 independent trajectories.
Coverage and mean width refer to the approximate $95\%$
likelihood-ratio intervals.}
\label{tab:infillrho}
\end{table}

For both generating correlations, grid refinement is accompanied by
a monotone reduction in RMSE and mean likelihood-ratio interval
width.

\subsection{Computational scaling}\label{sub:computational-scaling-supp}

The regular-grid recursion reduces the number of numerical integrations
needed to construct the covariance matrix.  For coordinates
$a,b\in\{1,\ldots,p\}$, let $h_{a,b}=H_a+H_b$ and define
\[
 g_{a\mid b}(m)
 =
 \int_0^\Delta e^{-\beta_a u}
 \bigl((m+1)\Delta-u\bigr)^{h_{a,b}}\,du,
 \qquad m=0,\ldots,n-1,
\]
and
\[
 k_{a,b}(d)
 =
 \int_0^\Delta\!\int_0^\Delta
 e^{-\beta_a u}
 \lvert d\Delta+u-v\rvert^{h_{a,b}}
 e^{-\beta_b v}\,dv\,du,
 \qquad d=-(n-1),\ldots,n-1.
\]
Equation~\eqref{covy} of the main text can then be written as
\[
 \operatorname{cov}(y_{a,i},y_{b,j})
 =
 \frac{\rho_{a,b}}{2}
 \left\{
 \frac{1-e^{-\beta_b\Delta}}{\beta_b}g_{a\mid b}(i)
 +
 \frac{1-e^{-\beta_a\Delta}}{\beta_a}g_{b\mid a}(j)
 -k_{a,b}(j-i)
 \right\}.
\]
Thus, for a cross-coordinate block, the first two terms require two
sequences of $n$ one-dimensional integrals and the last term requires
at most $2n-1$ signed-lag double integrals.  For a marginal block,
the two one-dimensional sequences coincide and
$k_{a,a}(d)=k_{a,a}(-d)$, so only $n$ double integrals are required.
The direct evaluation of a separate double integral for each of the
$n^2$ entries is therefore replaced by $O(n)$ numerical integrations
per coordinate-pair block.  Across all marginal and cross-coordinate
blocks, this gives $O(p^2n)$ numerical integrations for fixed-grid
simulation in $p$ dimensions.  This is the multivariate extension of
the regular-grid reduction in \citet[Remark~3.2]{JY}; it applies in
particular to both the two- and three-dimensional models considered in
the article.

This reduction concerns covariance construction.  Filling and storing
the resulting dense $pn\times pn$ matrix still require
$O((pn)^2)$ arithmetic and storage.  The dominant operation in a
general Gaussian simulation or likelihood evaluation is a Cholesky
factorization of that matrix, which requires $O((pn)^3)$ operations.
Joint estimation does not alter these per-evaluation orders, and its
total cost also depends on the optimization dimension, the number of
starting values, and the convergence criteria.  The
admissible-correlation parametrization ensures that every parameter
vector considered during optimization defines a valid covariance
matrix, but it does not change the dense-matrix orders.

\section{Empirical analyses}\label{ApenD}

The empirical analyses below use the half-minute records and
preprocessing described in Section~\ref{sub:empirical-stage} of the main text.
Altitude-based analyses are restricted to Bats~3--5, and missing
Altitude values are not imputed. Likelihood calculations retain
Longitude and Latitude in degrees and express Altitude in kilometers;
the directional and cross-coordinate summaries use the metric
transformation and observed-time-gap normalization defined in the
main text. This appendix reports the complete diagnostic results,
cross-coordinate summaries, and supporting likelihood comparisons
referenced there.
\subsection{Altitude diagnostics and marginal model comparison}
\label{sub:altitude-diagnostics}

The Altitude records of Bats~3--5 are examined using the same classes
of diagnostics previously considered for Longitude and Latitude by
\citet{RamirezEtAl2026}: rolling summaries, sample autocorrelation
functions, augmented Dickey--Fuller (ADF) and KPSS tests, and
detrended fluctuation analysis (DFA). The precise implementation used
for Altitude is specified below. Figures~\ref{fig:altstat}--
\ref{fig:altdfa} report the graphical diagnostics, and
Tables~\ref{tab:altstat}--\ref{tab:altdfa} give the corresponding
ADF, KPSS, and DFA results.

The calculations use the retained observation sequences without
interpolating empty half-minute intervals. Consequently, rolling
windows, autocorrelation lags, deterministic trends in the ADF
regressions, and DFA scales index successive retained observations
rather than equal increments of elapsed time. The rolling summaries
in Figure~\ref{fig:altstat} are centered means and sample variances
calculated from complete windows of 10 retained observations. Sample
autocorrelation functions are displayed through lag 20. For Bat~\(b\)
and differencing order \(d\in\{0,1,2\}\), let
\[
n_{b,d}=N_b-d
\]
denote the length of the corresponding level or differenced series.
The reference bounds for that series are
\[
\pm\frac{1.96}{\sqrt{n_{b,d}}}.
\]

For a tested series of length \(N\), the ADF regression includes an
intercept, a linear trend in the retained-observation index, and
\[
k
=
\left\lfloor (N-1)^{1/3}\right\rfloor
\]
lagged differences. Its null hypothesis is a unit root
\citep{DickeyFuller1979}. The KPSS test uses the level-stationarity
null and
\[
L
=
\min\left\{
N-2,\,
\left\lceil
12\left(\frac{N}{100}\right)^{1/4}
\right\rceil
\right\}
\]
Bartlett-weighted lags \citep{KPSS1992}. The two tests are interpreted
jointly: rejection of the ADF null together with failure to reject
the KPSS null is compatible with level stationarity, whereas failure
to reject both opposing null hypotheses is inconclusive. Values
reported as \(p\leq0.010\) for ADF or \(p\geq0.100\) for KPSS are
bounds of the tabulated \(p\)-value approximations rather than exact
values. Because the retained observations are not equally spaced in
elapsed time, these tests are used as sequence-based diagnostics and
their nominal \(p\)-values are interpreted descriptively.

\begin{figure}[H]
\centering
\includegraphics[width=\textwidth]
{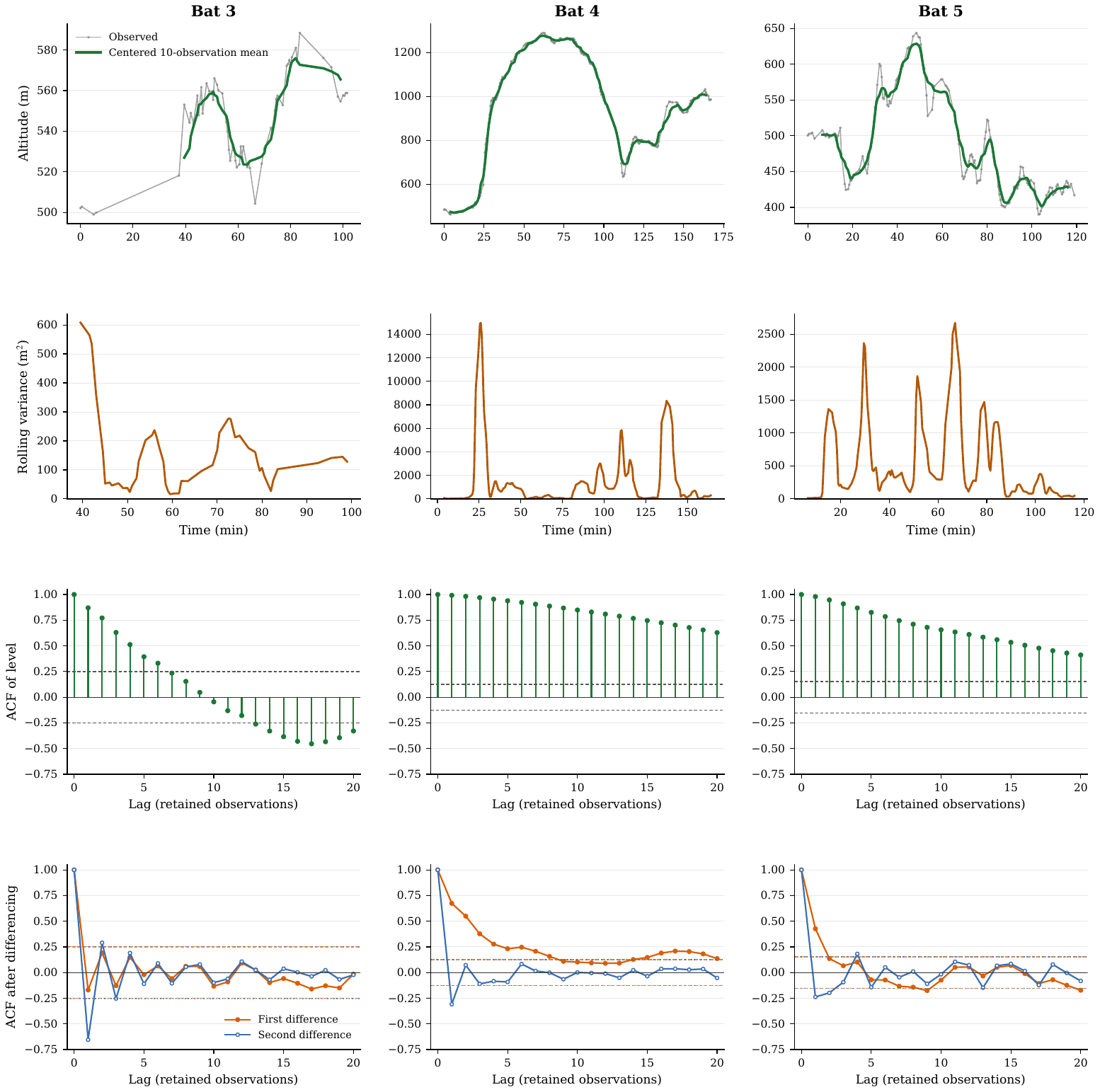}
\caption{Altitude diagnostics for Bats~3--5. The first row shows the
observed Altitude series and centered rolling means based on complete
windows of 10 retained observations; the second row shows the
corresponding centered rolling sample variances. The third row
reports the sample ACFs of the observed levels, and the fourth row
reports the ACFs of the first and second differences. For a series of
differencing order \(d\), the dashed lines are the approximate
pointwise white-noise bounds
\(\pm1.96/\sqrt{n_{b,d}}\), where \(n_{b,d}=N_b-d\). Rolling windows
and autocorrelation lags index successive retained observations;
empty half-minute intervals are not interpolated.}
\label{fig:altstat}
\end{figure}

\begin{table}[H]
\centering
\small
\setlength{\tabcolsep}{6pt}
\renewcommand{\arraystretch}{1.10}
\begin{tabular}{@{}clrrrr@{}}
\toprule
Bat & Series
& ADF statistic & ADF \(p\)
& KPSS statistic & KPSS \(p\)\\
\midrule
3 & Level
& -2.475 & 0.383
& 0.251 & \(\geq0.100\)\\
3 & First difference
& -3.349 & 0.072
& 0.118 & \(\geq0.100\)\\
3 & Second difference
& -8.202 & \(\leq0.010\)
& 0.110 & \(\geq0.100\)\\
\addlinespace[2pt]
4 & Level
& -1.790 & 0.664
& 0.362 & 0.094\\
4 & First difference
& -3.858 & 0.017
& 0.301 & \(\geq0.100\)\\
4 & Second difference
& -7.429 & \(\leq0.010\)
& 0.042 & \(\geq0.100\)\\
\addlinespace[2pt]
5 & Level
& -2.122 & 0.525
& 0.581 & 0.024\\
5 & First difference
& -4.802 & \(\leq0.010\)
& 0.096 & \(\geq0.100\)\\
5 & Second difference
& -6.846 & \(\leq0.010\)
& 0.064 & \(\geq0.100\)\\
\bottomrule
\end{tabular}
\caption{ADF and KPSS results for the Altitude levels and their first
and second differences. The ADF null hypothesis is a unit root,
whereas the KPSS null hypothesis is level stationarity. Values shown
as \(\leq0.010\) or \(\geq0.100\) are bounds of the tabulated
\(p\)-value approximations.}
\label{tab:altstat}
\end{table}

The level ACFs are positive at short lags for all three bats, with
slower decay for Bats~4--5. After first differencing, the Bat~3
autocorrelations at lags \(1,\ldots,20\) remain within the displayed
pointwise bounds. Bat~4 retains positive first-difference
autocorrelation over several lags, whereas Bat~5 retains a positive
lag-one autocorrelation and isolated later values outside the
corresponding bounds. The second-difference autocorrelation at lag
one is negative and outside its bound for all three bats. These
comparisons are descriptive because the bounds are pointwise and the
observation gaps are irregular.

The level tests do not give a common conclusion across the three
Altitude records. For Bats~3--4, neither test rejects its respective
null hypothesis, so the paired results are inconclusive. For Bat~5,
ADF fails to reject a unit root and KPSS rejects level stationarity.
After first differencing, the ADF and KPSS results are compatible
with stationarity for Bats~4--5. The Bat~3 result remains
inconclusive because ADF gives \(p=0.072\). After second
differencing, the two tests are compatible with stationarity for all
three bats. Taken together with the autocorrelation functions, these
results do not justify restricting the marginal Altitude comparison
to stationary covariance families.

For DFA, let \(X_1,\ldots,X_N\) denote an Altitude sequence and
define its centered cumulative profile by
\[
Y(j)
=
\sum_{r=1}^{j}
\left(X_r-\overline X\right),
\qquad
j=1,\ldots,N.
\]
At scale \(s\), let
\[
q_s
=
\left\lfloor\frac{N}{s}\right\rfloor.
\]
The profile is partitioned into \(q_s\) complete nonoverlapping
windows of length \(s\), starting from the beginning of the series.
The partition is then repeated from the end, yielding \(2q_s\)
windows and retaining the observations otherwise excluded when
\(N\) is not divisible by \(s\). Let
\(\mathcal I_{\nu,s}\) denote window \(\nu\), and let
\(\widehat Y_{\nu,s}(j)\) be the least-squares linear trend fitted
within that window. Define
\[
F_{\nu}^{2}(s)
=
\frac{1}{s}
\sum_{j\in\mathcal I_{\nu,s}}
\left\{
Y(j)-\widehat Y_{\nu,s}(j)
\right\}^{2},
\qquad
\nu=1,\ldots,2q_s,
\]
and
\[
F(s)
=
\left\{
\frac{1}{2q_s}
\sum_{\nu=1}^{2q_s}
F_{\nu}^{2}(s)
\right\}^{1/2}.
\]
This is the standard two-direction DFA construction
\citep{Peng1994}.

Up to 12 distinct integer scales are obtained by flooring a
logarithmically spaced grid from \(4\) to
\(\lfloor N/4\rfloor\), with both endpoints included. If \(J\)
distinct scales are available, the regression
\[
\log F(s)=a+\alpha\log s
\]
is fitted over the largest
\(\max\{4,\lceil J/2\rceil\}\) scales. The fitted scales are
\(9,10,11,13,15\) for Bat~3;
\(17,22,28,36,46,60\) for Bat~4; and
\(14,17,21,26,33,41\) for Bat~5. The reported interval is the usual
ordinary-least-squares regression interval for the fitted slope and
is used only as a finite-range regression summary. It is not a
confidence interval for a mifOU Hurst parameter.

\begin{figure}[H]
\centering
\includegraphics[width=\textwidth]{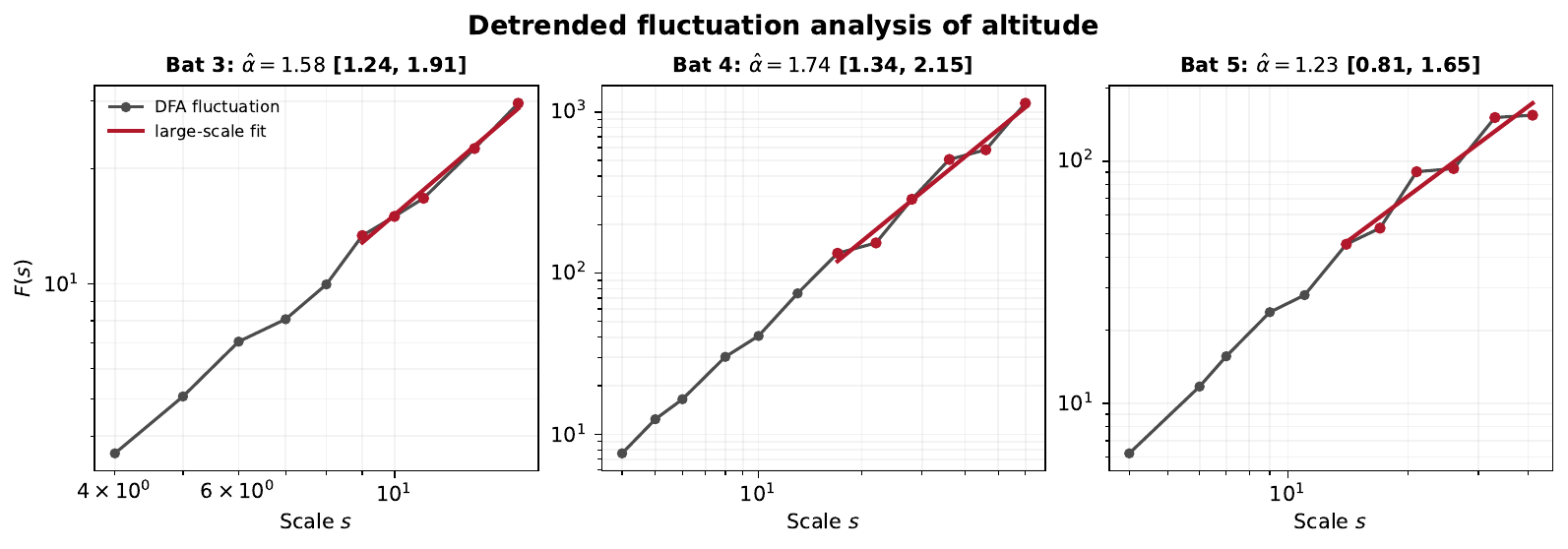}
\caption{Altitude DFA curves for Bats~3--5. Points show the
fluctuation function \(F(s)\) at the available integer scales, and
red segments show the ordinary-least-squares fits over the selected
large-scale ranges. Panel headings report the fitted slope and its
95\% regression interval. DFA scales index successive retained
observations rather than equal intervals of elapsed time.}
\label{fig:altdfa}
\end{figure}

\begin{table}[H]
\centering
\small
\setlength{\tabcolsep}{8pt}
\begin{tabular}{@{}crrr@{}}
\toprule
Bat
& \(\widehat{\alpha}\)
& 95\% regression interval
& \(R^2\)\\
\midrule
3 & 1.578 & [1.241, 1.915] & 0.987\\
4 & 1.742 & [1.337, 2.147] & 0.973\\
5 & 1.229 & [0.811, 1.648] & 0.943\\
\bottomrule
\end{tabular}
\caption{Estimated DFA slopes over the selected large-scale ranges,
with their ordinary-least-squares regression intervals and
coefficients of determination.}
\label{tab:altdfa}
\end{table}

The fitted DFA slopes exceed one for all three Altitude sequences,
although the regression interval for Bat~5 includes one. These
values describe scaling over the selected finite ranges and are
interpreted together with the rolling summaries, autocorrelation
functions, and stationarity diagnostics. They are not used as
estimates of the mifOU Hurst parameters or as a stand-alone
classification of asymptotic memory.

\subsubsection{Likelihood comparison of marginal Altitude models}
\label{subsub:altitude-model-comparison}

Each Altitude series is fitted with the same seven Gaussian
covariance families used for Longitude and Latitude by
\citet{RamirezEtAl2026}. This comparison is strictly univariate: no
cross-coordinate dependence parameter is estimated. Let \(A_b(t)\)
denote Altitude in meters and define
\[
Y_{b,i}
=
\frac{A_b(t_i)-A_b(0)}{1000},
\qquad t_i>0.
\]
Thus, the response used in every likelihood calculation is expressed
in kilometers. Within each bat, every candidate is evaluated on the
same response vector and observation times under the same zero-mean
convention. For candidate \(m\), with covariance matrix
\(\boldsymbol\Sigma_{b,m}(\theta_m)\) and \(k_m\) estimated
parameters, the Akaike information criterion \citep{Akaike1974} is
\[
\operatorname{AIC}_{b,m}
=
-2\ell_{b,m}(\widehat\theta_m)+2k_m,
\qquad
\widehat\theta_m
=
\arg\max_{\theta_m}\ell_{b,m}(\theta_m),
\]
where \(\ell_{b,m}\) is the Gaussian log-likelihood. AIC values are
compared only among candidates fitted to the same bat. For a
stationary candidate, subtracting the observed initial Altitude is
deterministic centering; its covariance is therefore evaluated as
\(R(t_i-t_j)\), rather than as the covariance of a stochastic
difference from \(A_b(0)\).

Table~\ref{tab:altmodelspec} defines the seven covariance families
and records the parameters counted in their AIC values.

{\scriptsize
\setlength{\tabcolsep}{3pt}
\begin{longtable}{
p{.11\textwidth}
p{.32\textwidth}
p{.18\textwidth}
p{.25\textwidth}}
\caption{Candidate families for the univariate Altitude comparison.
\(U\) denotes Tricomi's confluent hypergeometric function. For
\(\zeta\), the positive-\(\beta\) and \(\beta=0\) specifications are
evaluated separately, and the retained specification is counted with
its corresponding number of estimated parameters. The \(\zeta\),
ifOU, and iOU scale parameters are profiled when possible.}
\label{tab:altmodelspec}\\
\toprule
Model
& Covariance or construction
& Estimated
& Fixed/interpretation\\
\midrule
\(\zeta\)
&
\(2\sigma^2\!\int_0^{s\wedge t}\!
e^{-\beta u}
[g(s+t-2u)-g(s-u)-g(t-u)]\,du\),
\(g(x)=x\log x\)
&
\(\beta,\sigma^2\) if \(\beta>0\);
\(\sigma^2\) if \(\beta=0\)
&
nonstationary, not intrinsically stationary;
the boundary specification fixes \(\beta=0\)
\\
ifOU \(\mu_H\)
&
position obtained by integrating an fOU velocity
&
\(\beta,\sigma,H\)
&
nonstationary and not intrinsically stationary;
covariance grows as \(T^{2H-1}\) for \(H>1/2\) and approaches a
finite limit for \(H\leq1/2\)
\\
iOU \(\mu_{1/2}\)
&
same position model with Brownian driving noise
&
\(\beta,\sigma\)
&
\(H=1/2\)
\\
\(\sigma W_H\)
&
\(\frac{\sigma^2}{2}
(s^{2H}+t^{2H}-|t-s|^{2H})\)
&
\(H,\sigma\)
&
nonstationary with stationary increments
\\
Confluent \(A\)
&
\(\sigma^2\Gamma(\eta+\alpha)/\Gamma(\eta)\)
\newline
\(\times
U\!\left(
\alpha,1-\eta,
\eta\{|t-s|/\beta\}^2
\right)\)
&
\(\sigma^2,\eta,\alpha,\beta\)
&
stationary; local and tail parameters separated
\\
Stein \(S^{(2)}\)
&
\(\sigma^2(1+\lambda|t-s|)^{-1}\)
&
\(\sigma^2,\lambda\)
&
stationary; exponent fixed at \(2/2\)
\\
Stein \(S^{(3)}\)
&
\(\sigma^2(1+\lambda|t-s|)^{-3/2}\)
&
\(\sigma^2,\lambda\)
&
stationary; exponent fixed at \(3/2\)
\\
\bottomrule
\end{longtable}}

For \(\zeta\), the interior specification \(\beta>0\) and the
boundary specification \(\beta=0\) are evaluated separately. The
interior fit has two estimated parameters, \((\beta,\sigma^2)\), and
its AIC uses \(k=2\). At the boundary, \(\beta\) is fixed at zero,
only \(\sigma^2\) is estimated, and the AIC uses \(k=1\). The
lower-AIC specification is retained as the \(\zeta\) entry in the
seven-family comparison.

For ifOU and iOU, the search interval is
\(0.001<\beta\leq400\); an estimate of \(400\) denotes the
numerical bound rather than an interior optimum. For the Confluent
model, \(0.001<\alpha\leq100\) and
\(0.001<\beta\leq2000\). 
Table~\ref{tab:altitude-top3-supp} reports the three lowest-AIC
Altitude fits for each individual.

\begin{table}[H]
\centering
\scriptsize
\setlength{\tabcolsep}{4pt}
\renewcommand{\arraystretch}{1.08}
\begin{tabularx}{\textwidth}{
@{}cc
>{\raggedright\arraybackslash}p{.17\textwidth}
>{\raggedright\arraybackslash}X
rrr@{}}
\toprule
Bat & Rank & Model & MLE & AIC & \(\Delta\)AIC & \(w\)\\
\midrule
3 & 1 & fBM
& \((\widehat H,\widehat\sigma)=(0.45014,0.008295)\)
& -411.838 & 0.000 & 0.431\\
3 & 2 & iOU
& \((\widehat\beta,\widehat\sigma)=(400^*,3.3428)\)
& -411.233 & 0.605 & 0.318\\
3 & 3 & ifOU
& \((\widehat\beta,\widehat\sigma,\widehat H)=(400^*,3.3276,0.44914)\)
& -409.818 & 2.020 & 0.157\\
\midrule
4 & 1 & \(\zeta\)
& \((\widehat\beta,\widehat{\sigma}^2)=(0.005518,1.7032\times10^{-4})\)
& -1532.139 & 0.000 & 0.994\\
4 & 2 & ifOU
& \((\widehat\beta,\widehat\sigma,\widehat H)=(0.06778,0.04709,0.06148)\)
& -1521.443 & 10.695 & 0.00473\\
4 & 3 & fBM
& \((\widehat H,\widehat\sigma)=(0.96548,0.04102)\)
& -1518.242 & 13.897 & 0.00096\\
\midrule
5 & 1 & iOU
& \((\widehat\beta,\widehat\sigma)=(1.6401,0.03799)\)
& -1061.670 & 0.000 & 0.509\\
5 & 2 & ifOU
& \((\widehat\beta,\widehat\sigma,\widehat H)=(0.87739,0.03180,0.30353)\)
& -1060.851 & 0.819 & 0.338\\
5 & 3 & \(\zeta\)
& \((\widehat\beta,\widehat{\sigma}^2)=(0.012414,2.0344\times10^{-4})\)
& -1059.092 & 2.577 & 0.140\\
\bottomrule
\end{tabularx}
\caption{Three lowest-AIC univariate Altitude models for Bats~3--5.
Akaike weights are normalized over the seven retained family entries
for the corresponding individual. A star denotes the numerical
bound \(\beta_{\max}=400\).}
\label{tab:altitude-top3-supp}
\end{table}

The minimum-AIC Altitude family differs among individuals. For Bat~3,
fBM and iOU differ by \(0.605\) AIC units, and ifOU is \(2.020\)
units above fBM. For Bat~4, \(\zeta\) is separated from ifOU by
\(10.695\) AIC units. For Bat~5, iOU and ifOU differ by \(0.819\)
AIC units. The selected fBM model for Bat~3 has stationary
short-memory increments, with \(\widehat H=0.4501\). The selected
\(\zeta\) model for Bat~4 is neither stationary nor intrinsically
stationary and has logarithmic covariance growth. The selected iOU
model for Bat~5 corresponds to \(H=1/2\), for which the covariance
between separated increments decays exponentially. Thus, the fitted
marginal Altitude covariance family is not the same for the three
individuals.

\subsection{Movement directions and descriptive coordinate
associations}
\label{sub:movement-associations}

Figures~\ref{fig:angles12},~\ref{fig:angles34}, and~\ref{fig:angles5} summarize the directions of successive metric
displacements. Table~\ref{tab:desccorr} reports Pearson and Spearman
correlations between coordinate-specific displacement rates.

\begin{figure}[H]
\centering
\includegraphics[width=.82\textwidth]{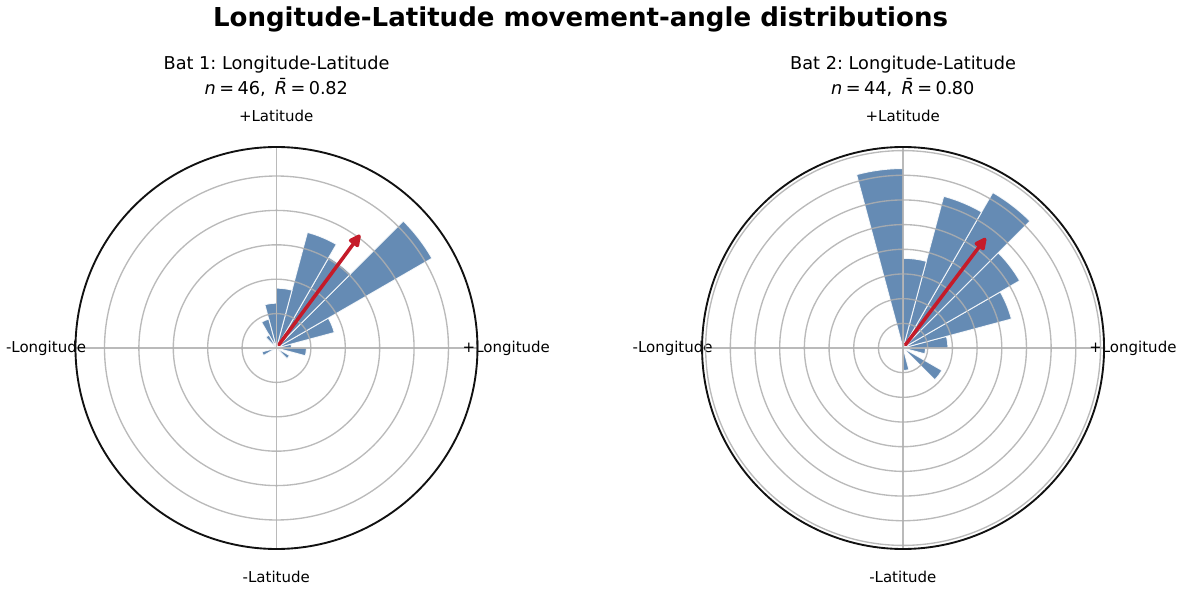}
\caption{Directions of successive Longitude--Latitude displacements
for Bats~1--2 after conversion of both horizontal coordinates to
kilometers. Bars are relative frequencies in 24 angular bins; the
arrow indicates the circular mean direction, and \(\bar R\) denotes
the mean resultant length.}
\label{fig:angles12}
\end{figure}

\begin{figure}[H]
\centering
\includegraphics[width=.95\textwidth]{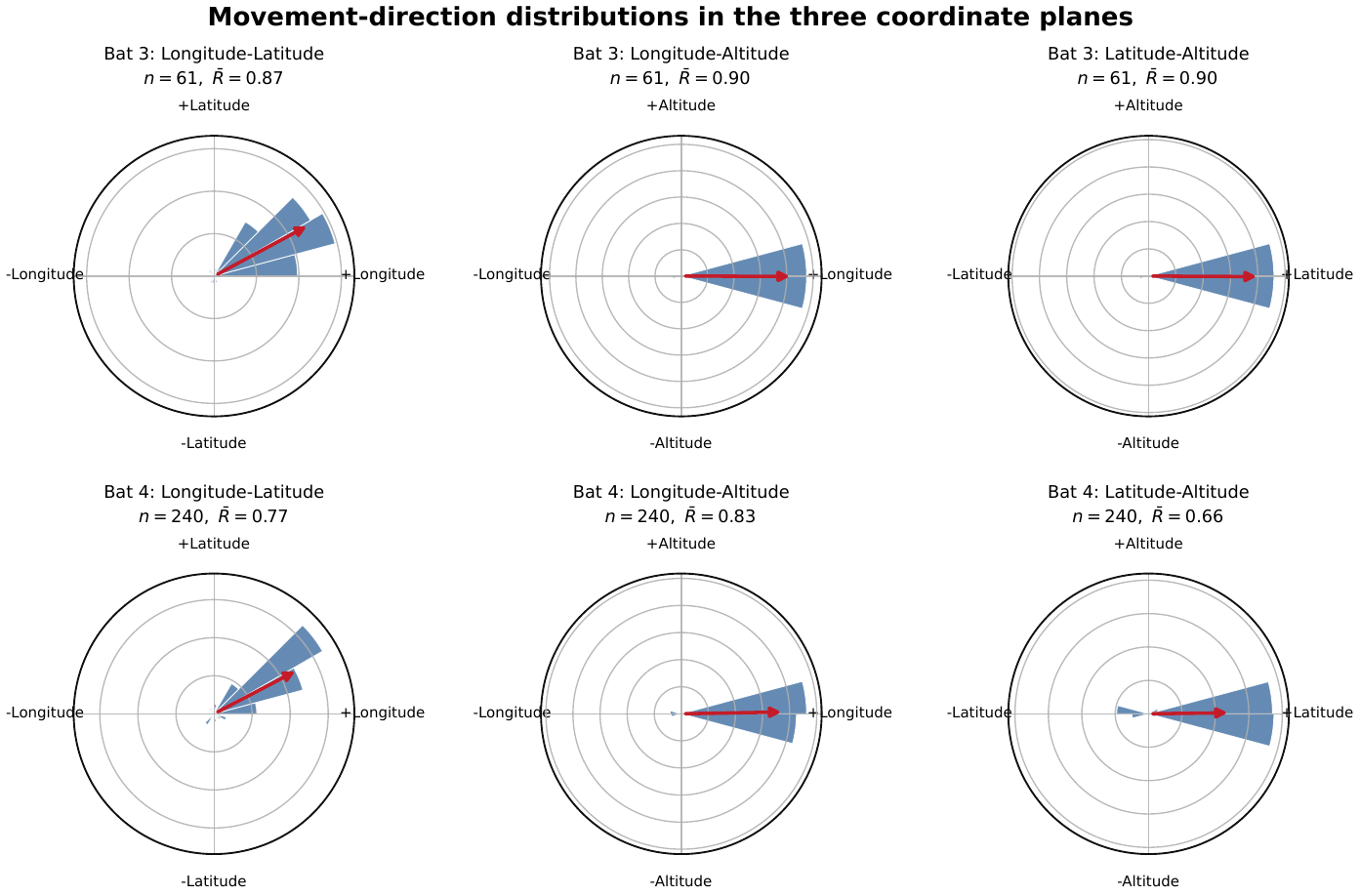}
\caption{Directions of successive metric displacements for Bats~3--4
in the three coordinate planes. All displacement components are
expressed in kilometers. Bars are relative frequencies in 24 angular
bins; the arrow indicates the circular mean direction, and
\(\bar R\) denotes the mean resultant length.}
\label{fig:angles34}
\end{figure}

\begin{figure}[H]
\centering
\includegraphics[width=.90\textwidth]{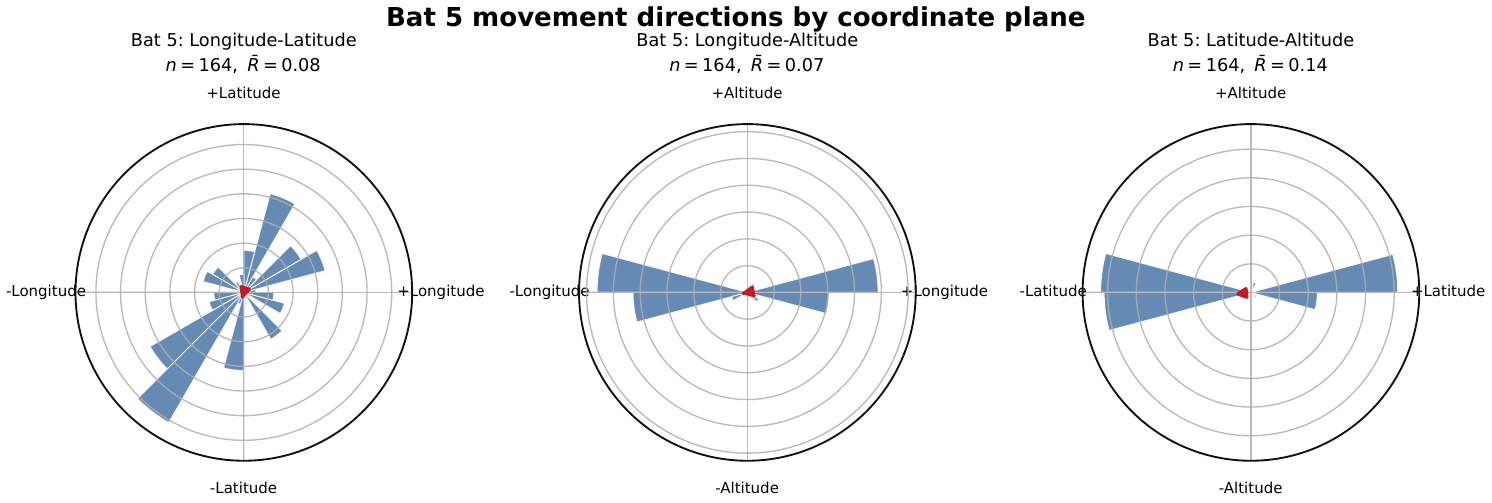}
\caption{Directions of successive metric displacements for Bat~5 in
the three coordinate planes. All displacement components are
expressed in kilometers. Bars are relative frequencies in 24 angular
bins; the arrow indicates the circular mean direction, and
\(\bar R\) denotes the mean resultant length.}
\label{fig:angles5}
\end{figure}

\begin{table}[H]
\centering
\footnotesize
\setlength{\tabcolsep}{4pt}
\begin{tabular}{@{}clrrr@{}}
\toprule
Bat
& Coordinate displacement-rate pair
& \(n_{\mathrm{rate}}\)
& Pearson \(r\) (\(p\))
& Spearman \(\rho_s\) (\(p\))\\
\midrule
1 & Longitude--Latitude
& 46 & -0.284 (0.0557) & -0.201 (0.180)\\
2 & Longitude--Latitude
& 44 & -0.381 (0.0107) & -0.420 (0.0045)\\
3 & Longitude--Latitude
& 61 & 0.580 (\(<10^{-6}\)) & 0.747 (\(<10^{-11}\))\\
3 & Longitude--Altitude
& 61 & -0.061 (0.640) & 0.016 (0.901)\\
3 & Latitude--Altitude
& 61 & 0.051 (0.695) & 0.199 (0.124)\\
4 & Longitude--Latitude
& 240 & 0.471 (\(<10^{-13}\)) & 0.444 (\(<10^{-12}\))\\
4 & Longitude--Altitude
& 240 & -0.063 (0.333) & -0.082 (0.205)\\
4 & Latitude--Altitude
& 240 & 0.092 (0.157) & 0.068 (0.296)\\
5 & Longitude--Latitude
& 164 & 0.317 (\(<10^{-4}\)) & 0.472 (\(<10^{-9}\))\\
5 & Longitude--Altitude
& 164 & 0.047 (0.547) & 0.098 (0.213)\\
5 & Latitude--Altitude
& 164 & 0.146 (0.0615) & 0.186 (0.0172)\\
\bottomrule
\end{tabular}
\caption{Pearson and Spearman correlations between coordinate
displacement rates expressed in kilometers per minute.
\(n_{\mathrm{rate}}\) is the number of successive rates entering
each calculation. Parentheses contain conventional, unadjusted
two-sided \(p\)-values, which account for neither serial dependence
between adjacent rates nor the multiple coordinate pairs examined.
These descriptive correlations are not estimates of the mifOU
cross-coordinate dependence parameters \(\rho_{i,j}\).}
\label{tab:desccorr}
\end{table}

For Bat~1, the Longitude--Latitude coefficients are negative, but
their unadjusted \(p\)-values exceed \(0.05\). For Bat~2, the Pearson
and Spearman coefficients are \(-0.381\) and \(-0.420\),
respectively, with unadjusted \(p\)-values below \(0.05\). For
Bats~3--5, both measures give positive Longitude--Latitude
associations, with the largest Pearson and Spearman coefficients for
Bat~3.

Associations involving Altitude are smaller. Their Pearson
correlations have absolute value at most \(0.146\), and none of the
corresponding Pearson tests has an unadjusted \(p\)-value below
\(0.05\). Among the Spearman correlations involving Altitude, only
the Latitude--Altitude pair for Bat~5 has \(p<0.05\), with
coefficient \(0.186\) and \(p=0.0172\).

In the Longitude--Latitude plane, the mean resultant lengths for
Bats~3--5 are \(0.875\), \(0.765\), and \(0.079\), respectively.
Successive horizontal displacements are therefore strongly
concentrated around a preferred direction for Bats~3--4 but not for
Bat~5. Directional concentration and correlation between
coordinate-specific displacement rates describe different
properties of the observed trajectories and are interpreted
separately.

\subsection{Complete-trajectory model comparison}\label{sub:complete-trajectory-supp}

{\small
The complete-model comparison is built from the seven univariate families in Table~1 of \citet{RamirezEtAl2026}. All independent products are retained, and the horizontal reference AIC values are used exactly as reported in that table. For Bats~3--5, the corresponding Altitude reference AIC values are those of Section~\ref{subsub:altitude-model-comparison} of this appendix. The selected coordinatewise families are fBM/fBM for Bats~1--2, \(\zeta\)/ifOU for Bat~3, fBM/\(\zeta\) for Bat~4, and \(\zeta\)/Confluent for Bat~5; the selected Altitude families for Bats~3--5 are fBM, \(\zeta\), and iOU, respectively.

For Bats~1--2, the \(7^2=49\) independent products are compared with correlated iOU and ifOU Longitude--Latitude models, giving 51 candidates. For Bats~3--5, the comparison contains all \(7^3=343\) independent products, all \(3\times2\times7=42\) models formed by one correlated iOU or ifOU coordinate pair and one independent singleton family, and the full three-dimensional iOU and ifOU models. This gives 387 candidates per three-dimensional trajectory. The other five univariate families enter a correlated candidate only as the independent singleton; no cross-covariance involving those families is imposed.

All marginal and dependence parameters are estimated jointly within each correlated model. Let \(C\) denote a correlated candidate and let \(I(C)\) denote the independently fitted product model with the same coordinate-specific marginal families as \(C\). At half-minute resolution, the coordinatewise reference AIC values define the marginal scale of the complete ranking. The contribution of cross-coordinate dependence is measured by the AIC difference between \(C\) and \(I(C)\), calculated with the common multivariate likelihood implementation. The calibrated score is
\[
\mathrm{AIC}_{\mathrm{cal}}(C)
=\sum_j\mathrm{AIC}^{\mathrm{ref}}_j(M_j)
+\left\{\mathrm{AIC}_{\mathrm g}(C)
-\mathrm{AIC}_{\mathrm g}\!\left(I(C)\right)\right\},
\]
where the sum is over the coordinates observed for that individual. The term in braces compares two separate maximum-likelihood fits under the same marginal-family combination, response vector, observation times, and likelihood implementation. It therefore does not hold the marginal parameters at the estimates from \(I(C)\). Absolute values of \(\mathrm{AIC}_{\mathrm g}\) are used only within a fixed marginal-family combination. Parameters reported for \(C\) are specific to that model and include every coordinate-specific \(\beta\), \(\sigma\), and \(H\), together with the active correlations.

The reference AIC values retain the parameter counts used in their respective univariate comparisons. In the common multivariate fits entering the term in braces, \(C\) and \(I(C)\) use the same nominal marginal parametrizations, so the shared marginal counts cancel from their AIC difference. Complete-model parameter counts include the nominal marginal dimensions and the active correlations; in these multivariate fits, \(\zeta\) is represented by \((\beta,\sigma^2)\), including when its optimum occurs at \(\widehat\beta=0\). The ifOU, iOU, and multivariate OU-position fits use \(0.001<\beta\leq400\). Newly fitted Confluent Altitude components and the Confluent fits at one- and two-minute resolution use \(0.001<\alpha\leq100\) and \(0.001<\beta\leq2000\); the half-minute horizontal Confluent reference fits are taken unchanged from Table~1 of \citet{RamirezEtAl2026} and are not reoptimized under these new bounds. For each independent product, \(\mathrm{AIC}_{\mathrm{cal}}\) is the sum of its coordinatewise reference AIC values. Akaike weights are normalized separately within each bat over the 51 candidates for Bats~1--2 and the 387 candidates for Bats~3--5. Table~\ref{tab:topthree-main} of the main text reports the complete model-specific estimates for the three leading half-minute models.
}

\subsubsection{Fractional temporal structure}

The complete ranking combines two distinct questions. The first is whether the temporal flexibility provided by estimated Hurst parameters is supported after the dependence structure has been optimized. The second is whether cross-coordinate correlation improves the fit after the Hurst parameters have been profiled. The restricted comparisons in this and the following subsection separate these questions.

For the first comparison, the likelihood is optimized over the admissible dependence structures within two pure OU-position families. In the mifOU family, every observed coordinate is modeled by ifOU and its Hurst parameter is estimated. In the miOU family, every Hurst parameter is fixed at \(H_i=1/2\). For Bats~1--2, the response is two-dimensional and the admissible structures are independence and a correlated Longitude--Latitude block. For Bats~3--5, the response is three-dimensional and the search includes independence, each correlated two-coordinate block with the remaining coordinate independent, and full three-dimensional dependence.

All candidates in this restricted comparison use the same response and the same multivariate likelihood implementation, so their direct global-fit AIC values are comparable without the calibration used to connect the complete ranking to the published coordinatewise AIC values. For a model with \(k\) free parameters, AICc is calculated as
\[
\mathrm{AICc}
=
\mathrm{AIC}
+
\frac{2k(k+1)}{m-k-1},
\qquad m=n-1,
\]
where \(m\) is the number of noninitial observation times. Each model is maximized separately over its complete parameter vector; in particular, a full three-dimensional fit is not evaluated at estimates inherited from an independent or reduced fit. Because the specification minimizing AIC within a temporal family can differ from the one minimizing AICc, Table~\ref{tab:familyaic} reports the criterion-specific best mifOU and miOU specifications.

\begin{landscape}
\begin{table}[p]
\centering\scriptsize
\setlength{\tabcolsep}{3.6pt}
\renewcommand{\arraystretch}{1.06}
\begin{tabularx}{\linewidth}{@{}cc>{\raggedright\arraybackslash}Xr>{\raggedright\arraybackslash}Xrr@{}}
\toprule
Bat & Response & Best mifOU specification & Score & Best miOU specification (\(H_i=1/2\)) & Score & Difference\\
\midrule
\multicolumn{7}{l}{\textit{Panel A: AIC}}\\
1 & Lon--Lat & ifOU(Lon)\(\oplus\)ifOU(Lat) & -527.999 & iOU(Lon,Lat) & -491.153 & 36.845\\
2 & Lon--Lat & ifOU(Lon)\(\oplus\)ifOU(Lat) & -551.618 & iOU(Lon,Lat) & -534.223 & 17.395\\
3 & Lon--Lat--Alt & ifOU(Lon,Lat,Alt) & -1584.285 & iOU(Lon,Lat,Alt) & -1565.470 & 18.815\\
4 & Lon--Lat--Alt & ifOU(Lon,Lat)\(\oplus\)ifOU(Alt) & -6407.044 & iOU(Lon,Lat)\(\oplus\)iOU(Alt) & -6286.792 & 120.252\\
5 & Lon--Lat--Alt & ifOU(Lon,Lat)\(\oplus\)ifOU(Alt) & -4296.970 & iOU(Lon)\(\oplus\)iOU(Lat)\(\oplus\)iOU(Alt) & -4188.101 & 108.869\\
\midrule
\multicolumn{7}{l}{\textit{Panel B: AICc with \(m=n-1\)}}\\
1 & Lon--Lat & ifOU(Lon)\(\oplus\)ifOU(Lat) & -525.845 & iOU(Lon,Lat) & -489.653 & 36.192\\
2 & Lon--Lat & ifOU(Lon)\(\oplus\)ifOU(Lat) & -549.348 & iOU(Lon,Lat) & -532.644 & 16.704\\
3 & Lon--Lat--Alt & ifOU(Lon,Lat)\(\oplus\)ifOU(Alt) & -1579.360 & iOU(Lon,Lat,Alt) & -1561.941 & 17.419\\
4 & Lon--Lat--Alt & ifOU(Lon,Lat)\(\oplus\)ifOU(Alt) & -6406.083 & iOU(Lon,Lat)\(\oplus\)iOU(Alt) & -6286.309 & 119.774\\
5 & Lon--Lat--Alt & ifOU(Lon,Lat)\(\oplus\)ifOU(Alt) & -4295.532 & iOU(Lon)\(\oplus\)iOU(Lat)\(\oplus\)iOU(Alt) & -4187.566 & 107.966\\
\bottomrule
\end{tabularx}
\caption{Criterion-specific comparison of the best pure mifOU and miOU specifications. The dependence structure is optimized separately within each temporal family. ``Difference'' is the miOU score minus the mifOU score, so positive values favor estimated Hurst parameters. Every entry is obtained from a separate global maximum-likelihood fit.}
\label{tab:familyaic}
\end{table}
\end{landscape}

Both criteria favor mifOU over the \(H_i=1/2\) miOU restriction for every trajectory. The smallest AICc difference is \(16.704\) units for Bat~2, and the differences exceed \(100\) units for Bats~4--5. Dependence is optimized independently within both temporal families.

For Bat~3, AIC selects the full three-dimensional ifOU model, whereas AICc selects a correlated ifOU Longitude--Latitude block with independent ifOU Altitude. Both criteria still favor estimated Hurst parameters over the miOU restriction; the information criteria alone do not determine the corresponding memory classification.

\subsubsection{Cross-coordinate dependence after profiling the Hurst parameters}

The second comparison asks whether correlation improves fit once temporal flexibility has been admitted. Here iOU is included as the \(H=1/2\) special case of the ifOU family rather than as a separate temporal family. Within each dependence class, an independent component may therefore select either an estimated \(H\) or the fixed value \(H=1/2\). For Bats~1--2, Table~\ref{tab:3daic} compares independent coordinates with a correlated Longitude--Latitude block. For Bats~3--5, it compares all coordinates independent, each of the three separately maximized two-dimensional blocks with the remaining coordinate independent, and a separately maximized fully correlated three-dimensional model.

The specification shown for each class is the one with the smallest ordinary AIC after a separate optimization over all parameters in that class. The same specification also minimizes ordinary AICc within every class in this data set. Thus, the comparison does not hold the marginal estimates fixed while adding a correlation parameter, and the full-model estimates are not restricted to the reduced-model estimates.

For Bats~1--2, adding a Longitude--Latitude correlation does not offset the extra dependence parameter. The AIC differences are \(1.151\) and \(0.436\), and the corresponding AICc differences are \(1.944\) and \(1.277\).
\begin{landscape}
\begin{table}[p]
\centering\scriptsize
\setlength{\tabcolsep}{3.2pt}
\renewcommand{\arraystretch}{1.04}
\begin{tabularx}{\linewidth}{@{}cc>{\raggedright\arraybackslash}p{.255\linewidth}crrrr>{\raggedright\arraybackslash}X@{}}
\toprule
Bat & Dependence class & Selected OU-position specification & \(k\) & AIC & \(\Delta\)AIC & AICc & \(\Delta\)AICc & Active correlations\\
\midrule
1 & independent & ifOU(Lon)\(\oplus\)ifOU(Lat) & 6 & \textbf{-527.999} & 0.000 & \textbf{-525.845} & 0.000 & zero (independent)\\
 & Lon--Lat correlated & ifOU(Lon,Lat) & 7 & -526.848 & 1.151 & -523.901 & 1.944 & \(\widehat\rho_{\rm Lon,Lat}=-0.147\)\\
\cmidrule(lr){1-9}
2 & independent & ifOU(Lon)\(\oplus\)ifOU(Lat) & 6 & \textbf{-551.618} & 0.000 & \textbf{-549.348} & 0.000 & zero (independent)\\
 & Lon--Lat correlated & ifOU(Lon,Lat) & 7 & -551.182 & 0.436 & -548.071 & 1.277 & \(\widehat\rho_{\rm Lon,Lat}=-0.211\)\\
\cmidrule(lr){1-9}
3 & independent & ifOU(Lon)\(\oplus\)ifOU(Lat)\(\oplus\)iOU(Alt) & 8 & -1557.730 & 27.445 & -1554.961 & 26.684 & zero (independent)\\
 & Lon--Lat correlated & ifOU(Lon,Lat)\(\oplus\)iOU(Alt) & 9 & \textbf{-1585.175} & 0.000 & \textbf{-1581.646} & 0.000 & \(\widehat\rho_{\rm Lon,Lat}=0.639\)\\
 & Lon--Alt correlated & ifOU(Lon,Alt)\(\oplus\)ifOU(Lat) & 10 & -1555.515 & 29.660 & -1551.115 & 30.531 & \(\widehat\rho_{\rm Lon,Alt}=0.110\)\\
 & Lat--Alt correlated & iOU(Lat,Alt)\(\oplus\)ifOU(Lon) & 8 & -1561.443 & 23.732 & -1558.674 & 22.972 & \(\widehat\rho_{\rm Lat,Alt}=0.459\)\\
 & full 3D & ifOU(Lon,Lat,Alt) & 12 & -1584.285 & 0.890 & -1577.785 & 3.860 & \((0.631,0.183,0.338)\) in Lon--Lat, Lon--Alt, Lat--Alt order\\
\cmidrule(lr){1-9}
4 & independent & ifOU(Lon)\(\oplus\)ifOU(Lat)\(\oplus\)ifOU(Alt) & 9 & -6315.345 & 91.699 & -6314.562 & 91.521 & zero (independent)\\
 & Lon--Lat correlated & ifOU(Lon,Lat)\(\oplus\)ifOU(Alt) & 10 & \textbf{-6407.044} & 0.000 & \textbf{-6406.083} & 0.000 & \(\widehat\rho_{\rm Lon,Lat}=0.579\)\\
 & Lon--Alt correlated & ifOU(Lon,Alt)\(\oplus\)ifOU(Lat) & 10 & -6313.855 & 93.189 & -6312.894 & 93.189 & \(\widehat\rho_{\rm Lon,Alt}=0.026\)\\
 & Lat--Alt correlated & ifOU(Lat,Alt)\(\oplus\)ifOU(Lon) & 10 & -6315.193 & 91.850 & -6314.233 & 91.850 & \(\widehat\rho_{\rm Lat,Alt}=0.040\)\\
 & full 3D & ifOU(Lon,Lat,Alt) & 12 & -6404.460 & 2.583 & -6403.086 & 2.997 & \((0.577,0.026,0.042)\) in Lon--Lat, Lon--Alt, Lat--Alt order\\
\cmidrule(lr){1-9}
5 & independent & ifOU(Lon)\(\oplus\)ifOU(Lat)\(\oplus\)iOU(Alt) & 8 & -4292.563 & 5.226 & -4291.634 & 4.986 & zero (independent)\\
 & Lon--Lat correlated & ifOU(Lon,Lat)\(\oplus\)iOU(Alt) & 9 & \textbf{-4297.789} & 0.000 & \textbf{-4296.620} & 0.000 & \(\widehat\rho_{\rm Lon,Lat}=-0.194\)\\
 & Lon--Alt correlated & ifOU(Lon,Alt)\(\oplus\)ifOU(Lat) & 10 & -4289.739 & 8.050 & -4288.301 & 8.319 & \(\widehat\rho_{\rm Lon,Alt}=-0.001\)\\
 & Lat--Alt correlated & ifOU(Lat,Alt)\(\oplus\)ifOU(Lon) & 10 & -4289.785 & 8.004 & -4288.348 & 8.273 & \(\widehat\rho_{\rm Lat,Alt}=0.014\)\\
 & full 3D & ifOU(Lon,Lat,Alt) & 12 & -4292.965 & 4.825 & -4290.898 & 5.722 & \((-0.195,-0.001,0.011)\) in Lon--Lat, Lon--Alt, Lat--Alt order\\
\bottomrule
\end{tabularx}
\caption{Dependence-structure comparison within the OU-position family after profiling the Hurst parameters, with iOU included as the \(H=1/2\) special case. AIC and AICc are ordinary, uncalibrated scores from the common multivariate likelihood. Each row is a separate global fit and reports the lowest-AIC temporal specification within that dependence class. For every two-dimensional block, the remaining observed coordinate is retained as an independent component.}
\label{tab:3daic}
\end{table}
\end{landscape}

For Bat~3, the Longitude--Latitude block improves ordinary AIC by \(27.445\) and ordinary AICc by \(26.684\) relative to independence. The Longitude--Altitude and Latitude--Altitude blocks have \(\Delta\mathrm{AIC}=29.660\) and \(23.732\), respectively. The full model has \(\Delta\mathrm{AIC}=0.890\) and \(\Delta\mathrm{AICc}=3.860\) relative to the reduced model, with correlations \((0.631,0.183,0.338)\).

Bat~4 provides the clearest evidence for correlation. The Longitude--Latitude block improves AIC by \(91.699\) and AICc by \(91.521\) over independence. The Longitude--Altitude and Latitude--Altitude blocks are \(93.189\) and \(91.850\) AIC units above it. The full model is \(2.583\) AIC and \(2.997\) AICc units above it, and its correlations involving Altitude are \(0.026\) and \(0.042\).

For Bat~5, the restricted Longitude--Latitude model has the smallest AIC and AICc. The independent, Longitude--Altitude, Latitude--Altitude, and full models lie \(5.226\), \(8.050\), \(8.004\), and \(4.825\) AIC units above it. This restricted conclusion is weaker than for Bats~3--4, and the complete seven-family ranking selects an independent marginal product for Bat~5.

Together, Tables~\ref{tab:familyaic} and~\ref{tab:3daic} establish the two conclusions separately. Estimated Hurst parameters improve penalized fit relative to the miOU restriction for all five trajectories after optimizing dependence. Within this restricted OU-position comparison, cross-coordinate correlation is not selected for Bats~1--2 but is supported primarily in the Longitude--Latitude pair for Bats~3--5. In the complete seven-family ranking, the more flexible independent marginal combination remains selected for Bat~5.

\subsubsection{Components of the half-minute calibrated AIC}

Table~\ref{tab:aic-components} gives the numerical components of the calibrated AIC for the three leading half-minute models. For a correlated candidate, the global-refit increment is the difference between separate maximum-likelihood fits of that model and its independent counterpart with the same coordinate-specific families. This difference is added to the corresponding reference marginal AIC sum. A negative increment means that the correlated candidate improves the common-fit AIC relative to its independently fitted counterpart.

\begin{landscape}
\begin{table}[p]
\centering\tiny
\setlength{\tabcolsep}{3pt}
\renewcommand{\arraystretch}{1.03}
\begin{tabularx}{\linewidth}{@{}cc>{\raggedright\arraybackslash}X>{\raggedright\arraybackslash}p{.12\linewidth}rrr@{}}
\toprule
Bat & Rank & Model \(C\) & Class & Marginal AIC sum & Global-refit increment & \(\mathrm{AIC}_{\mathrm{cal}}(C)\)\\
\midrule
1 & 1 & fBM(Lon)\(\oplus\)fBM(Lat) & independent & -532.067 & 0.000 & -532.067\\
 & 2 & ifOU(Lon)\(\oplus\)fBM(Lat) & independent & -530.042 & 0.000 & -530.042\\
 & 3 & fBM(Lon)\(\oplus\)ifOU(Lat) & independent & -530.024 & 0.000 & -530.024\\
\addlinespace
2 & 1 & fBM(Lon)\(\oplus\)fBM(Lat) & independent & -555.663 & 0.000 & -555.663\\
 & 2 & fBM(Lon)\(\oplus\)ifOU(Lat) & independent & -553.641 & 0.000 & -553.641\\
 & 3 & ifOU(Lon)\(\oplus\)fBM(Lat) & independent & -553.640 & 0.000 & -553.640\\
\addlinespace
3 & 1 & ifOU(Lon,Lat)\(\oplus\)fBM(Alt) & Lon--Lat block & -1558.336 & -27.445 & -1585.780\\
 & 2 & ifOU(Lon,Lat)\(\oplus\)iOU(Alt) & Lon--Lat block & -1557.730 & -27.445 & -1585.175\\
 & 3 & ifOU(Lon,Lat,Alt) & full 3D & -1556.316 & -27.970 & -1584.285\\
\addlinespace
4 & 1 & ifOU(Lon,Lat)\(\oplus\)\(\zeta\)(Alt) & Lon--Lat block & -6326.040 & -91.699 & -6417.739\\
 & 2 & ifOU(Lon,Lat)\(\oplus\)ifOU(Alt) & Lon--Lat block & -6315.345 & -91.699 & -6407.044\\
 & 3 & ifOU(Lon,Lat,Alt) & full 3D & -6315.345 & -89.116 & -6404.460\\
\addlinespace
5 & 1 & \(\zeta\)(Lon)\(\oplus\)Confluent(Lat)\(\oplus\)iOU(Alt) & independent & -4299.775 & 0.000 & -4299.775\\
 & 2 & \(\zeta\)(Lon)\(\oplus\)fBM(Lat)\(\oplus\)iOU(Alt) & independent & -4299.147 & 0.000 & -4299.147\\
 & 3 & \(\zeta\)(Lon)\(\oplus\)Confluent(Lat)\(\oplus\)ifOU(Alt) & independent & -4298.956 & 0.000 & -4298.956\\
\bottomrule
\end{tabularx}
\caption{Components of the half-minute calibrated AIC\@. Independent rows have zero global-refit increment. For a correlated row, the increment is the difference between separate global fits of that model and its independent counterpart with the same marginal families.}
\label{tab:aic-components}
\end{table}
\end{landscape}

The exhaustive ranking selects independent products for Bats~1--2 and~5 and a correlated Longitude--Latitude block for Bats~3--4. For Bat~5, the best correlated candidate ranks 8, \(1.986\) calibrated AIC units above the independent minimum.

For Bat~4, the reduced and full ifOU models have global-refit increments of \(-91.699\) and \(-89.116\), respectively. Their estimate vectors differ because the full model is fitted jointly rather than evaluated at the reduced-model estimates.

\subsubsection{Sensitivity to one- and two-minute aggregation}

The half-minute analysis remains the primary analysis. For sensitivity, observations were aggregated into nonoverlapping one- and two-minute bins, and the complete candidate set was refitted at each resolution, including the independent benchmarks, the correlated two-coordinate blocks with an independent singleton, and the full three-dimensional models. At one minute, the numbers of retained locations for Bats~1--5 are \(47,45,46,144,\) and \(104\); at two minutes they are \(44,39,31,81,\) and \(59\). Every comparison is within bat and resolution. Because all coarser-resolution candidates use the same likelihood implementation, Table~\ref{tab:resolution-top3} reports ordinary AIC values.

\begin{landscape}
\begin{table}[p]
\centering\tiny
\setlength{\tabcolsep}{3pt}
\renewcommand{\arraystretch}{1.00}
\begin{tabularx}{\linewidth}{@{}cccc>{\raggedright\arraybackslash}Xrrrr>{\raggedright\arraybackslash}p{.17\linewidth}@{}}
\toprule
Resolution & Bat & Rank & \(k\) & Candidate & AIC & \(\Delta\)AIC & Weight & Active correlation(s)\\
\midrule
1 min & 1 & 1 & 4 & fBM(Lon)\(\oplus\)fBM(Lat) & -521.805 & 0.000 & 0.331 & none\\
 &  & 2 & 6 & fBM(Lon)\(\oplus\)Confluent(Lat) & -521.324 & 0.481 & 0.260 & none\\
 &  & 3 & 5 & ifOU(Lon)\(\oplus\)fBM(Lat) & -519.799 & 2.007 & 0.121 & none\\
\addlinespace[1pt]
 & 2 & 1 & 4 & fBM(Lon)\(\oplus\)fBM(Lat) & -555.663 & 0.000 & 0.496 & none\\
 &  & 2 & 5 & fBM(Lon)\(\oplus\)ifOU(Lat) & -553.661 & 2.002 & 0.182 & none\\
 &  & 3 & 5 & ifOU(Lon)\(\oplus\)fBM(Lat) & -553.660 & 2.003 & 0.182 & none\\
\addlinespace[1pt]
 & 3 & 1 & 9 & ifOU(Lon,Lat)\(\oplus\)fBM(Alt) & -1062.327 & 0.000 & 0.275 & \(\widehat\rho_{\rm Lon,Lat}=0.673\)\\
 &  & 2 & 9 & ifOU(Lon,Lat)\(\oplus\)iOU(Alt) & -1062.196 & 0.130 & 0.257 & \(\widehat\rho_{\rm Lon,Lat}=0.673\)\\
 &  & 3 & 12 & ifOU(Lon,Lat,Alt) & -1061.088 & 1.238 & 0.148 & \((0.675,0.198,0.305)\) in Lon--Lat, Lon--Alt, Lat--Alt order\\
\addlinespace[1pt]
 & 4 & 1 & 10 & ifOU(Lon,Lat)\(\oplus\)ifOU(Alt) & -3345.461 & 0.000 & 0.312 & \(\widehat\rho_{\rm Lon,Lat}=0.496\)\\
 &  & 2 & 9 & ifOU(Lon,Lat)\(\oplus\)iOU(Alt) & -3345.408 & 0.053 & 0.304 & \(\widehat\rho_{\rm Lon,Lat}=0.496\)\\
 &  & 3 & 9 & ifOU(Lon,Lat)\(\oplus\)\(\zeta\)(Alt) & -3344.763 & 0.698 & 0.220 & \(\widehat\rho_{\rm Lon,Lat}=0.496\)\\
\addlinespace[1pt]
 & 5 & 1 & 6 & \(\zeta\)(Lon)\(\oplus\)\(\zeta\)(Lat)\(\oplus\)iOU(Alt) & -2421.712 & 0.000 & 0.388 & none\\
 &  & 2 & 7 & \(\zeta\)(Lon)\(\oplus\)\(\zeta\)(Lat)\(\oplus\)ifOU(Alt) & -2421.230 & 0.482 & 0.305 & none\\
 &  & 3 & 9 & ifOU(Lon,Alt)\(\oplus\)\(\zeta\)(Lat) & -2417.798 & 3.913 & 0.055 & \(\widehat\rho_{\rm Lon,Alt}=0.155\)\\
\midrule
2 min & 1 & 1 & 6 & fBM(Lon)\(\oplus\)Confluent(Lat) & -495.050 & 0.000 & 0.563 & none\\
 &  & 2 & 7 & ifOU(Lon)\(\oplus\)Confluent(Lat) & -493.050 & 2.000 & 0.207 & none\\
 &  & 3 & 4 & fBM(Lon)\(\oplus\)fBM(Lat) & -491.895 & 3.155 & 0.116 & none\\
\addlinespace[1pt]
 & 2 & 1 & 4 & fBM(Lon)\(\oplus\)fBM(Lat) & -481.557 & 0.000 & 0.473 & none\\
 &  & 2 & 5 & fBM(Lon)\(\oplus\)ifOU(Lat) & -479.556 & 2.002 & 0.174 & none\\
 &  & 3 & 5 & ifOU(Lon)\(\oplus\)fBM(Lat) & -479.542 & 2.015 & 0.173 & none\\
\addlinespace[1pt]
 & 3 & 1 & 9 & ifOU(Lon,Lat)\(\oplus\)fBM(Alt) & -647.661 & 0.000 & 0.298 & \(\widehat\rho_{\rm Lon,Lat}=0.653\)\\
 &  & 2 & 9 & ifOU(Lon,Lat)\(\oplus\)iOU(Alt) & -647.616 & 0.045 & 0.291 & \(\widehat\rho_{\rm Lon,Lat}=0.653\)\\
 &  & 3 & 10 & ifOU(Lon,Lat)\(\oplus\)ifOU(Alt) & -645.660 & 2.001 & 0.109 & \(\widehat\rho_{\rm Lon,Lat}=0.653\)\\
\addlinespace[1pt]
 & 4 & 1 & 9 & ifOU(Lon,Lat)\(\oplus\)iOU(Alt) & -1604.901 & 0.000 & 0.547 & \(\widehat\rho_{\rm Lon,Lat}=0.414\)\\
 &  & 2 & 10 & ifOU(Lon,Lat)\(\oplus\)ifOU(Alt) & -1603.032 & 1.869 & 0.215 & \(\widehat\rho_{\rm Lon,Lat}=0.414\)\\
 &  & 3 & 12 & ifOU(Lon,Lat,Alt) & -1602.291 & 2.610 & 0.148 & \((0.413,0.161,0.211)\) in Lon--Lat, Lon--Alt, Lat--Alt order\\
\addlinespace[1pt]
 & 5 & 1 & 6 & \(\zeta\)(Lon)\(\oplus\)\(\zeta\)(Lat)\(\oplus\)iOU(Alt) & -1155.391 & 0.000 & 0.242 & none\\
 &  & 2 & 7 & \(\zeta\)(Lon)\(\oplus\)\(\zeta\)(Lat)\(\oplus\)ifOU(Alt) & -1154.601 & 0.790 & 0.163 & none\\
 &  & 3 & 9 & ifOU(Lon,Alt)\(\oplus\)\(\zeta\)(Lat) & -1153.631 & 1.760 & 0.101 & \(\widehat\rho_{\rm Lon,Alt}=0.248\)\\
\bottomrule
\end{tabularx}
\caption{Three lowest ordinary AIC values after separate global optimization at one- and two-minute resolutions. Weights are normalized over 51 candidates for Bats~1--2 and 387 candidates for Bats~3--5.}
\label{tab:resolution-top3}
\end{table}
\end{landscape}

Under ordinary AIC, Bats~1--2 retain independent horizontal coordinates at both coarser resolutions. Bat~1 selects fBM for both coordinates at one minute and fBM Longitude with Confluent Latitude at two minutes; Bat~2 retains fBM for both coordinates. Bat~3 retains the correlated Longitude--Latitude block with independent fBM Altitude; its iOU Altitude counterpart is only \(0.130\) and \(0.045\) AIC units higher at one and two minutes. Bat~4 retains the horizontal block while the selected Altitude family changes from \(\zeta\) at half a minute to ifOU at one minute and iOU at two minutes. Bat~5 remains independent, with \(\zeta\) selected for both horizontal coordinates at one and two minutes. The selected dependence pattern is stable under temporal aggregation, although several marginal choices change.

\subsubsection{Penalty sensitivity across resolutions}

For the half-minute comparison, define the implied fit term
\[
-2\ell_{\mathrm{cal}}
=
\mathrm{AIC}_{\mathrm{cal}}-2k
\]
and the derived penalty-sensitivity scores
\[
\mathrm{AICc}_{\mathrm{cal}}(m)
=
\mathrm{AIC}_{\mathrm{cal}}
+
\frac{2k(k+1)}{m-k-1},
\qquad
\mathrm{BIC}_{\mathrm{cal}}(m)
=
-2\ell_{\mathrm{cal}}+k\log m.
\]
These half-minute quantities change only the complexity penalty attached to the calibrated fit term; they are not presented as small-sample criteria newly derived for irregular multivariate Gaussian trajectories. For the direct one- and two-minute fits, ordinary AICc and BIC are computed from the common likelihood implementation. In every case, \(m=n-1\) counts retained transitions once rather than multiplying the sample size by the number of simultaneous coordinates. Table~\ref{tab:criterion-sensitivity-supp} reports the selected model under each penalty.

\begin{landscape}
\begin{table}[p]
\centering\scriptsize
\setlength{\tabcolsep}{4pt}
\renewcommand{\arraystretch}{1.04}
\begin{tabularx}{\linewidth}{@{}cc>{\raggedright\arraybackslash}X>{\raggedright\arraybackslash}X>{\raggedright\arraybackslash}X@{}}
\toprule
Resolution & Bat & AIC winner & AICc winner & BIC winner\\
\midrule
0.5 min & 1 & A & A & A\\
 & 2 & A & A & A\\
 & 3 & B & B & B\\
 & 4 & C & C & C\\
 & 5 & D & D & E\\
\midrule
1 min & 1 & A & A & A\\
 & 2 & A & A & A\\
 & 3 & B & B & B\\
 & 4 & F & G & G\\
 & 5 & H & H & H\\
\midrule
2 min & 1 & I & I & A\\
 & 2 & A & A & A\\
 & 3 & B & B & B\\
 & 4 & G & G & G\\
 & 5 & H & H & H\\
\bottomrule
\end{tabularx}
\caption{Criterion winners by temporal resolution. At half-minute resolution, AIC, AICc, and BIC denote the calibrated penalty-sensitivity scores; at one and two minutes they denote ordinary criteria from the common likelihood implementation. A: fBM(Lon)\(\oplus\)fBM(Lat); B: ifOU(Lon,Lat)\(\oplus\)fBM(Alt); C: ifOU(Lon,Lat)\(\oplus\)\(\zeta\)(Alt); D: \(\zeta\)(Lon)\(\oplus\)Confluent(Lat)\(\oplus\)iOU(Alt); E: \(\zeta\)(Lon)\(\oplus\)fBM(Lat)\(\oplus\)iOU(Alt); F: ifOU(Lon,Lat)\(\oplus\)ifOU(Alt); G: ifOU(Lon,Lat)\(\oplus\)iOU(Alt); H: \(\zeta\)(Lon)\(\oplus\)\(\zeta\)(Lat)\(\oplus\)iOU(Alt); I: fBM(Lon)\(\oplus\)Confluent(Lat).}
\label{tab:criterion-sensitivity-supp}
\end{table}
\end{landscape}

At half-minute resolution, the penalties agree for Bats~1--4, while Bat~5 changes from Confluent to fBM Latitude under BIC\@. At one minute, Bat~4 changes from ifOU Altitude under AIC to iOU Altitude under AICc and BIC; their AIC separation is only \(0.053\). At two minutes, Bat~1 changes from Confluent to fBM Latitude under BIC\@. These changes quantify complexity-penalty sensitivity without changing the selected cross-coordinate dependence pattern.

\end{document}